\documentclass[preprint,12pt]{elsarticle}

\usepackage{amssymb}
\usepackage{amsmath}
\usepackage{amsthm}
\usepackage{graphicx}
\usepackage{tikz}
\usepackage{algorithm}
\usepackage{algpseudocode}
\usepackage{booktabs}
\usepackage[margin=1in]{geometry}
\usepackage[table]{xcolor} 
\usepackage{placeins}
\usepackage{booktabs}      
\usepackage{tabularx}     
\usepackage{caption}      
\usepackage{url}
\usepackage{multirow} 
\usepackage{enumitem}
\usetikzlibrary{arrows.meta, positioning}
\usepackage[most]{tcolorbox}
\usetikzlibrary{calc}
\usepackage{algcompatible} 
\usetikzlibrary{positioning,calc}
\journal{Graphical Models}
\usetikzlibrary{positioning,fit,shapes.geometric}
\newtheorem{theorem}{Theorem}[section]     
\newtheorem{lemma}[theorem]{Lemma}         
\newtheorem{corollary}[theorem]{Corollary} 
 
\newtheorem{proposition}[theorem]{Proposition}
\newtheorem{theoreticalproperty}{Theoretical Property}
\newtheorem{remark}{Remark}
\newtheorem{assumption}{Assumption}

\definecolor{myGreen}{HTML}{33FF00}
\definecolor{myRed}{HTML}{FF3030}
\definecolor{myGrey}{HTML}{AA5555}
\definecolor{myWhite}{HTML}{FFFFFF}
\definecolor{josef}{HTML}{FF3030}
\definecolor{indigo}{RGB}{75, 0, 130}  

\definecolor{maroon}{RGB}{128, 0, 0} 
\definecolor{petr}{RGB}{200, 30, 30}

\begin{document}

\begin{frontmatter}

\title{Structured Bitmap-to-Mesh Triangulation for Geometry-Aware Discretization of Image-Derived Domains}

\author[ouc]{Wei Feng\fnref{fn1}}
\author[ouc]{Haiyong Zheng\corref{cor1}\fnref{fn2}}

\cortext[cor1]{Corresponding author}
\fntext[fn1]{Email: weifeng@stu.ouc.edu.cn}
\fntext[fn2]{Email: zhenghaiyong@ouc.edu.cn}

\address[ouc]{College of Electronic Engineering, Ocean University of China, Qingdao 266404, Shandong, China}

%% Abstract
\begin{abstract}
%% Text of abstract
We introduce a template-driven triangulation framework for embedding discrete boundaries into a regular triangular grid, enabling raster- or segmentation-derived domains to support structure-preserving, numerically stable PDE discretization. Unlike constrained Delaunay triangulation (CDT), which requires global connectivity updates, our method retriangulates only boundary-intersecting triangles, preserving the base mesh and enabling synchronization-free parallel execution. To ensure determinism and scalability, all local intersection patterns are classified under discrete equivalence and triangle symmetry, forming a finite symbolic lookup table mapping each case to a conflict-free retriangulation template. The resulting mesh is provably closed, angle-bounded, and compatible with cotangent-based discretizations and finite element methods. Numerical experiments—including elliptic and parabolic PDEs, signal interpolation, and structural evaluation—demonstrate fewer slivers, more equilateral elements, and greater geometric fidelity near complex boundaries. These properties make the framework well suited for real-time geometric analysis and physically grounded simulation over image-derived domains.

\end{abstract}

%%Graphical abstract
\begin{graphicalabstract}
\centering
\includegraphics[width=\textwidth]{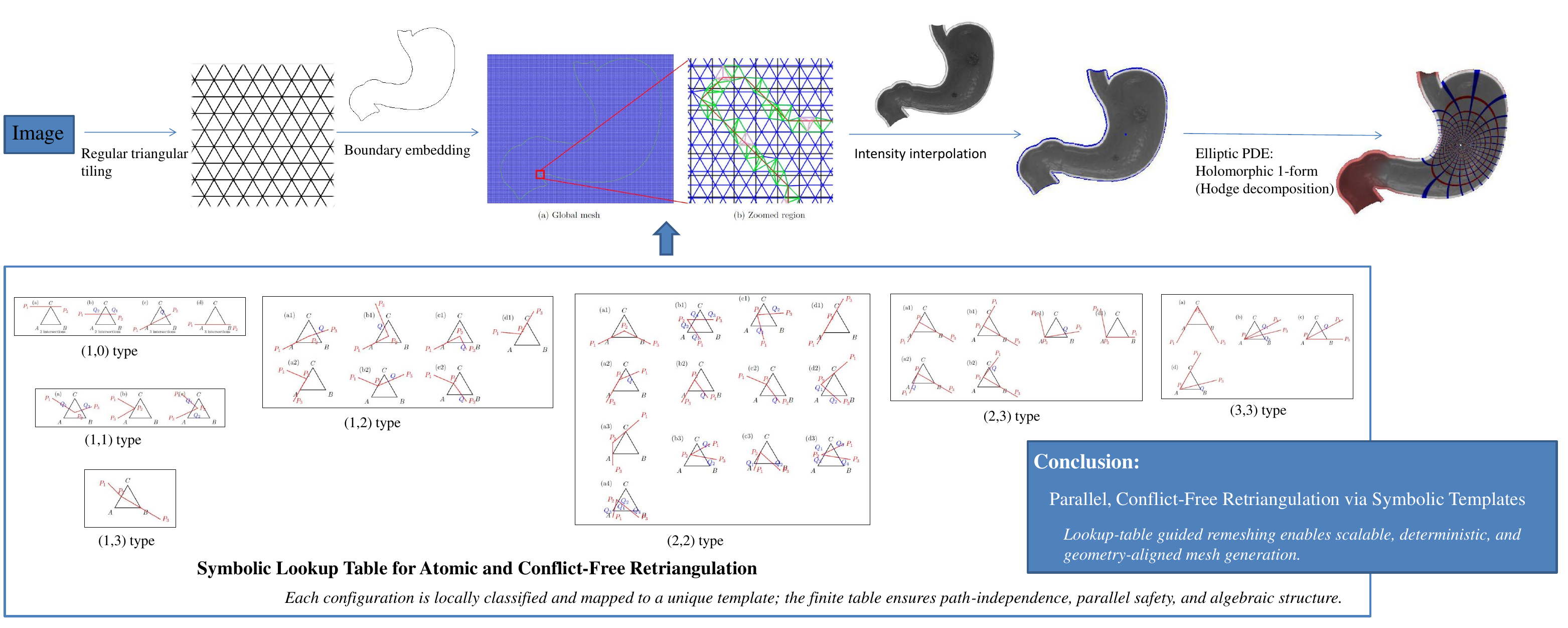}
\end{graphicalabstract}

%%Research highlights

\begin{highlights}
\item A template-driven remeshing method embeds digital boundaries exactly into a regular equilateral grid.
\item All admissible segment--triangle configurations are resolved by a finite, orientation-aware lookup table.
\item The pipeline provides provable local non-degeneracy (angle/area lower bounds) and guarantees topological closure.
\item Stateless local retriangulation enables deterministic, massively parallel execution on raster-derived domains.
\item Experiments against Triangle (CDT) and Gmsh demonstrate sliver-free meshes with highly uniform, near-equilateral interiors and stable PDE behavior.
\end{highlights}

%% Keywords
\begin{keyword}
Structured remeshing, Raster-conforming triangulation, Lookup table methods, Metric consistency, Parallel mesh generation, PDE discretization
\end{keyword}

\end{frontmatter}

%% Add \usepackage{lineno} before \begin{document} and uncomment 
%% following line to enable line numbers
%% \linenumbers

%% main text
%%

%% Use \section commands to start a section
\section{Introduction}

Modern triangulation frameworks often prioritize geometric fidelity or numerical quality~\cite{hoppe1993mesh, ruppert1995delaunay}, yet overlook a more foundational property: \emph{structure}. Most pipelines—whether based on constrained Delaunay triangulation (CDT)~\cite{chew1987constrained, Shewchuk1996Triangle, de1992line}, refinement~\cite{ruppert1995delaunay, Ortiz1991}, or local optimization~\cite{hoppe1993mesh, nealen2006laplacian}—lack composable local rules, symbolic closure, or deterministic retriangulation behavior. These deficiencies lead to global coupling, history dependence, and synchronization barriers that hinder theoretical analysis and scalable numerical simulation~\cite{chazelle1991triangulating, bern1995mesh,idelsohn2006mesh}.

These challenges are amplified in raster-based domains, which are widely used in medical modeling, GIS, and PDE-based simulations~\cite{gonzales1987digital, li2004digital, bankman2008handbook, leveque2007finite}. While raster data is storage-efficient and hardware-friendly, it lacks explicit geometry—limiting mesh quality, introducing boundary artifacts, and obstructing high-fidelity numerical solvers~\cite{shewchuk2002good, goodchild1992geographical}. A reliable raster-to-mesh conversion method is thus critical to bridge discrete representations with structure-sensitive PDE discretizations~\cite{lawson1977software}.

However, existing triangulation methods face structural and numerical limitations. CDT-based algorithms require global topological updates~\cite{ruppert1995delaunay}; grid-based schemes suffer from angular degeneracy and mesh anisotropy~\cite{meyer2003discrete}; and parallel methods~\cite{chernikov2006parallel, DDM, kadow2004parallel, KOHOUT2005491} often incur refinement complexity and inter-thread coupling due to their non-symbolic global logic.

We argue that what is missing is not another geometric heuristic, but a \emph{structure-preserving meshing framework}—one that guarantees local consistency, symbolic determinism, and compatibility with scalable numerical solvers.

\emph{Structured Bitmap-to-Mesh Triangulation (SBMT)} is our response. Rather than globally reconstructing mesh connectivity, SBMT applies symbolic, rule-driven retriangulation to only those triangles intersected by domain boundaries. Each intersection is resolved via a finite, conflict-free template from a symbolic lookup table, ensuring reproducibility and compositional closure.

SBMT is founded on three structural principles:

\begin{itemize}
  \item \textbf{Algebraic structure:} Segment–triangle intersections are classified by discrete equivalence and triangle symmetry (group \( C_3 \)), forming a finite symbolic system with free-group composition. This guarantees determinism, composability, and local closure.

  \item \textbf{Geometric structure:} The mesh is derived from a regular triangular grid with provable angle and area bounds. The near-equilateral quality supports stable interpolation and accurate numerical stencils for PDE solvers~\cite{shewchuk2002good}.

  \item \textbf{Computational structure:} All retriangulation operations are local and stateless, ensuring O(1) updates per triangle and eliminating the need for global topology updates. This localized approach minimizes computational overhead, supports efficient parallelization, and is well-suited for real-time applications and large-scale simulations.
\end{itemize}

Together, these principles yield a template-driven and reproducible meshing system tailored for numerical simulation on image-derived or pixel-aligned domains. Unlike global meshing heuristics, SBMT enforces local structure and predictability—supporting both theoretical correctness and practical scalability.

\paragraph{Main contributions} This work introduces:

\begin{enumerate}
  \item A template-driven meshing framework—\emph{Structured Bitmap-to-Mesh Triangulation \\(SBMT)}—that embeds raster-derived boundaries into a regular triangular scaffold using a finite-state template system. The framework unifies boundary conformity, structural regularity, and computational scalability.

  \item A rigorous classification of all segment–triangle intersection types, with symmetry-based encoding and template composition. The resulting system guarantees local determinism, bounded triangle quality, and stateless retriangulation.

  \item A prototype implementation with quantitative comparisons against Shewchuk’s Triangle~\cite{Shewchuk1996Triangle} and Gmsh. Across raster-derived domains, SBMT yields essentially sliver-free meshes with provable local quality bounds and a highly regular near-equilateral interior, often at markedly lower triangle counts than these CDT baselines, which typically trade larger minimum angles for heavier boundary refinement and higher element counts. We further validate SBMT as a PDE-ready discretization by demonstrating stable elliptic/parabolic solves, including heat diffusion and holomorphic 1-form reconstruction via Poisson equations.

\end{enumerate}

As summarized in Figure~\ref{fig:sbmt-flowchart}, SBMT proceeds in five deterministic stages, from a regular equilateral background grid to a boundary-conforming triangulation. The geometric preprocessing and lookup-based retriangulation boxes are
made precise in Algorithms~\ref{alg:sbmt-preprocess} and
\ref{alg:sbmt-template}, which operate per triangle cell using the
finite template table in Appendix~G.

% Define a style for flowchart nodes
\tikzset{
  flowchart node/.style={
    draw,
    thick,
    minimum height=1.3em,
    minimum width=2.5cm,
    align=center,
    font=\small,
    rounded corners=2pt,
    fill=gray!5
  },
  flowchart arrow/.style={
    thick,
    -{Latex[length=2mm]}
  }
}

\begin{figure}[htbp]
\centering
\begin{tikzpicture}[node distance=1.8cm and 0.7cm, every node/.style={flowchart node}]
\node (init) {Regular Triangular \\ Grid Initialization};
\node (embed) [right=of init] {Boundary \\ Embedding};
\node (preprocess) [right=of embed] {Geometric \\ Preprocessing \\ (thresholds $a,b,c$)};
\node (classify) [right=of preprocess] {Intersection \\ Classification};
\node (lookup) [right=of classify] {Lookup-Based \\ Retriangulation};

% Draw arrows
\draw[flowchart arrow] (init) -- (embed);
\draw[flowchart arrow] (embed) -- (preprocess);
\draw[flowchart arrow] (preprocess) -- (classify);
\draw[flowchart arrow] (classify) -- (lookup);
\end{tikzpicture}

\vspace{0.5em}
\caption{Overview of the SBMT execution pipeline. The process begins with a regular triangular grid, followed by boundary embedding, threshold-based preprocessing, discrete classification of boundary–triangle intersections, and local retriangulation via symbolic lookup templates. Each stage is deterministic, parallelizable, and free of global mesh coupling.}
\label{fig:sbmt-flowchart}
\end{figure}
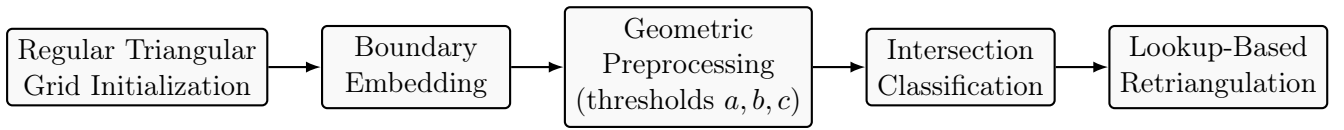

\section{Related Work}

The widespread use of raster data in visual computing, scientific imaging, and numerical simulations demands triangulation methods that are \emph{geometry-aware and structure-preserving}~\cite{gonzales1987digital, demmel1997applied, leyva2023satellite}. However, most existing methods—whether discrete~\cite{lorensen1998marching, treece1999regularised, hilton1996marching} or continuous~\cite{chew1987constrained, held2001fist, Shewchuk1996Triangle}—struggle to balance fidelity, structure, efficiency, and regularity in raster-based domains.

We categorize existing constrained triangulation methods into four families:
(1) lookup-table–based schemes,
(2) continuous geometry-based triangulations with or without Steiner insertion,
(3) curved and structure-aware image triangulations, and
(4) parallel refinement strategies.

\paragraph{Lookup-Table-Based Methods}
Classical lookup-table methods~\cite{lorensen1998marching, treece1999regularised, hilton1996marching} reconstruct surfaces by matching local vertex patterns to pre-defined triangulation rules. These methods are efficient but rely on grid resolution, limiting precision for accurate boundary representation. Our approach builds on this paradigm but enhances it with symbolic intersection taxonomy and $C_3$-symmetric templates, enabling composable, closed-form reconstruction for PDE solvers.

\paragraph{Continuous Geometry-Based Methods}
Constrained Delaunay Triangulation (CDT)~\cite{chew1987constrained} and its variants achieve $O(n \log n)$ complexity but are limited by the need for careful handling of strip interfaces, which limits parallel scalability. The FIST algorithm~\cite{held2001fist} offers fast triangulation but struggles with parallelism and may generate low-quality triangles, especially for large-scale polygons. Point-insertion methods like Delaunay Refinement~\cite{ruppert1995delaunay} offer advantages in controlling triangle quality but face scalability and efficiency issues. \texttt{Triangle}~\cite{Shewchuk1996Triangle, shewchuk2002delaunay, shewchuk2002good} is widely used but does not exploit the regularity of raster inputs, limiting its efficiency in dense bitmap-derived domains.

\paragraph{Curved and structure-aware triangulation}
Structure-aware image triangulation methods optimize an image-approximation or geometric energy over vertex placement/connectivity, often using curved or higher-order elements (e.g., TriWild and related schemes~\cite{hu2019triwild, xiao2022image},
data-dependent triangulations~\cite{Rippa92, li2013tuned, peyre2011review}, and stylized triangulations~\cite{lawonn2019stylized}).
While expressive, they typically require iterative global/semi-global optimization and complex refinement data structures, which can be expensive on high-resolution bitmaps.
SBMT instead assumes a given boundary network and performs a one-pass, template-driven remeshing on a fixed equilateral scaffold to obtain a boundary-exact, PDE-ready straight-edge discretization; hence these methods are complementary rather than direct baselines.

\paragraph{Parallel Triangulation}
Parallel methods~\cite{parallelrefinement, chernikov2006parallel, algo872} improve scalability by partitioning the domain into subregions and applying local refinements. However, they face challenges in GPU or embedded systems due to complex refinement logic and inter-thread communication, which hinder real-time performance. In contrast, SBMT uses a stateless, rule-driven architecture, where each triangle is processed independently, ensuring high compatibility with GPU and embedded systems.

We now describe the proposed method in detail.

\section{Method}\label{sec:main}

We present the \emph{Structured Bitmap-to-Mesh Triangulation} (SBMT) framework—a symbolic, template-driven remeshing method that embeds discrete polygonal boundaries into a regular triangular grid with guaranteed geometric fidelity and structural consistency. Unlike heuristic or recursive pipelines, SBMT resolves all segment–triangle intersections through a finite, deterministic rule set, enabling reproducible and synchronization-free execution.

SBMT is built upon three structural components:
\begin{enumerate}
  \item \textbf{Intersection Classification:} A complete taxonomy of all topologically and geometrically distinct segment–triangle intersection types;
  \item \textbf{Template Lookup:} A finite retriangulation table mapping each type to a unique, conflict-free triangle replacement, ensuring local closure and global consistency;
  \item \textbf{Symbolic Clipping Scheme:} A deterministic rule set that embeds boundaries via local template application along intersection paths, guaranteeing topological validity and alignment.
\end{enumerate}

Together, these components define a rule-based, stateless system that produces angle-bounded, interpolation-stable, and highly parallelizable meshes over raster domains.

\subsection{Preliminaries: Limitations of Raster Representations and the Need for Structured Reconstruction}

Raster data, typically in bitmap form, is the default format for image acquisition, storage, and early-stage processing due to its grid regularity, memory compactness, and hardware compatibility. Consequently, pixel-based signals dominate real-time imaging, embedded vision, and GPU-accelerated simulations.

However, despite its practical ubiquity, raster representation suffers from severe limitations in geometric and topological expressiveness:
\begin{itemize}
  \item \textbf{Differential operators are not naturally defined:} Core analytical tools such as gradients, divergence, or Laplacians are most naturally formulated on smooth domains, whereas pixel grids only provide coarse, grid-aligned finite-difference approximations.
  \item \textbf{Geometric structures are severely degraded:} Features such as normals, curvature, and tangent-aligned contours become jagged and aliased during rasterization, making high-quality reconstruction nontrivial.
  \item \textbf{Topological coherence is fragile:} Pixel boundaries form jagged, axis-aligned patterns with only implicit connectivity, which complicates robust segment tracing, mesh generation, and downstream PDE solvers.
\end{itemize}

These drawbacks pose significant challenges for geometry-aware numerical computation. In many pipelines, signals are rasterized early—either through sensor sampling or intermediate discretization—necessitating a reconstruction step that restores analytic continuity and structural integrity prior to mesh-based simulation.

To address this, we propose a structured remeshing framework that converts discrete raster boundaries into a Planar Straight Line Graph (PSLG), followed by conflict-free embedding into a regular triangular grid. Unlike global meshing strategies, SBMT performs fully local retriangulation via a symbolic template system, preserving the underlying scaffold and supporting deterministic, stateless execution.

The resulting mesh is topologically consistent, faithfully aligned to the original raster boundary, and supports well-defined discrete differential operators, allowing PDE solvers to achieve geometric accuracy and numerical robustness—particularly for elliptic problems where solution quality is sensitive to mesh regularity near boundaries.

\paragraph{Input boundary as prescribed prior}
SBMT assumes that the input boundary is provided as a polygonal network extracted from the raster (e.g., via segmentation or contour tracing) and treats it as prescribed prior information.
This is analogous to specifying boundary data in PDE discretization: the mesh is constructed \emph{conditioned on} this boundary so that subsequent elliptic/parabolic solves are carried out on a boundary-conforming domain, rather than using meshing itself to infer or optimize the boundary.

In the sections that follow, we detail the four structural pillars of SBMT:
\begin{itemize}
  \item a geometric protocol for classifying segment–triangle intersections;
  \item a template-based retriangulation scheme for local updates;
  \item a parallel execution architecture for scalable deployment;
  \item and a time complexity analysis to characterize performance bounds.
\end{itemize}

\subsection{Geometric Protocol for Safe Retriangulation}
\label{sec:geometric-protocol}

To ensure deterministic, conflict-free retriangulation, we impose two mild geometric constraints on how the input digital boundary intersects the equilateral base mesh. Their role is to eliminate degeneracies and to bound the space of admissible segment--triangle interaction patterns, so that all nontrivial cases can be resolved by a finite lookup table.

\paragraph{(1) Angular constraint (digital boundary)}
The constraint is imposed on the \emph{polygonal boundary extracted from the bitmap}, not on the underlying continuous contour: the interior angle between two consecutive boundary segments is required to be at least $90^\circ$.
This is a mild regularity condition for raster-derived boundaries. Standard 4-/8-connected contour tracing produces chain-code turns in multiples of $45^\circ$, hence polygonal interior angles are typically $\ge 90^\circ$ (e.g., a $45^\circ$ direction change corresponds to a $135^\circ$ interior
angle). If an extracted chain contains isolated nonconforming discrete corners, we enforce the condition by a local simplification/splitting step without changing the boundary topology. Sharp \emph{continuous} corners (e.g., $36^\circ$ star tips) are represented by short staircase patterns; increasing
image resolution refines this approximation while preserving the discrete $\ge 90^\circ$ condition.

\paragraph{(2) Length constraint (no fragment clustering)}
Each boundary segment is required to be longer than the local maximum edge length of the base mesh. This rules out overly short fragments that would otherwise create clustered intersection points inside a single cell and complicate local case handling.

\medskip
\noindent
\textbf{Consequences used by SBMT.}
Under the above constraints (together with the intersection-counting rule in Algorithm~\ref{alg:intersection-type}), SBMT relies on the following three properties:
\begin{enumerate}[label=(P\arabic*)]
  \item \textbf{Bounded locality:} each base triangle is intersected by at most
        two boundary segments from a chain (Lemma~\ref{lem:segment-triangle-bound}).
  \item \textbf{Bounded boundary events:} after merging coincident incidences
        (vertex-on-segment, endpoint-on-edge), a triangle boundary contains at most
        three distinct intersection points relevant to retriangulation.
  \item \textbf{Finite taxonomy:} therefore the set of nontrivial admissible
        segment--triangle patterns is finite and can be pre-enumerated
        (Appendix~F), with each canonical type assigned one or more static
        retriangulation templates (Appendix~G).
\end{enumerate}
The resulting symbolic \emph{coverage} of the template system is formalized by Theorem~\ref{thm:lookup_table_completeness} (proof in Appendix~B), which states that every admissible nontrivial configuration falls into one of the canonical types and is handled by at least one lookup-table entry.

\begin{figure}[htbp]
\centering
\includegraphics[width=0.40\linewidth]{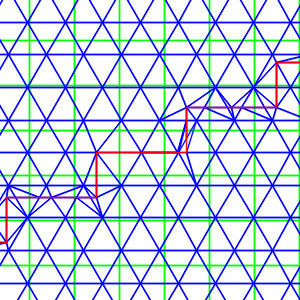}
\caption{
Structured equilateral mesh over a bitmap-derived domain.
Green squares indicate unit pixels; blue lines denote the triangular grid;
red polylines represent boundary segments intersecting the mesh.
This configuration satisfies the geometric protocol of
Section~\ref{sec:geometric-protocol} and supports lookup-based local retriangulation while preserving topological and numerical consistency. The grid shown uses parameters \( a = 0.2 \), \( b = 0.08 \), \( c = 0 \) (as defined in Section~\ref{subsubsec:theoretical-properties}), with triangle edge length \( \sqrt{0.7} \).
}
\label{fig:meshshow}
\end{figure}

\paragraph{Remark on fixed resolution and possible multi-resolution extensions}
In this work, SBMT operates on a single-resolution equilateral scaffold with edge length $e<1$ pixel. For bitmap-derived domains, resolution is typically chosen upstream: higher image resolution increases geometric fidelity at the cost of more elements, while lower resolution trades small-scale detail for
speed and memory efficiency. Extending SBMT to a conforming multi-resolution scaffold is conceptually possible, but would require additional machinery (e.g., refinement/coarsening policies and conformity enforcement) while preserving the same finite case taxonomy; we leave this direction for future work.

\subsection{Structured Triangular Filling of the Bitmap Domain}

To enable geometry-consistent numerical computation over raster data, we overlay a globally aligned triangular grid with irrational edge length \( h < 1 \) onto the bitmap domain. This structured scaffold serves as the geometric backbone of our remeshing pipeline, ensuring topological regularity and bounded triangle quality.

Pixels are modeled as unit squares centered at integer coordinates. The irrational value of \( h \) prevents periodic aliasing and grid-locking, leading to statistically uniform, scale-consistent intersection behavior. All triangle–pixel interactions remain local and geometry-driven.

Boundary segments are embedded into this mesh using a finite set of symbolic retriangulation templates, introduced below. Scalar fields from the bitmap are projected onto each triangle to support field-aware interpolation.

As shown in Figure~\ref{fig:meshshow}, segment–triangle intersections include edge crossings and vertex piercings. The interpolation accuracy is formally bounded in Theorem~\ref{thm:spectral_interp}, which guarantees spectral convergence for bandlimited signals over equilateral meshes—ensuring that SBMT retains numerical consistency and avoids distortion artifacts.

We now classify all segment–triangle intersection types and introduce the corresponding retriangulation strategies.

\begin{figure}[ht]
\begin{center}
\renewcommand{\arraystretch}{1.5}
\begin{tabular}{cccc}
% Row 1
\begin{tikzpicture}[scale=0.7]
\coordinate (A) at (1,0);
\coordinate (B) at (3,0);
\coordinate (C) at (2,1.732);
\draw[thick] (A) -- (B) -- (C) -- cycle;
\coordinate (P1) at (0,1);
\coordinate (P2) at (2,0.5);
\draw[red, thick] (P1) -- (P2);
\node at (2.2, -0.3) {\scriptsize 1 intersection};
\end{tikzpicture}
&
\begin{tikzpicture}[scale=0.7]
\coordinate (A) at (1,0);
\coordinate (B) at (3,0);
\coordinate (C) at (2,1.732);
\draw[thick] (A) -- (B) -- (C) -- cycle;
\coordinate (P1) at (0,1);
\coordinate (P2) at (1.5,0.866);
\draw[red, thick] (P1) -- (P2);
\node at (2.2, -0.3) {\scriptsize 1 intersection};
\end{tikzpicture}
&
\begin{tikzpicture}[scale=0.7]
\coordinate (A) at (1,0);
\coordinate (B) at (3,0);
\coordinate (C) at (2,1.732);
\draw[thick] (A) -- (B) -- (C) -- cycle;
\coordinate (P1) at (0,-0.5); % x为0，穿过A点
\coordinate (P2) at (2,0.5);
\draw[red, thick] (P1) -- (P2);
\node at (2.2, -0.3) {\scriptsize 2 intersections};
\end{tikzpicture}
&
\begin{tikzpicture}[scale=0.7]
\coordinate (A) at (1,0);
\coordinate (B) at (3,0);
\coordinate (C) at (2,1.732);
\draw[thick] (A) -- (B) -- (C) -- cycle;
\coordinate (P1) at (0.5,1); % x为0，穿过A点
\coordinate (P2) at (3.5,1);
\draw[red, thick] (P1) -- (P2);
\node at (2.2, -0.3) {\scriptsize 2 intersections};
\end{tikzpicture}
\\
% Row 2
\begin{tikzpicture}[scale=0.7]
\coordinate (A) at (1,0);
\coordinate (B) at (3,0);
\coordinate (C) at (2,1.732);
\draw[thick] (A) -- (B) -- (C) -- cycle;
\coordinate (P1) at (1.5,0.866); % x为0，穿过A点
\coordinate (P2) at (3.5,1);
\draw[red, thick] (P1) -- (P2);
\node at (2.2, -0.3) {\scriptsize 2 intersections};
\end{tikzpicture}
&
\begin{tikzpicture}[scale=0.7]
\coordinate (A) at (1,0);
\coordinate (B) at (3,0);
\coordinate (C) at (2,1.732);
\draw[thick] (A) -- (B) -- (C) -- cycle;
\coordinate (P1) at (1.5,1.732); % x为0，穿过A点
\coordinate (P2) at (3.5,1.732);
\draw[red, thick] (P1) -- (P2);
\node at (2.2, -0.3) {\scriptsize 2 intersections};
\end{tikzpicture}
&
\begin{tikzpicture}[scale=0.7]
\coordinate (A) at (1,0);
\coordinate (B) at (3,0);
\coordinate (C) at (2,1.732);
\draw[thick] (A) -- (B) -- (C) -- cycle;
\coordinate (P1) at (1.5,0); % x为0，穿过A点
\coordinate (P2) at (4,0);
\draw[red, thick] (P1) -- (P2);
\node at (2.2, -0.3) {\scriptsize 2 intersections};
\end{tikzpicture}
&
\begin{tikzpicture}[scale=0.7]
\coordinate (A) at (1,0);
\coordinate (B) at (3,0);
\coordinate (C) at (2,1.732);
\draw[thick] (A) -- (B) -- (C) -- cycle;
\coordinate (P1) at (3,0); % x为0，穿过A点
\coordinate (P2) at (3.5,2.7);
\draw[red, thick] (P1) -- (P2);
\node at (2.2, -0.3) {\scriptsize 2 intersections};
\end{tikzpicture}
\\
% Row 3
\begin{tikzpicture}[scale=0.7]
\coordinate (A) at (1,0);
\coordinate (B) at (3,0);
\coordinate (C) at (2,1.732);
\draw[thick] (A) -- (B) -- (C) -- cycle;
\coordinate (P1) at (2,0); % x为0，穿过A点
\coordinate (P2) at (2,2);
\draw[red, thick] (P1) -- (P2);
\node at (2.2, -0.3) {\scriptsize 3 intersections};
\end{tikzpicture}
&
\begin{tikzpicture}[scale=0.7]
\coordinate (A) at (1,0);
\coordinate (B) at (3,0);
\coordinate (C) at (2,1.732);
\draw[thick] (A) -- (B) -- (C) -- cycle;
\coordinate (P1) at (2,-1); % x为0，穿过A点
\coordinate (P2) at (2,1.732);
\draw[red, thick] (P1) -- (P2);
\node at (2.2, -0.3) {\scriptsize 3 intersections};
\end{tikzpicture}
&
\begin{tikzpicture}[scale=0.7]
\coordinate (A) at (1,0);
\coordinate (B) at (3,0);
\coordinate (C) at (2,1.732);
\draw[thick] (A) -- (B) -- (C) -- cycle;
\coordinate (P1) at (2,-1); % x为0，穿过A点
\coordinate (P2) at (2,2);
\draw[red, thick] (P1) -- (P2);
\node at (2.2, -0.3) {\scriptsize 3 intersections};
\end{tikzpicture}
&
\begin{tikzpicture}[scale=0.7]
\coordinate (A) at (1,0);
\coordinate (B) at (3,0);
\coordinate (C) at (2,1.732);
\draw[thick] (A) -- (B) -- (C) -- cycle;
\coordinate (P1) at (0,0); % x为0，穿过A点
\coordinate (P2) at (3.5,0);
\draw[red, thick] (P1) -- (P2);
\node at (2.2, -0.3) {\scriptsize 3 intersections};
\end{tikzpicture}

\\
% Row 4
\begin{tikzpicture}[scale=0.7]
\coordinate (A) at (1,0);
\coordinate (B) at (3,0);
\coordinate (C) at (2,1.732);
\draw[thick] (A) -- (B) -- (C) -- cycle;
\coordinate (P1) at (1,0); % x为0，穿过A点
\coordinate (P2) at (4,0);
\draw[red, thick] (P1) -- (P2);
\node at (2.2, -0.3) {\scriptsize 3 intersections};
\end{tikzpicture}
&

&

&

\end{tabular}
\end{center}

\caption{
Exhaustive enumeration of canonical intersection patterns induced by a single boundary segment intersecting a triangular cell. Each configuration is treated as an atomic unit, abstracted independently of any additional segments that may also intersect the same cell. Subfigures are labeled by the number of intersection points (e.g., \texttt{1 intersection}, \texttt{2 intersections}, etc.), omitting edge identities. These localized, template-based configurations form the basis of our lookup-driven remeshing framework, enabling robust and scalable subdivision of structured triangular domains with embedded polygonal boundaries.
}
\label{fig:lookup_table}
\end{figure}
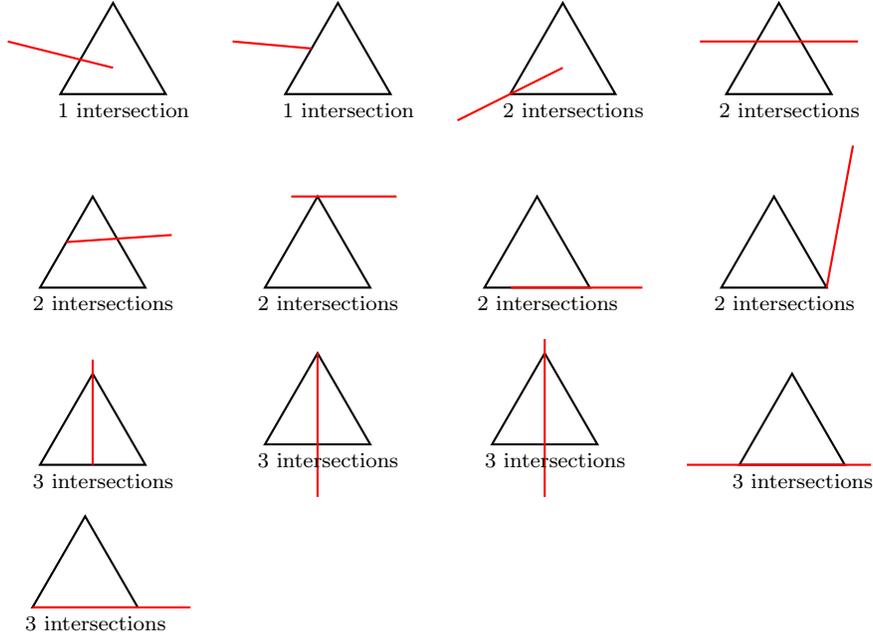

\subsubsection{Intersection Classification}
\label{subsubsec:intersection_classification}
A key step in embedding polygonal boundaries is to classify how each segment interacts with a structured triangular cell, providing symbolic keys for deterministic template lookup.

Let $T = \triangle ABC$ be an equilateral triangle and $L = \overline{PQ}$ a boundary segment. For each triangle edge $e$, the intersection with $L$ is encoded as follows:

\begin{algorithm}[htb]
\caption{Classify Segment–Edge Intersection}
\label{alg:intersection-type}
\begin{algorithmic}[1]
\REQUIRE Triangle edge $e$; segment $L = \overline{PQ}$
\ENSURE Intersection count and attribution
\IF{$L$ crosses $e$ (interior or endpoint-on-edge)}
    \STATE Count as one intersection
\ELSIF{vertex of $e$ lies on $L$}
    \STATE Count as two intersections (for both incident edges)
\ELSIF{$L$ overlaps $e$ (colinear)}
    \STATE Count as one intersection;
    \STATE Attribute by fixed priority: vertex $A$, $B$, endpoint $P$, then $Q$
\ENDIF
\end{algorithmic}
\end{algorithm}

Each segment produces at most three intersections per triangle, and every intersection pattern is uniquely encoded by a bitmask (AB~$=2$, BC~$=4$, CA~$=8$). Under the geometric constraints (Section~\ref{sec:geometric-protocol}), this ensures a finite and deterministic configuration space, capturing all atomic interaction types: edge crossings, vertex contacts, and colinear overlaps.

These symbolic encodings, as listed in Figure~\ref{fig:lookup_table}, form the basis of our lookup-driven retriangulation method.

\subsubsection{Local Retriangulation Strategy}
\label{subsubsec:local_retriangulation_strategy}

SBMT updates the mesh \emph{locally}: each base triangle intersected by the boundary is replaced by a small set of sub-triangles that (i) embed the boundary segment(s) exactly and (ii) preserve conformity with neighboring cells. To make this update deterministic and parallelizable, SBMT uses a \emph{table-driven} strategy: once the local boundary--triangle interaction is classified, a fixed retriangulation template is selected from a finite lookup table (Appendix~G).

\paragraph{Intersection notation and classification}
For triangles intersected by two consecutive boundary segments along the same boundary chain, we encode the local configuration in a base triangle $K$ by an ordered pair $(m,n)$, where
$m$ (resp.\ $n$) is the number of \emph{distinct} intersection points contributed by the first (resp.\ second) segment on $\partial K$, counted by the deterministic convention in
Algorithm~\ref{alg:intersection-type}. Here the ``first'' and ``second'' segments follow the oriented order along the boundary chain (e.g., increasing polyline index), so $(m,n)$ is
well-defined across triangles. We count both edge crossings and vertex incidences; intersections that coincide after snapping are merged and counted once. The resulting label $(m,n)$ (together
with its rotation/mirror variant) serves as the index into the finite template table.

We group admissible configurations by (a) the number of intersecting segments (single vs.\ two-segment cases)
and (b) the topology of the induced intersection points on $\partial K$. Single-segment cases (proper edge crossing, vertex incidence, and colinear overlap handled by a fixed priority rule)
are summarized in Figure~\ref{fig:single_segment_cases}. Two-segment cases include type $(1,1)$ (Figure~\ref{fig:1_1}) as well as mixed types such as $(1,2)$, $(1,3)$, and $(2,2)$ enumerated in
Section~\ref{subsubsec:intersection_classification} and Appendix~F.
For each admissible type, Appendix~G provides one or more static retriangulation templates, which embed the boundary exactly, respect the quality safeguards in Section~\ref{subsubsec:theoretical-properties}, and enable independent per-triangle updates.

\begin{lemma}[Bound on segment--triangle intersections]
\label{lem:segment-triangle-bound}
Let the polygonal boundary be a union of grid-aligned segments satisfying
the angular and length constraints in
Section~\ref{sec:geometric-protocol}.
Then for every base triangle $T$ in the equilateral mesh the following
properties hold:
\begin{enumerate}[label=(\roman*)]
  \item $T$ is intersected by at most two consecutive boundary segments
        along any given boundary chain;
  \item for each such segment, Algorithm~\ref{alg:intersection-type}
        produces at most one edge-level intersection event per edge of
        $T$;
  \item consequently, the pattern of edge-level intersection events on
        $T$ belongs to a finite family of cases, and every admissible
        configuration is one of those enumerated in
        Figures~\ref{fig:single_segment_cases}--\ref{fig:1_1}
        and Appendix~F, each with an associated static retriangulation
        template.
\end{enumerate}
\end{lemma}

\begin{proof}[Proof sketch]
By the geometric protocol in Section~\ref{sec:geometric-protocol}, a boundary chain can interact with a base triangle $T$ only through a single local entry--exit across its neighborhood; hence at most two consecutive segments along the chain can intersect $T$ (see Appendix~J in the Supplementary Material
for a detailed geometric argument). For each segment--edge pair, a straight segment intersects a straight edge in at most one point (or yields a single colinear-overlap event), so Algorithm~\ref{alg:intersection-type} records at most one edge-level event per edge. Since $T$ has three edges and at most two segments are relevant, only finitely many edge-event patterns can occur (up to triangle symmetry); these are exactly the cases enumerated in Figures~\ref{fig:single_segment_cases}--\ref{fig:1_1} and Appendix~F, each mapped to a static retriangulation template.
\end{proof}

With Lemma~\ref{lem:segment-triangle-bound} in place, we restrict our
attention to the finite set of intersection types shown in
Figures~\ref{fig:single_segment_cases}--\ref{fig:1_1} and
Appendix~F. For each type, SBMT associates a unique retriangulation
template that embeds the corresponding boundary segments and
intersection events exactly.

Most configurations involve 1 to 4 intersection points; a rare fifth occurs in $(2,3)$ cases where a segment passes directly through a triangle vertex (internal vertex contact) in addition to edge crossings.

\begin{figure}[htbp]
\begin{center}
\renewcommand{\arraystretch}{1.5}
\begin{tabular}{cccc}
% 图 (a)
\begin{tikzpicture}[scale=0.7]
\node at (0.5, 2.1) {(a)};  % 编号
\coordinate (A) at (1,0);
\coordinate (B) at (3,0);
\coordinate (C) at (2,1.732);
\draw[thick] (A) -- (B) -- (C) -- cycle;
\node[below left] at (A) {$A$};
\node[below right] at (B) {$B$};
\node[above] at (C) {$C$};
\coordinate (P1) at (0,1.732);
\coordinate (P2) at (3,1.732);
\draw[red, thick] (P1) -- (P2);
\filldraw[red] (P1) circle (0.03) node[left] {$P_1$};
\filldraw[red] (P2) circle (0.03) node[below right] {$P_2$};
\node at (2.2, -0.7) {\scriptsize 2 intersections};
\end{tikzpicture}
&
% 图 (b)
\begin{tikzpicture}[scale=0.7]
\node at (0.5, 2.1) {(b)};  % 编号
\coordinate (A) at (1,0);
\coordinate (B) at (3,0);
\coordinate (C) at (2,1.732);
\draw[thick] (A) -- (B) -- (C) -- cycle;
\node[below left] at (A) {$A$};
\node[below right] at (B) {$B$};
\node[above] at (C) {$C$};
\coordinate (P1) at (0,1.0);
\coordinate (P2) at (3,1.0);
\draw[red, thick] (P1) -- (P2);
\filldraw[red] (P1) circle (0.03) node[left] {$P_1$};
\filldraw[red] (P2) circle (0.03) node[below right] {$P_2$};
\coordinate (Q1) at (2.423, 1.0);
\coordinate (Q2) at (1.578, 1.0);
\filldraw[blue] (Q1) circle (0.03) node[above right] {$Q_1$};
\filldraw[blue] (Q2) circle (0.03) node[above left] {$Q_2$};
\node at (2.2, -0.7) {\scriptsize 2 intersections};
\end{tikzpicture}
&
% 图 (c)
\begin{tikzpicture}[scale=0.7]
\node at (0.5, 2.1) {(c)};  % 编号
\coordinate (A) at (1,0);
\coordinate (B) at (3,0);
\coordinate (C) at (2,1.732);
\draw[thick] (A) -- (B) -- (C) -- cycle;
\node[below left] at (A) {$A$};
\node[below right] at (B) {$B$};
\node[above] at (C) {$C$};
\coordinate (P1) at (0,-0.5);
\coordinate (P2) at (3,1.0);
\draw[red, thick] (P1) -- (P2);
\filldraw[red] (P1) circle (0.03) node[left] {$P_1$};
\filldraw[red] (P2) circle (0.03) node[above right] {$P_2$};
\coordinate (Q) at (2.552, 0.776);
\filldraw[blue] (Q) circle (0.03) node[above left] {$Q$};
\node at (2.2, -0.7) {\scriptsize 3 intersections};
\end{tikzpicture}
&
% 图 (d)
\begin{tikzpicture}[scale=0.7]
\node at (0.5, 2.1) {(d)};  % 编号
\coordinate (A) at (1,0);
\coordinate (B) at (3,0);
\coordinate (C) at (2,1.732);
\draw[thick] (A) -- (B) -- (C) -- cycle;
\node[below left] at (A) {$A$};
\node[below right] at (B) {$B$};
\node[above] at (C) {$C$};
\coordinate (P1) at (0,0);
\coordinate (P2) at (3.5,0);
\draw[red, thick] (P1) -- (P2);
\filldraw[red] (P1) circle (0.03) node[left] {$P_1$};
\filldraw[red] (P2) circle (0.03) node[below right] {$P_2$};
\node at (2.2, -0.7) {\scriptsize 3 intersections};
\end{tikzpicture}
\end{tabular}
\end{center}

\caption{
Geometric classification of all local configurations in which a single boundary segment intersects a triangle cell. The four cases include vertex intersection (a), edge piercing (b), single-edge contact (c), and edge-aligned overlap (d). This classification forms the base of the local retriangulation lookup table.
}
\label{fig:single_segment_cases}
\end{figure}
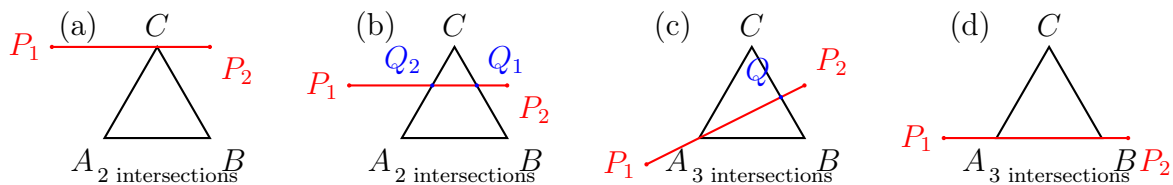

The retriangulation table is thus defined over this compact and exhaustive configuration space, with each entry corresponding to a unique and deterministic retriangulation rule.
Figure~\ref{fig:1_1} illustrates a representative $(1,1)$ retriangulation case, annotated with geometric configurations and embedded intersection points. Additional multi-segment configurations are provided in Supplementary Appendix~F.
The complete set of retriangulation templates and their exact specifications are tabulated and described in detail in Appendix G of the Supplementary Material.

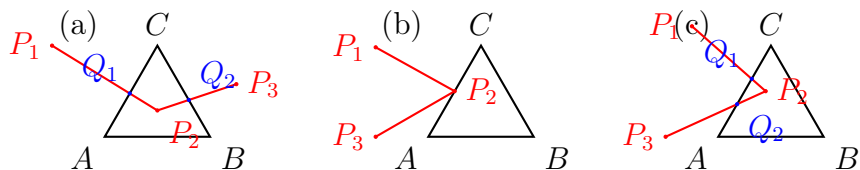
\begin{figure}[htbp]
\begin{center}
\renewcommand{\arraystretch}{1.5}
\begin{tabular}{cccc}
% 图 (a)
\begin{tikzpicture}[scale=0.7]
\node at (0.5, 2.1) {(a)};  % 编号
\coordinate (A) at (1,0);
\coordinate (B) at (3,0);
\coordinate (C) at (2,1.732);
\draw[thick] (A) -- (B) -- (C) -- cycle;
\node[below left] at (A) {$A$};
\node[below right] at (B) {$B$};
\node[above] at (C) {$C$};

% 红折线段
\coordinate (P1) at (0,1.732);
\coordinate (P2) at (2,0.5);
\coordinate (P3) at (3.5,1);
\draw[red, thick] (P1) -- (P2);
\draw[red, thick] (P2) -- (P3);
\filldraw[red] (P1) circle (0.03) node[left] {$P_1$};
\filldraw[red] (P2) circle (0.03) node[below right] {$P_2$};
\filldraw[red] (P3) circle (0.03) node[right] {$P_3$};

% 交点
\coordinate (Q1) at (1.476, 0.823);
\coordinate (Q2) at (2.597, 0.699);
\filldraw[blue] (Q1) circle (0.03) node[above left] {$Q_1$};
\filldraw[blue] (Q2) circle (0.03) node[above right] {$Q_2$};
\end{tikzpicture}
&
% 图 (b)
\begin{tikzpicture}[scale=0.7]
\node at (0.5, 2.1) {(b)};  % 编号
\coordinate (A) at (1,0);
\coordinate (B) at (3,0);
\coordinate (C) at (2,1.732);
\draw[thick] (A) -- (B) -- (C) -- cycle;
\node[below left] at (A) {$A$};
\node[below right] at (B) {$B$};
\node[above] at (C) {$C$};
\coordinate (P1) at (0,1.7);
\coordinate (P2) at (1.5,0.866);
\coordinate (P3) at (0,0);
\draw[red, thick] (P1) -- (P2);
\draw[red, thick] (P2) -- (P3);
\filldraw[red] (P1) circle (0.03) node[left] {$P_1$};
\filldraw[red] (P2) circle (0.03) node[right] {$P_2$};
\filldraw[red] (P3) circle (0.03) node[left] {$P_3$};
\end{tikzpicture}
&
% 图 (c)
\begin{tikzpicture}[scale=0.7]
\node at (0.5, 2.1) {(c)};  % 编号
\coordinate (A) at (1,0);
\coordinate (B) at (3,0);
\coordinate (C) at (2,1.732);
\draw[thick] (A) -- (B) -- (C) -- cycle;
\node[below left] at (A) {$A$};
\node[below right] at (B) {$B$};
\node[above] at (C) {$C$};

% 红线段
\coordinate (P1) at (0.5,2.1);
\coordinate (P2) at (1.9,0.866);
\coordinate (P3) at (0,0);
\draw[red, thick] (P1) -- (P2);
\draw[red, thick] (P3) -- (P2);
\filldraw[red] (P1) circle (0.03) node[left] {$P_1$};
\filldraw[red] (P2) circle (0.03) node[right] {$P_2$};
\filldraw[red] (P3) circle (0.03) node[left] {$P_3$};

% 交点标注
\coordinate (Q1) at (1.636, 1.101);
\coordinate (Q2) at (1.358, 0.619);
\filldraw[blue] (Q1) circle (0.03) node[above left] {$Q_1$};
\filldraw[blue] (Q2) circle (0.03) node[below right] {$Q_2$};
\end{tikzpicture}
&

\end{tabular}
\end{center}

\caption{
Geometric configurations of type \textbf{(1,1)}, where two boundary segments intersect a triangle at one point each. These cases are grouped by intersection count rather than edge location, forming a unified retriangulation pattern applicable to all edge combinations.
}
\label{fig:1_1}
\end{figure}

This structured and exhaustive classification enables a unified retriangulation pipeline that simultaneously preserves boundary fidelity, accommodates arbitrary domain topologies (including those of higher genus), and guarantees parallelizability. 
In the present work, we primarily illustrate the SBMT framework using simply connected domains to facilitate the exposition of differential forms and the demonstration of numerical computations.

Although the SBMT angle protocol enforces that adjacent boundary segments form angles of at least $90^\circ$—thereby excluding complex intersection patterns such as type~\((3,3)\)—we emphasize that the classification framework remains complete even when this constraint is relaxed.

In particular, type~\((3,3)\), illustrated in Appendix F of the Supplementary Material, represents the most complex feasible intersection, where two segments each contribute three entry or exit points within a triangle. As long as no more than two segments intersect a given triangle, even sharp junctions with angles as small as $45^\circ$ can be robustly handled using the existing retriangulation table. This ensures the theoretical completeness and closure of the framework across all local geometric configurations.

\subsubsection{Template-driven local retriangulation}
\label{subsubsec:template-retriangulation}

\paragraph{Single-valued lookup (no runtime template selection)}
A coarse intersection \emph{class} such as $(1,1)$ or $(2,2)$ may contain multiple geometric subcases (cf.\ Appendix~G). SBMT, however, performs \emph{no} runtime choice among alternative triangulations. Instead, each base triangle $K$ is assigned an
\emph{atomic, canonical configuration key} that refines the coarse class using edge-attributed, ordered intersection records and a fixed symmetry-breaking rule.

Concretely, for each $K$, we form the deterministic key
\[
\kappa(K)\;:=\;\Big(\mathrm{type}(K),\;\mathbf{q}_{\partial K},\;\sigma(K)\Big),
\]
where $\mathrm{type}(K)$ is the coarse counting class,
$\mathbf{q}_{\partial K}$ is the ordered list of intersection records on $\partial K$ (including incident edge identities and parametric order along each edge) read from the global registry $\mathcal{R}$, and $\sigma(K)$ is the canonicalization tag induced by the fixed $D_3$ symmetry-breaking convention.

\begin{assumption}[Frozen registry and single-valued lookup]
\label{assump:frozen-lookup}
The registry $\mathcal{R}$ is constructed once, then frozen and read-only. For each base triangle $K$, the lookup table is a \emph{single-valued} map
\[
\mathcal{L}:\mathsf{Key}\to \mathsf{Patch},\qquad P_K=\mathcal{L}\big(\kappa(K)\big),
\]
so the update of $K$ is defined out-of-place by the unique patch $P_K$ determined by its canonical key $\kappa(K)$.
\end{assumption}

Equivalently, any ambiguity that exists at the coarse-class level is eliminated by the refined key $\kappa(K)$ together with the canonicalization convention; SBMT therefore never faces a runtime decision among multiple valid templates.

\paragraph{Template selection policy}
For configurations with non-unique retriangulations (e.g., representative $(2,2)$ and $(1,2)$ cases illustrated in Appendix~F), SBMT selects a single predefined template from the precomputed set and does not invoke local Delaunay refinement. This design choice keeps the procedure simple and reproducible, while preserving robust boundary embedding and parallelizable local updates.

\begin{algorithm}[htb]
\caption{Template-driven retriangulation of a single triangle}
\label{alg:sbmt-template}
\begin{algorithmic}[1]
\REQUIRE Triangle $K=\triangle ABC$; boundary-induced point set
$\mathcal{S}_K$ (ordered, with edge attribution); segment incidence info
\ENSURE Replace $K$ by sub-triangles

\STATE \textbf{Symbolic classification.}
Compute the type $\tau$ and a symmetry tag $s\in D_3$ (rotation/mirror)
using Algorithm~\ref{alg:intersection-type} and
Section~\ref{subsubsec:local_retriangulation_strategy}.

\STATE \textbf{Fetch canonical template.}
Load $\mathcal{P}_\tau=(\bar V_\tau,\bar F_\tau)$ from Appendix~G,
where $\bar V_\tau=\{\bar A,\bar B,\bar C\}\cup\{\bar q_j\}\cup\{\bar p_i\}$
are canonical vertex \emph{labels} and $\bar F_\tau$ is a list of oriented faces.

\STATE \textbf{Instantiate by a single map.}
Define $\pi_{K,\tau}:\bar V_\tau\rightarrow\mathbb{R}^2$ by
\[
\pi_{K,\tau}(\bar A,\bar B,\bar C)=(A,B,C),\qquad
\pi_{K,\tau}(\bar q_j)=q_j,\qquad
\pi_{K,\tau}(\bar p_i)=p_i,
\]
where $q_j,p_i\in\mathcal{S}_K$ are the corresponding geometric points
(ordered and edge-attributed as required by $\tau$), and pre-compose
$\pi_{K,\tau}$ with $s$ if $K$ is in a rotated/mirrored configuration.

\STATE \textbf{Emit sub-triangles.}
For each $(u,v,w)\in \bar F_\tau$, insert $\triangle\,\pi_{K,\tau}(u)\,\pi_{K,\tau}(v)\,\pi_{K,\tau}(w)$
with CCW orientation; remove $K$.

\end{algorithmic}
\end{algorithm}

Although local Delaunay flips could marginally improve element quality, we do not pursue them here because they introduce additional algorithmic complexity and synchronization/iteration, which would compromise SBMT’s purely table-driven and embarrassingly parallel design; we leave such hybrid variants for future work.

\paragraph{Topological Generality}
SBMT operates in a topology-agnostic fashion: each triangle is processed independently based on its local intersection pattern. Under mild geometric assumptions—e.g., each triangle intersected by at most two segments—the same lookup table supports multiply connected domains, interior holes, and open boundaries without modification.

High-level topology (e.g., genus, homology) is orthogonal to the retriangulation logic and may be handled separately if needed, but is not required for SBMT execution. The method imposes no global constraints or preprocessing.

\subsubsection{Theoretical Properties}
\label{subsubsec:theoretical-properties}

Beyond empirical design, the SBMT framework is based on structural guarantees that ensure correctness, determinism, and numerical applicability. We formally establish the following properties: (1) the finiteness and completeness of segment–triangle intersection classification; (2) the determinism of retriangulation outcomes; (3) the local-to-global composability of template application; and (4) the topological and geometric fidelity of the final mesh.

\paragraph{Finiteness and Completeness of Interaction Types}  
Under SBMT's geometric protocol—where each triangle is intersected by at
most two polygonal segments and adjacent segments form angles no sharper
than $90^\circ$—the number of distinct segment--triangle interaction
types is strictly bounded.
Lemma~\ref{lem:segment-triangle-bound} guarantees that each segment
produces at most one edge-level intersection event per edge of a base
triangle, so a segment can meet the triangle boundary in at most three
distinct points.
In the $(k,\ell)$ notation introduced above, this implies that
$k,\ell \in \{0,1,2,3\}$, with the trivial case $(0,0)$ corresponding to
no intersection and requiring no remeshing.

\begin{theorem}[Completeness of retriangulation lookup table under SBMT constraints]
\label{thm:lookup_table_completeness}
Let $T$ be a base triangle in the equilateral mesh, and let its
segment--triangle intersections satisfy the geometric protocol of
Section~\ref{sec:geometric-protocol} and the intersection rules of
Algorithm~\ref{alg:intersection-type}.
Then every admissible configuration that induces embedded intersection
points on $T$ belongs to one of the following canonical types:
\[
(0,1), (0,2), (0,3),\;
(1,1), (1,2), (2,2), (1,3),\;
\text{or } (2,3),
\]
in the $(k,\ell)$ counting described above.
Each of these types is explicitly enumerated in
Figures~\ref{fig:single_segment_cases}--\ref{fig:1_1} and Appendix~F and
is associated with at least one retriangulation pattern in the SBMT
lookup table (Appendix~G), ensuring complete coverage of all nontrivial,
remeshing-relevant cases.
\end{theorem}

\begin{proof}[Sketch]
By Lemma~\ref{lem:segment-triangle-bound}, each base triangle $T$ is intersected
by at most two consecutive boundary segments $S_1,S_2$. Under the counting
convention of Algorithm~\ref{alg:intersection-type}, each segment contributes
$k,\ell\in\{0,1,2,3\}$ (after merging coincident events by preprocessing), hence
every configuration has a label $(k,\ell)\in\{0,1,2,3\}^2$.
The trivial case $(0,0)$ needs no remeshing, while $(3,3)$ is excluded by the
SBMT protocol and Lemma~\ref{lem:segment-triangle-bound} (it would force an
inadmissible number of distinct embedded intersection events on $\partial T$).
Therefore, every nontrivial admissible configuration must be one of
\[
(0,1),(0,2),(0,3),(1,1),(1,2),(2,2),(1,3),(2,3).
\]
A complete case enumeration (up to symmetry) and the association to templates
are given in Supplementary Material, Appendix~B.
\end{proof}

\theoreticalproperty{Local-to-Global Consistency}

\begin{theorem}[Edge Conformity under Finite Precision]
\label{thm:precision_conformity}
Let \( T_i \) and \( T_j \) be two adjacent triangles sharing edge \( e \), intersected by a boundary segment at a point \( p_e \). Suppose all geometric predicates in SBMT employ a global absolute tolerance \( \epsilon > 0 \). Then, the retriangulated subdivisions of \( e \) by \( T_i \) and \( T_j \) are consistent up to an error bounded by \( \mathcal{O}(\epsilon) \).
\end{theorem}

\begin{proof}
This follows directly from Lemma~\ref{lemma:precision} in \ref{appendix:category}, which formalizes floating-point bounded consistency under the SBMT lookup functor \( F: \mathsf{Config} \to \mathsf{Template} \). Specifically, edge traversal order and tolerance-aware indexing ensure that adjacent triangles share a logically consistent interpretation of intersection events, up to bounded perturbation.
\end{proof}

\begin{remark}
SBMT employs a half-edge data structure with global intersection hashing and \(\epsilon\)-snapping, ensuring that shared edges remain topologically watertight even in the presence of finite-precision discrepancies. This mechanism supports mesh validity and avoids artifacts such as T-junctions or floating-point cracks.
\end{remark}

\theoreticalproperty{Deterministic Retriangulation}

\begin{lemma}[Uniqueness of Retriangulation per Atomic Type]
\label{lemma:unique_template}
For each atomic configuration \( c \in \mathcal{C} \), SBMT assigns a unique retriangulation template via a deterministic table lookup.
\end{lemma}

\begin{proof}
SBMT defines a functor \( F : \mathsf{Config} \to \mathsf{Template} \) (cf. \ref{appendix:category}), mapping each atomic configuration to a symbolic template.

By construction, \( F \) is injective on geometric equivalence classes \( \mathcal{C}/\sim \), ensuring that each configuration type selects a single, unambiguous template. Table-driven assignment eliminates runtime ambiguity.
\end{proof}

\begin{remark}
This determinism arises from the categorical structure of SBMT: functorial mapping guarantees consistency, while injectivity rules out conflicting retriangulations. These properties enable symbolic reasoning and formal extension of the framework.
\end{remark}

\theoreticalproperty{Boundary Preservation}

\begin{theorem}[Exact Embedding of Boundary Segments]
Let $\mathcal{S}$ be a set of input boundary segments, each intersecting only finitely many triangles. Then SBMT retriangulation embeds every $s \in \mathcal{S}$ as a contiguous, watertight chain of mesh edges that exactly interpolates its endpoints and intersection points.
\end{theorem}

\begin{proof}[Sketch]
Each segment–triangle intersection is resolved via a unique local template \\(Lemma~\ref{lemma:unique_template}), placing vertices at all intersection points. The resulting edge fragments are aligned across adjacent triangles (Lemma~\ref{lemma:precision}), ensuring global continuity and watertightness.

Thus, each boundary segment is exactly realized in the mesh as a connected polyline, up to floating-point tolerance.
\end{proof}

\theoreticalproperty{Triangle Quality Safeguards}\label{prop:triangle-quality}

To ensure well-shaped elements and numerical stability, SBMT enforces explicit geometric thresholds that bound triangle deformations during retriangulation.

\paragraph{Threshold Definitions}
Three control parameters are used:

\begin{itemize}
  \item \textbf{Vertex snapping threshold} $a$:
If a boundary point $p$ lies within distance $a$ of a mesh vertex $v$, we snap $v$ to $p$.
If multiple boundary points fall within the radius-$a$ neighborhood of $v$, we select the closest one; ties are broken deterministically (e.g., by a fixed ordering along the polyline), so the snapping map is single-valued and reproducible.
  
  \item \textbf{Edge elimination threshold} $b$: If a boundary point lies within distance $b$ of a triangle edge, the edge is collapsed and the resulting quadrilateral is retriangulated using the boundary point.
  
  \item \textbf{Foot projection threshold} $c$: If a triangle vertex lies within distance $c$ of a boundary segment, it is displaced perpendicularly to achieve exactly distance $c$, preserving its sidedness with respect to the segment.
\end{itemize}

\paragraph{Admissible Range Constraints.}
To prevent quality degradation, we require
\begin{equation}
  b < \frac{a}{2}, \qquad
  c < \frac{a}{\sqrt{2}}, \qquad
  e < \ell_{\min}^{\text{boundary}},
  \label{eq:admissible-range}
\end{equation}
where $e>0$ is the edge length of the equilateral triangles in the base mesh,
and $\ell_{\min}^{\text{boundary}}$ is the minimal spacing between adjacent
boundary segment endpoints (typically one pixel).

\begin{remark}[Sufficient condition to avoid snapping collisions]
\label{rem:no-snapping-collision}
In the regular equilateral background mesh with edge length~$e$, any two distinct vertices are at least distance~$e$ apart.
If the snapping radius satisfies $a<e/2$, then the radius-$a$ vertex neighborhoods are disjoint. Consequently, each boundary point lies in the snapping neighborhood of at most one vertex, and Step~1 cannot map two distinct vertices to the same boundary sample.

All parameter settings used in our experiments satisfy~ $a<e/2$, so this degeneracy is excluded in practice. The edge-elimination rule (threshold $b$) is retained to handle other near-boundary degeneracies and to enforce the admissible boundary--mesh interaction patterns required by the lookup table.

\end{remark}

\paragraph{Quality Guarantees}
These constraints collectively enforce:

\begin{itemize}
  \item \textbf{No double edge collapse:} The bound $b < \frac{a}{2}$ ensures that two edges in the same triangle are not eliminated simultaneously, preserving angular quality.
  
  \item \textbf{Atomic snapping/repulsion:} The constraint $c < \frac{a}{\sqrt{2}}$ guarantees that each vertex is acted upon by at most one segment, ensuring deterministic, order-invariant behavior.

  \item \textbf{Localized interactions:} The combined threshold design ensures each triangle interacts with at most one boundary segment at a time, simplifying the retriangulation logic and preserving minimum angle and area bounds.
\end{itemize}

\theoreticalproperty{Quasi-Delaunay Characteristics}

Although SBMT does not explicitly enforce Delaunay criteria, its rule-based operations—\\vertex snapping, repulsion, and short-edge elimination—collectively promote spatial regularity, suppress sliver triangles, and separate nearby features. Together with the minimal-angle and minimal-area safeguards discussed above, these mechanisms emulate several core geometric benefits of Delaunay refinement in a purely local and parallelizable fashion. 

Empirically, as shown by the angle and aspect-ratio statistics in
Section~\ref{sec:comparative-eval}, the resulting meshes exhibit Delaunay-like regularity in most regions (in the empirical sense captured by our angle/aspect-ratio statistics): a dominant population of near-equilateral elements and very few highly distorted triangles, while avoiding the cost of global retriangulation. This supports interpolation stability and geometric fidelity without requiring a fully Delaunay mesh.

\theoreticalproperty{Geometric Guarantees of SBMT}
SBMT ensures high-quality triangulation through strict geometric bounds. Two critical guarantees are the minimal angle and minimal area bounds of the generated sub-triangles.

\begin{theorem}[Minimal Angle Bound in SBMT]
\label{thm:min-angle-bound-main}
Assume that the geometric parameters $a,b,c$ and the base edge length $e$
satisfy the admissible range constraints~\eqref{eq:admissible-range}.
Then every sub-triangle generated through SBMT retriangulation has a minimal
interior angle bounded below by
\[
\theta_{\min} >
\min\left\{
  \arctan\left( \frac{b}{e + a - \sqrt{a^2 - b^2}} \right),\;
  \arctan\left( \frac{c}{e + a} \right)
\right\}.
\]
\end{theorem}

\begin{theorem}[Minimal Area Bound in SBMT]
\label{thm:min-area-bound-main}
Let \(b\) and \(c\) denote the edge elimination and foot projection thresholds, respectively. Then every sub-triangle generated through SBMT retriangulation satisfies the minimal area constraint:
\[
A_{\min} > \frac{1}{2} b \cdot c.
\]

\end{theorem}

The detailed proofs of Theorem~\ref{thm:min-angle-bound-main} and Theorem~\ref{thm:min-area-bound-main} are provided in Appendix D of the Supplementary Material.

\subsection{Interpolation Fidelity for Analytic and Effectively Bandlimited Signals}
\label{sec:interpolation-fidelity}

Bitmap images are discretized observations of physical signals, which are often effectively bandlimited. Most physical acquisition systems, such as optical lenses or CCD sensors, introduce inherent low-pass filtering, smoothing high-frequency components before digital sampling. To ensure high-fidelity interpolation, we establish a sampling condition for piecewise interpolation over equilateral triangular meshes.

\begin{theorem}[Spectrally Accurate Interpolation over Equilateral Meshes]
\label{thm:spectral_interp}
Let \( f \) be a compactly supported function that is effectively bandlimited with cutoff frequency \( \Omega \), and let \( \mathcal{T}_h \) be a structured triangular mesh composed of equilateral triangles of edge length \( h \). Then, if
\[
h < \frac{1.86}{\Omega},
\]
any piecewise polynomial interpolation \( \tilde{f} \) over \( \mathcal{T}_h \) (e.g., linear or quadratic) satisfies
\[
\| \tilde{f} - f \|_{L^\infty(\Omega)} < \varepsilon,
\]
where \( \varepsilon \) depends on the high-frequency tail of \( \hat{f} \) beyond \( \Omega \), and vanishes as \( \Omega \to \infty \) or \( h \to 0 \).
\end{theorem}

\emph{The proof and additional details for Theorem~\ref{thm:spectral_interp} can be found in Appendix E of the Supplementary Material.
}

To illustrate a practical effect of boundary-aware triangular interpolation in an illustrative setting,
Figure~\ref{fig:reconstruction_comparison} presents a three-stage visualization:

\begin{figure}[htbp]
\centering
\renewcommand{\arraystretch}{1.0}
\setlength{\tabcolsep}{2pt}
\begin{tabular}{ccc}
\includegraphics[width=0.31\linewidth]{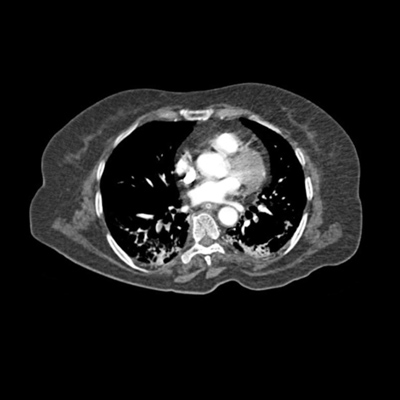} &
\includegraphics[width=0.31\linewidth]{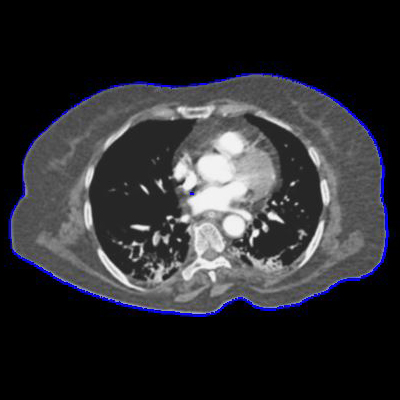} &
\includegraphics[width=0.31\linewidth]{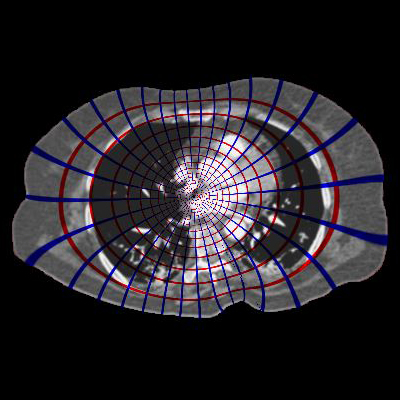} \\
\parbox{0.31\linewidth}{\centering\scriptsize (a) Original MR image} &
\parbox{0.31\linewidth}{\centering\scriptsize (b) Interpolated reconstruction} &
\parbox{0.31\linewidth}{\centering\scriptsize (c) Holomorphic chart ($u/v$ foliations)}
\end{tabular}
\caption{
Structure-aware interpolation on an SBMT mesh via a holomorphic chart.
(a) Original bitmap MR image.
(b) Interpolation on a globally uniform equilateral mesh that ignores anatomical boundaries.
(c) Gu--Yau holomorphic chart computed on the SBMT mesh: the harmonic potentials $(u,v)$ induce two orthogonal foliations, where red curves visualize equipotential lines ($u=\mathrm{const}$) and blue curves visualize trajectory/field lines ($v=\mathrm{const}$, integral curves aligned with $\nabla u$).
This panel visualizes the PDE-derived coordinate field (analogous to a 2D electrostatic potential/stream function), not a mesh-quality plot; curve crowding reflects conformal scaling in thin or high-curvature regions rather than remeshing artifacts.
All meshes use a fixed edge length $e=\sqrt{0.7}$ for a fair comparison. Rendering uses OpenGL; brightness variations are due to lighting.
}
\label{fig:reconstruction_comparison}
\end{figure}

\paragraph{Holomorphic 1-form chart for structure-aware interpolation}
\label{sec:holo-recon}
In the third panel of Figure~\ref{fig:reconstruction_comparison}, the embedded boundaries are aligned with gradient trajectories induced by a holomorphic $1$-form, following the conformal parameterization framework of Gu and Yau~\cite{Gu2002ComputingCS}.
Concretely, once the domain has been retriangulated by SBMT, we solve for a pair of harmonic conjugate potentials $(u,v)$ on the SBMT mesh via discrete Hodge decomposition (equivalently, a small number of sparse Laplace/Poisson systems).
The resulting holomorphic chart $(u,v)$ defines an orthogonal foliation whose level sets provide geometry-aware coordinates for interpolation and significantly reduce distortion near discontinuities.
Throughout this section, we use this holomorphic $1$-form chart as a generic elliptic PDE benchmark on SBMT-generated meshes and as a structure-aware interpolation tool, rather than as a core component of the SBMT algorithm itself.

\subsection{Parallel Execution Architecture}
\label{sec:parallel-execution}

Having established the structural and geometric guarantees of SBMT, 
We now formalize the parallel decomposition of SBMT’s geometric preprocessing and retriangulation stages, and demonstrate that the algorithm admits a race-free, scalable, and fully deterministic parallel implementation.

\begin{proposition}[Safe Decoupling of Preprocessing Steps]
\label{prop:decoupled-preprocessing}
Step~1 performs \emph{vertex snapping} based on a proximity threshold~\(a\); 
Step~2 enforces \emph{vertex repulsion} using a minimal distance threshold~\(c\); 
and Step~3 applies \emph{edge elimination} to remove short segments below a length 
threshold~\(b\) (see Figure~\ref{fig:b_threshold_split}). 
The execution dependency is summarized in Figure~\ref{fig:preprocess-dependency}.

Under the admissible range constraints~\eqref{eq:admissible-range}
(on $a,b,c,e$, and in particular $b < \tfrac{a}{2}$ and
$c < \tfrac{a}{\sqrt{2}}$), the following properties hold:
\begin{enumerate}
  \item[(i)] \textbf{Step 1 and Step 2 commute}: Both are local pointwise vertex updates. 
  The spatial threshold $c < \tfrac{a}{\sqrt{2}}$ ensures their influence zones are disjoint, 
  allowing concurrent execution without interference.
  
  \item[(ii)] \textbf{Step 3 is scheduled after vertex updates}: Edge elimination requires
  finalized vertex locations and must be executed after Steps~1--2
  (see Figure~\ref{fig:preprocess-dependency}).

  \item[(iii)] \textbf{Edge deletions are conflict-free}: The constraint
  $b < \tfrac{a}{2}$ ensures that no triangle has more than one edge
  flagged for deletion, preserving retriangulation consistency.
\end{enumerate}
\end{proposition}

\begin{proof}[Sketch]
Steps 1 and 2 operate on spatially isolated vertices under the given thresholds. No shared influence implies commutativity and suitability for atomic parallel execution. Step 3 relies on updated vertex positions, and the geometry guarantees that each boundary point lies within the influence zone of at most one triangle edge. Thus, no triangle receives conflicting deletion signals, avoiding retriangulation conflicts.
\end{proof}

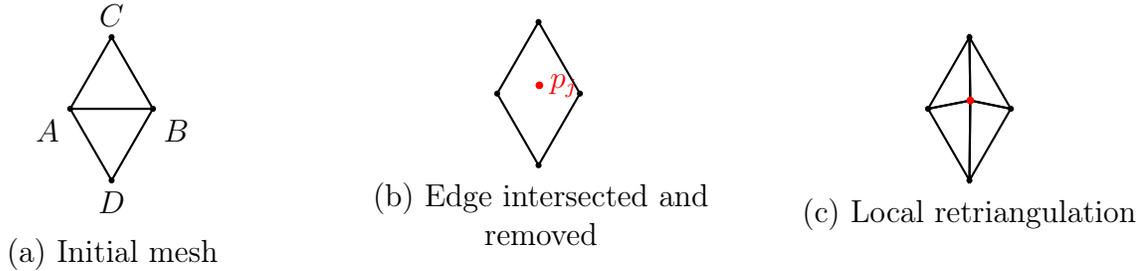
\begin{figure}[H]
\centering

% === 图 (a) 初始网格 ===
\begin{minipage}{0.31\linewidth}
\centering
\begin{tikzpicture}[scale=1.1]
\coordinate (A) at (0,0);
\coordinate (B) at (1,0);
\coordinate (C) at (0.5,0.866);
\coordinate (D) at (0.5,-0.866);
\draw[thick] (A) -- (B) -- (C) -- cycle;
\draw[thick] (A) -- (B) -- (D) -- cycle;
\draw[dashed] (A) -- (B);
\filldraw[black] (A) circle (0.03) node[below left] {$A$};
\filldraw[black] (B) circle (0.03) node[below right] {$B$};
\filldraw[black] (C) circle (0.03) node[above] {$C$};
\filldraw[black] (D) circle (0.03) node[below] {$D$};
\end{tikzpicture}
\\\centering (a) Initial mesh
\end{minipage}
\hfill
% === 图 (b) 删除边 + 插入点 ===
\begin{minipage}{0.31\linewidth}
\centering
\begin{tikzpicture}[scale=1.1]
\coordinate (A) at (0,0);
\coordinate (B) at (1,0);
\coordinate (C) at (0.5,0.866);
\coordinate (D) at (0.5,-0.866);
\coordinate (P) at (0.51,0.1);
\draw[thick] (A) -- (C) -- (B);
\draw[thick] (A) -- (D) -- (B);
\filldraw[black] (A) circle (0.03);
\filldraw[black] (B) circle (0.03);
\filldraw[black] (C) circle (0.03);
\filldraw[black] (D) circle (0.03);
\filldraw[red] (P) circle (0.04) node[right] {$p_j$};
\end{tikzpicture}
\\\centering (b) Edge intersected and removed
\end{minipage}
\hfill
% === 图 (c) 重剖分 ===
\begin{minipage}{0.31\linewidth}
\centering
\begin{tikzpicture}[scale=1.1]
\coordinate (A) at (0,0);
\coordinate (B) at (1,0);
\coordinate (C) at (0.5,0.866);
\coordinate (D) at (0.5,-0.866);
\coordinate (P) at (0.51,0.1);
\draw[thick] (A) -- (P) -- (C) -- cycle;
\draw[thick] (B) -- (P) -- (C) -- cycle;
\draw[thick] (A) -- (P) -- (D) -- cycle;
\draw[thick] (B) -- (P) -- (D) -- cycle;
\filldraw[black] (A) circle (0.03);
\filldraw[black] (B) circle (0.03);
\filldraw[black] (C) circle (0.03);
\filldraw[black] (D) circle (0.03);
\filldraw[red] (P) circle (0.04);
\end{tikzpicture}
\\\centering (c) Local retriangulation
\end{minipage}

\vspace{0.5em}
\caption{
Illustration of the threshold-$b$ retriangulation rule. 
(a) Two triangles share an edge $(A,B)$. 
(b) A boundary point $p_j$ lies within distance $b$ of the shared edge and triggers edge deletion.
(c) The local region is retriangulated using $p_j$ as an internal vertex, resulting in four sub-triangles.
}
\label{fig:b_threshold_split}
\end{figure}

\begin{algorithm}[htb]
\caption{SBMT geometric preprocessing pipeline (Steps 1--3)}
\label{alg:sbmt-preprocess}
\begin{algorithmic}[1]
\REQUIRE Base equilateral mesh $\mathcal{T}_0$; boundary segment set $\mathcal{B}$;
         thresholds $a,b,c$ with $b < a/2$ and $c < a/\sqrt{2}$
\ENSURE Preprocessed mesh $\mathcal{T}$

\STATE $\mathcal{T} \leftarrow \mathcal{T}_0$

\STATE \textbf{// Step 1: vertex snapping (threshold $a$)}
\FORALL{vertices $v \in \mathcal{T}$ \textbf{in parallel}}
  \STATE Find closest boundary point $p \in \mathcal{B}$
  \IF{$\|v-p\| < a$}
    \STATE Move $v$ to $p$
  \ENDIF
\ENDFOR

\STATE \textbf{// Step 2: vertex repulsion (threshold $c$)}
\FORALL{vertices $v \in \mathcal{T}$ not lying on $\mathcal{B}$ \textbf{in parallel}}
  \STATE Find closest boundary segment $s \in \mathcal{B}$ and signed distance $d(v,s)$
  \IF{$|d(v,s)| < c$}
    \STATE Move $v$ along the normal of $s$ so that $|d(v,s)| = c$
          and $v$ stays on the same side of $s$
  \ENDIF
\ENDFOR

\State \textbf{// Step 3: Edge Elimination and Retriangulation (threshold $b$)} \Comment{Order-dependent but locally scoped}
\FORALL{boundary points $p_j \in \mathcal{L}$}
  \STATE Find all triangle edges $e_k$ such that $\text{dist}(p_j, e_k) < b$
  \FORALL{such edges $e_k$}
    \STATE Remove $e_k$ and affected triangles from $\mathcal{T}$
    \STATE Insert $p_j$ and retriangulate the resulting cavity
  \ENDFOR
\ENDFOR

\STATE \textbf{return} $\mathcal{T}$

\end{algorithmic}
\end{algorithm}

% in preamble:
% \usetikzlibrary{positioning,calc,fit,arrows.meta,shapes.geometric}

\begin{figure}[t]
\centering
\begin{tikzpicture}[
  >=Latex,
  box/.style={
    draw, rounded corners=2pt, align=center,
    font=\small, inner sep=4pt, text width=3.7cm,
    minimum height=1.15cm
  },
  group/.style={draw, dashed, rounded corners=3pt, inner sep=6pt},
  merge/.style={diamond, draw, aspect=2.2, inner sep=0.8pt, minimum size=5.5mm},
  note/.style={font=\small, align=center}
]

% --- Step 1 & Step 2 (parallel block) ---
\node[box] (s1) {Step 1\\Vertex snapping\\(threshold $a$)};
\node[box, right=12mm of s1] (s2) {Step 2\\Foot-based repulsion\\(threshold $c$)};

% dashed grouping + label
\node[group, fit=(s1)(s2)] (g12) {};
\node[note, above=2.5mm of g12.north] {$\parallel$ \; commute / parallelizable};

% --- merge (join) node ---
\node[merge, below=9mm of g12.south] (join) {};

% --- Step 3 ---
\node[box, below=9mm of join] (s3)
  {Step 3\\Edge deletion + local cavity\\retriangulation (threshold $b$)};

% --- arrows: two branches -> join -> step3 ---
\draw[->] (s1.south) -- ++(0,-2.5mm) -| (join.west);
\draw[->] (s2.south) -- ++(0,-2.5mm) -| (join.east);
\draw[->] (join) -- (s3);

% note under step3
\node[note, below=3.5mm of s3.south]
{Step 3 scheduled after vertex updates; conflict-free under $b<\tfrac{a}{2}$.};

\end{tikzpicture}
\caption{Execution dependency of the three preprocessing steps in Proposition~\ref{prop:decoupled-preprocessing}.}
\label{fig:preprocess-dependency}
\end{figure}

\begin{corollary}[Pipelined Parallel Scheduling]
\label{cor:parallel-architecture}
Under the admissible range constraints~\eqref{eq:admissible-range}
(in particular $b < \tfrac{a}{2}$ and $c < \tfrac{a}{\sqrt{2}}$),
the SBMT preprocessing and retriangulation admit the following
multi-stage parallel execution regions:

\begin{itemize}
  \item \textbf{Region I (vertex updates)}: Steps~1 (vertex snapping) and~2
  (vertex repulsion) are executed concurrently as local pointwise vertex
  updates. No fine-grained synchronization is required within each step,
  and a single barrier finalizes the vertex locations.

  \item \textbf{Region II (edge elimination)}: Step~3 (edge elimination) is
  executed \emph{after} Steps~1--2, since it requires finalized vertex
  locations (cf.\ Proposition~\ref{prop:decoupled-preprocessing}(ii) and
  Figure~\ref{fig:preprocess-dependency}). Under $b < \tfrac{a}{2}$, edge
  deletions are conflict-free in the sense of
  Proposition~\ref{prop:decoupled-preprocessing}(iii).

  \item \textbf{Region III (template retriangulation)}: After preprocessing,
  each triangle is retriangulated independently using only its local
  intersection class and a lookup template (Algorithm~\ref{alg:sbmt-template}
  and Appendix~G), enabling triangle-wise parallel execution.
\end{itemize}

Consequently, Regions~I--III can be scheduled in a pipelined (multi-stage)
fashion with a barrier between Region~I and Region~II, and the overall
execution is race-free and globally consistent.
\end{corollary}

\begin{proof}[Sketch]
Proposition~\ref{prop:decoupled-preprocessing}(i) implies that Steps~1 and~2
commute and may be executed concurrently as local pointwise vertex updates,
with a barrier to finalize vertex locations. Proposition~\ref{prop:decoupled-preprocessing}(ii)
requires Step~3 to be scheduled after Steps~1--2. Finally,
Proposition~\ref{prop:decoupled-preprocessing}(iii) guarantees conflict-free
edge deletions under $b < \tfrac{a}{2}$, and the subsequent retriangulation
is triangle-wise and lookup-driven given the local intersection class.
\end{proof}

\begin{corollary}[Structural Basis for Parallel Retriangulation (Lookup-driven)]
\label{thm:parallel-structure}
Assume the frozen-registry and single-valued lookup setting in
Assumption~\ref{assump:frozen-lookup}, and the edge-consistency guarantee of
Lemma~\ref{lemma:precision}. Let $\mathcal{T}$ be the preprocessed base mesh, and
for each $K\in\mathcal{T}$ let $P_K=\mathcal{L}\big(\kappa(K)\big)$ be the unique
lookup patch.

Then lookup-driven retriangulation is \emph{embarrassingly parallel}:
all patches $\{P_K\}_{K\in\mathcal{T}}$ can be generated out-of-place in arbitrary order
(or in parallel), followed by a single deterministic stitching stage. Consequently, the
assembled global mesh is independent of the execution schedule, up to the fixed
symmetry-breaking/tolerance conventions.
\end{corollary}

\begin{proof}[Sketch]
By Assumption~\ref{assump:frozen-lookup}, each patch $P_K$ is a deterministic function
of the frozen registry $\mathcal{R}$ and the local key $\kappa(K)$, and is generated
out-of-place; hence any execution order produces the same collection of patches
$\{P_K\}_{K\in\mathcal{T}}$. By Lemma~\ref{lemma:precision}, adjacent triangles induce
identical edge traces on shared base edges, so stitching is deterministic. Therefore the
final assembled mesh is schedule-independent.
\end{proof}

\paragraph{Scalability Potential}
Proposition~\ref{prop:decoupled-preprocessing} and Theorem~\ref{thm:parallel-structure} together imply a scalable parallel schedule: preprocessing is stage-separated, and triangle-wise retriangulation admits independent (disjoint-support) updates. This aligns SBMT with per-element parallel triangulation paradigms~\cite{blelloch1999design, spielman2007parallel} while retaining a deterministic, lookup-driven pipeline.

\subsection{Runtime complexity}
\label{sec:complexity}

A formal analysis of the SBMT pipeline’s time complexity—covering preprocessing, template-based retriangulation, and signal interpolation—is provided in Appendix C of the Supplementary Material. The overall runtime scales as $\mathcal{O}(n \log N + nm + N_t + N)$, where \(n\) be the number of input boundary segments, \(N\) the number of base-mesh triangles, \(N_t\) the number of triangles intersected by boundaries, and \(m\) the average number of triangles influenced by a segment (typically constant due to local thresholds).

\section{Numerical Results}
\label{sec:experiments}

We evaluate Structured Bitmap-to-Mesh Triangulation (SBMT) on synthetic and
real-world binary images with rasterized domains and pixel-level boundaries. We
compare against standard constrained Delaunay triangulation (CDT) baselines,
including Shewchuk’s Triangle~\cite{Shewchuk1996Triangle} and the CDT module in
Gmsh, focusing on boundary-conforming straight-edge meshes suitable for downstream
PDE discretization.

Our evaluation addresses four aspects:
\begin{itemize}
  \item \textbf{Boundary conformity:} exact embedding of the prescribed polygonal boundary;
  \item \textbf{Geometric quality:} minimum-angle distributions, aspect-ratio statistics, and sliver counts (angle $<5^\circ$);
  \item \textbf{Robustness:} behavior on sharp corners, narrow necks, and near-singular junctions;
  \item \textbf{Compactness:} element counts and unnecessary refinement relative to CDT-style pipelines.
\end{itemize}

\paragraph{Parameter settings}
Unless otherwise noted, we use fixed preprocessing thresholds
$a=0.26$, $b=0.125$, and $c=0.183$, which satisfy the sufficient safety conditions
in Section~\ref{subsubsec:theoretical-properties} and yield stable behavior across
all tested domains. These parameters control the three preprocessing steps in
Algorithm~\ref{alg:sbmt-preprocess} (snapping, repulsion, and short-edge elimination)
and serve only to enforce admissible boundary--mesh configurations so that the
subsequent template application in Section~\ref{sec:main} remains deterministic and
conflict-free. A focused sensitivity study around $(a,b,c)$ is reported in Appendix~I
(Supplementary Material), and an ablation study removing individual rules is given
in Section~\ref{subsec:ablation}.

\paragraph{Remark (fixed vs.\ adaptive thresholds)}
We keep $(a,b,c)$ globally fixed to preserve SBMT’s core design goal: a fully
lookup-driven, stateless pipeline without runtime refinement or synchronization.
Adaptive choices (e.g., spatially varying $a(x),b(x),c(x)$ based on local boundary
complexity or image evidence) are a promising direction but would require additional
indicators and safeguards to preserve symbolic closure and conflict-free updates.

Accordingly, we first validate boundary embedding qualitatively (Section~\ref{sec:case-coverage}),
then quantify mesh quality (Section~\ref{subsec:geom-quality}), isolate the role of each preprocessing
rule via ablation (Section~\ref{subsec:ablation}), and finally report three-way CDT baselines on
benchmark domains (Section~\ref{sec:comparative-eval}).

\subsection{Case Coverage and Visual Validation}
\label{sec:case-coverage}

This section validates the structure-preserving capabilities of SBMT through visual and symbolic evidence. We test both anatomical and synthetic domains, aiming to:
\begin{enumerate}[label=(\roman*)]
    \item confirm accurate boundary embedding;
    \item demonstrate local retriangulation stability across intersection types;
    \item verify lookup table activation coverage in real cases;
    \item illustrate that SBMT meshes are compatible with elliptic (Hodge/Poisson) solvers.
\end{enumerate}

In all examples below, the reconstruction step reuses the holomorphic $1$-form pipeline described after Figure~\ref{fig:reconstruction_comparison} (Section~\ref{sec:holo-recon}): SBMT first produces a boundary-conforming equilateral mesh, and then a holomorphic chart is computed on this mesh via discrete Hodge decomposition~\cite{Gu2002ComputingCS} to guide structure-aware interpolation. The focus here is on how SBMT remeshing behaves under this generic PDE workload, rather than on the holomorphic solver itself.

\subsubsection*{(1) Real-world case: stomach anatomy}

To illustrate the visual significance of structure-aware retriangulation, we apply our pipeline combining SBMT and holomorphic reconstruction to a medical binary image of a stomach cross-section, shown in Figure~\ref{fig:structure_aware_reconstruction}.
This shape features intricate, high-curvature boundaries and narrow lobes, posing significant challenges to boundary-conforming remeshing and PDE-based reconstruction.
Importantly, panel~(c) visualizes the \emph{holomorphic-chart foliations} $(u,v)$ computed on the SBMT mesh via the Gu--Yau pipeline: the red curves are equipotential lines $u=\mathrm{const}$ and the blue curves are the orthogonal trajectory/field lines $v=\mathrm{const}$ (analogous to a 2D electrostatic potential and its field lines).
Apparent bending or crowding of trajectories in thin, high-curvature distal regions reflects the chart's conformal scaling under the chosen boundary conditions, rather than remeshing error.
Mesh quality near constrained interfaces is evaluated separately via the local boundary ROI triangulation comparisons (Figure~\ref{fig:local_triangulation_comparison}) and the minimum-angle distribution plots (Figure~\ref{fig:minAngle_hist}).

\begin{figure}[htbp]
\centering
\renewcommand{\arraystretch}{1.0}
\setlength{\tabcolsep}{2pt}
\begin{tabular}{ccc}
\includegraphics[width=0.31\linewidth]{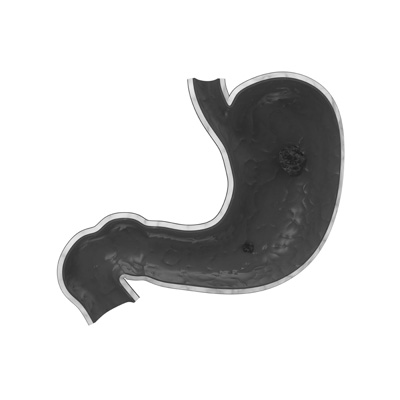} &
\includegraphics[width=0.31\linewidth]{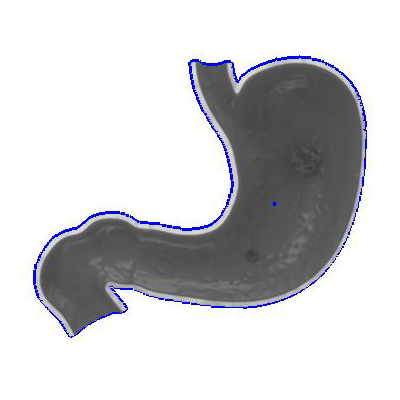} &
\includegraphics[width=0.31\linewidth]{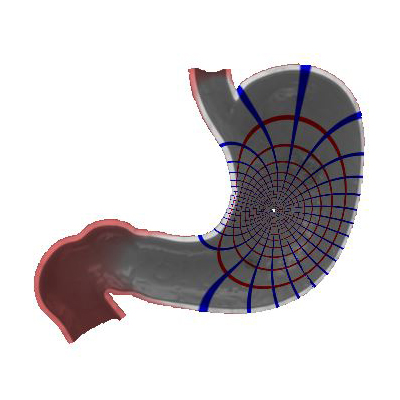} \\
\parbox{0.31\linewidth}{\centering\scriptsize (a) Original binary image} &
\parbox{0.31\linewidth}{\centering\scriptsize (b) Uniform mesh interpolation} &
\parbox{0.31\linewidth}{\centering\scriptsize (c) Holomorphic chart ($u/v$ foliations)}
\end{tabular}
\caption{
Structure-aware processing of a stomach cross-section.
(a) Original binary image.
(b) Interpolation on a uniform mesh without respecting anatomical boundaries.
(c) Holomorphic chart on the SBMT mesh (Gu--Yau pipeline; Section~\ref{sec:holo-recon}): we solve a discrete Poisson/Hodge system to obtain a holomorphic $1$-form and integrate it to produce harmonic conjugate potentials $(u,v)$.
Red curves show equipotential lines ($u=\mathrm{const}$) and blue curves show trajectory/field lines ($v=\mathrm{const}$), analogous to a 2D electrostatic potential and its orthogonal field lines.
This visualization reflects the conformal coordinate field; apparent bending/crowding near the distal thin regions is an expected conformal scale effect under the chosen boundary conditions, and should not be interpreted as remeshing error.
Mesh edges are omitted for clarity.
}
\label{fig:structure_aware_reconstruction}
\end{figure}

To further validate the completeness and practical activation of the SBMT lookup table, 
we visualize the retriangulation output on a complex boundary region of the stomach image, 
using color coding to indicate the local intersection type of each triangle. 
As shown in Figure~\ref{fig:zoomed_region}, 
triangles intersected by a single boundary segment are filled in green, 
those matched to the $(1,1)$ type in pink.

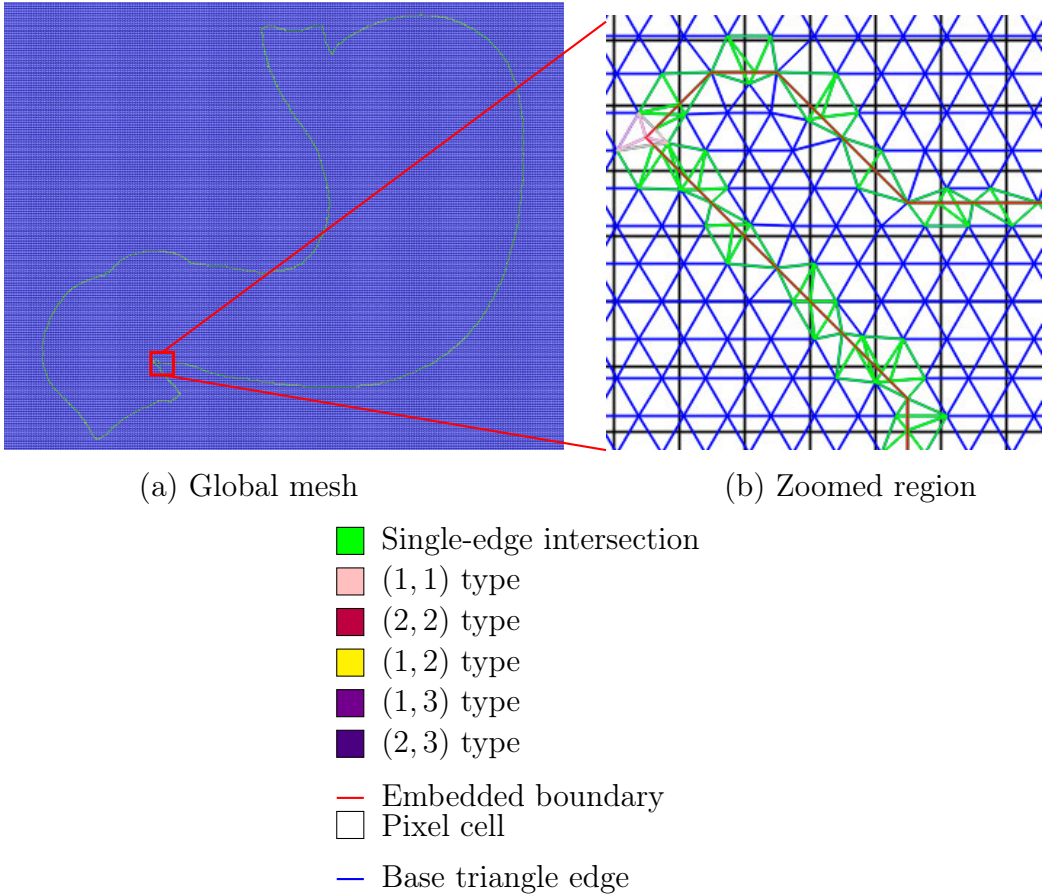
\begin{figure}[ht]
\centering
\begin{tikzpicture}

% ==== 左侧：原图 ====
\node[anchor=south west, inner sep=0] (full) at (0,0) {\includegraphics[width=0.45\linewidth]{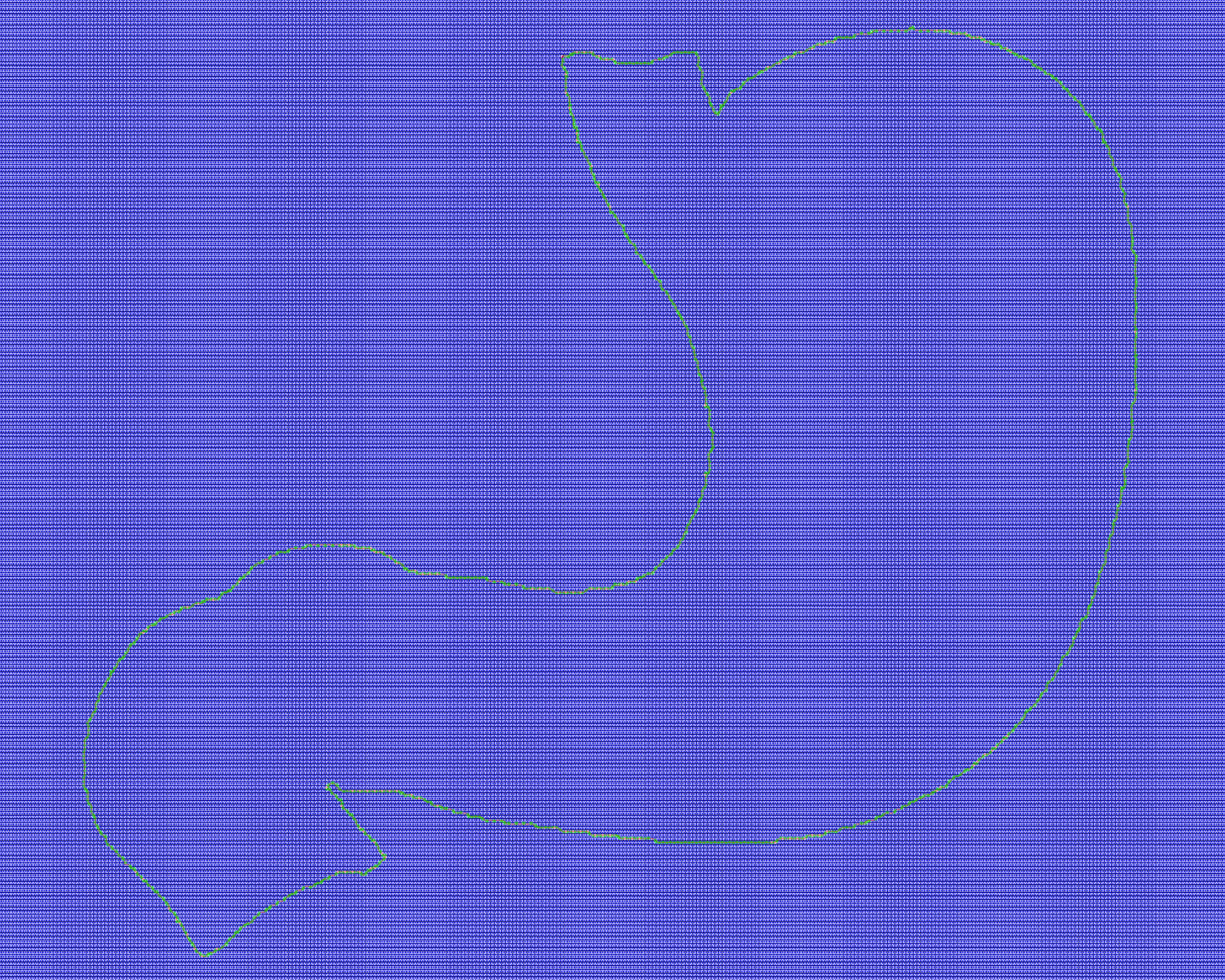}};

% ==== 右侧：放大图 ====
\node[anchor=south west, inner sep=0] (zoom) at (8,0) {\includegraphics[width=0.35\linewidth]{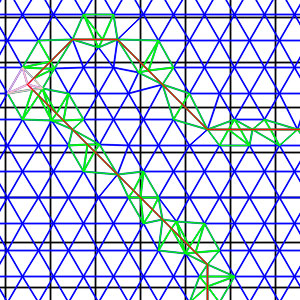}};

% ==== 红色矩形框（选区）====
\draw[red, very thick] (1.95,1.0) rectangle (2.25,1.3);  

% ==== 连线：从框角连到放大图 ====
\draw[red, thick] (2.10,1.3) -- (8,5.7);  
\draw[red, thick] (2.10,1.0) -- (8,0);  

% ==== 子图标题 ====
\node at (3.25,-0.5) {(a) Global mesh};
\node at (11.25,-0.5) {(b) Zoomed region};

\end{tikzpicture}

\begin{tikzpicture}[scale=0.9]

% 第一列 - 颜色块 + 说明
\draw[fill=green] (0,0) rectangle (0.4,0.4);
\node[right] at (0.5,0.2) {Single-edge intersection};

\draw[fill=pink] (0,-0.6) rectangle (0.4,-0.2);
\node[right] at (0.5,-0.4) {$(1,1)$ type};

\draw[fill=purple] (0,-1.2) rectangle (0.4,-0.8);
\node[right] at (0.5,-1.0) {$(2,2)$ type};

\draw[fill=yellow] (0,-1.8) rectangle (0.4,-1.4);
\node[right] at (0.5,-1.6) {$(1,2)$ type};

\draw[fill=blue!40!purple] (0,-2.4) rectangle (0.4,-2.0);
\node[right] at (0.5,-2.2) {$(1,3)$ type};

\draw[fill=indigo] (0,-3.0) rectangle (0.4,-2.6);
\node[right] at (0.5,-2.8) {$(2,3)$ type};

% 第二列 - 线条示意
\draw[red, thick] (0,-3.6) -- (0.4,-3.6);
\node[right] at (0.5,-3.6) {Embedded boundary};

\draw[black] (0,-4.2) rectangle (0.4,-3.8);
\node[right] at (0.5,-4.0) {Pixel cell};

\draw[blue, thick] (0,-4.8) -- (0.4,-4.8);
\node[right] at (0.5,-4.8) {Base triangle edge};

\end{tikzpicture}

\caption{
SBMT lookup table coverage on a complex anatomical boundary.
Dominant cases such as single-edge and $(1,1)$ types are highlighted, confirming the sufficiency of lookup table coverage in practice.
The clean triangulation and absence of conflicts illustrate the spatial coherence of the local retriangulation strategy.
}
\label{fig:zoomed_region}
\end{figure}

These two categories are the most frequently encountered in real-world data, 
accounting for the majority of retriangulation cases. 
More exotic types, such as $(2,2)$, $(2,3)$ and $(3,3)$, are rare and do not appear in this zoomed region, 
but are supported in the implementation and handled correctly when they occur 
(see Appendix~G of the Supplementary Material).

To aid visual interpretation, the background shows the underlying bitmap structure.
Each black square corresponds to a unit pixel, with its center representing the exact pixel coordinate. 
The blue triangles indicate the overlaid regular triangular mesh that defines the initial domain tiling.
Red lines mark the embedded polygonal boundary, which intersects the mesh and triggers local retriangulation.

This visualization confirms not only the coverage of key lookup table entries 
but also the spatial consistency of the local retriangulation strategy of SBMT. 
Each triangle is processed independently based on its type, 
resulting in a clean and well-organized triangulation despite the underlying geometric complexity.

\begin{figure}[htbp]
\centering
\renewcommand{\arraystretch}{1.0}
\setlength{\tabcolsep}{2pt}
\begin{tabular}{ccc}
\includegraphics[width=0.31\linewidth]{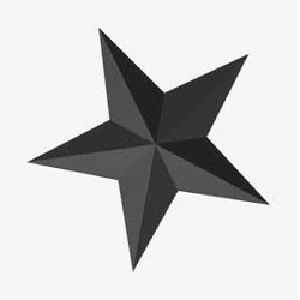} &
\includegraphics[width=0.31\linewidth]{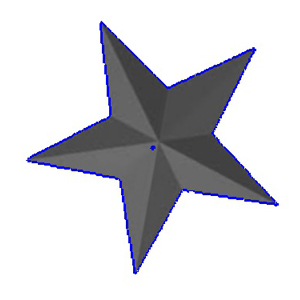} &
\includegraphics[width=0.31\linewidth]{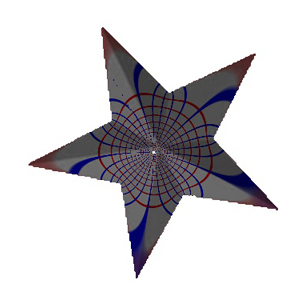} \\
\small (a) Original binary image &
\small (b) Uniform-mesh interpolation &
\small (c) Holomorphic chart
\end{tabular}
\caption{
Structure-aware representation of a synthetic star-shaped domain.
(a) Original binary image with sharp nonconvex boundaries.
(b) Bilinear interpolation on a uniform triangular mesh, which smears the star tips and narrow lobes.
(c) Holomorphic $1$-form foliation computed on the SBMT mesh using the Gu--Yau algorithm~\cite{Gu2002ComputingCS}: the conjugate harmonic potentials $(u,v)$ define an orthogonal chart whose level sets (blue curves) align with the geometry and serve as structure-aware coordinates for subsequent reconstruction (Section~\ref{sec:holo-recon}).
}
\label{fig:structure_aware_star}
\end{figure}

\subsubsection*{(2) Synthetic case: star domain}

To complement the anatomical example, we also test the same SBMT+holomorphic reconstruction pipeline on a synthetic star-shaped domain with sharp inward corners. This example is designed to trigger more extreme geometric configurations, including acute angles and narrow junctions, which are rare in natural anatomical boundaries but crucial for verifying the completeness and
robustness of the remeshing system.

Figure~\ref{fig:structure_aware_star} illustrates the overall setup: panel~(a) shows the original binary star image, panel~(b) demonstrates that bilinear interpolation on a uniform triangular mesh smears the star tips and narrow lobes, and panel~(c) visualizes the holomorphic $1$-form foliation computed on the SBMT mesh. The conjugate harmonic potentials $(u,v)$ define an
orthogonal chart whose level sets provide structure-aware coordinates for the subsequent reconstruction (see Section~\ref{sec:holo-recon}).

To better illustrate the geometric fidelity of the underlying retriangulation, we visualize the full triangulation and a selected high-curvature region of the synthetic star in Figure~\ref{fig:zoomed_star}. In the global view, the five-pointed shape and its sharp internal corners are accurately preserved by the structure-aware remeshing. The zoomed-in region shows clear boundary alignment, well-shaped triangles, and no sliver elements, despite the presence of acute features.

Different triangle types are color-coded using the same scheme as in Figure~\ref{fig:zoomed_region}. In this region, only $(1,1)$, $(2,2)$, and single-edge intersection types are activated, demonstrating that even in complex synthetic domains, most retriangulation cases fall into a few dominant categories. This supports the practicality of the lookup table in the SBMT
framework and highlights the local stability of its retriangulation rules.

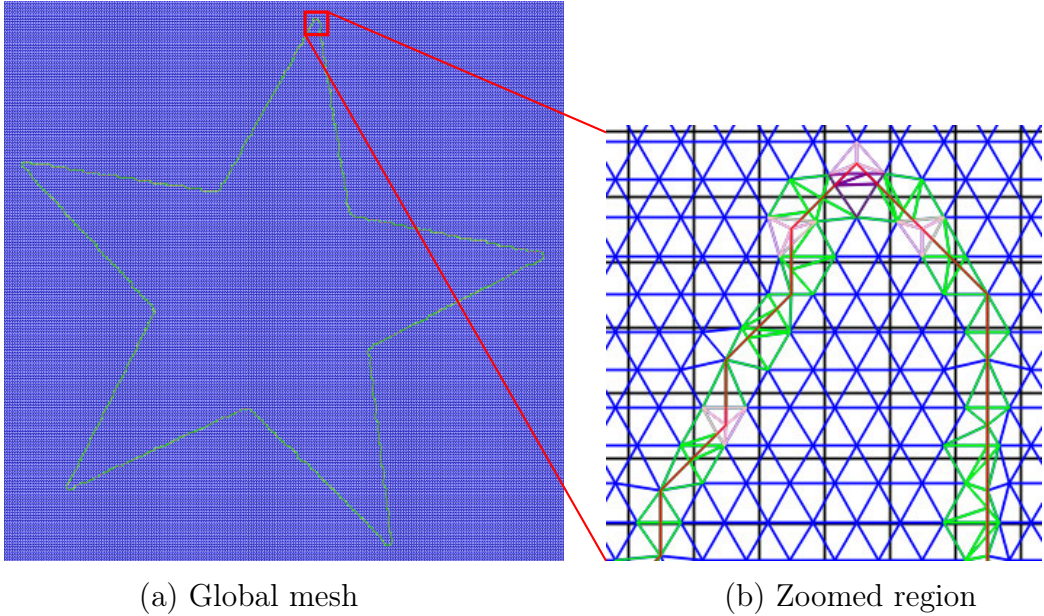
\begin{figure}[ht]
\centering
\begin{tikzpicture}

% ==== 左侧：原图 ====
\node[anchor=south west, inner sep=0] (full) at (0,0) {\includegraphics[width=0.45\linewidth]{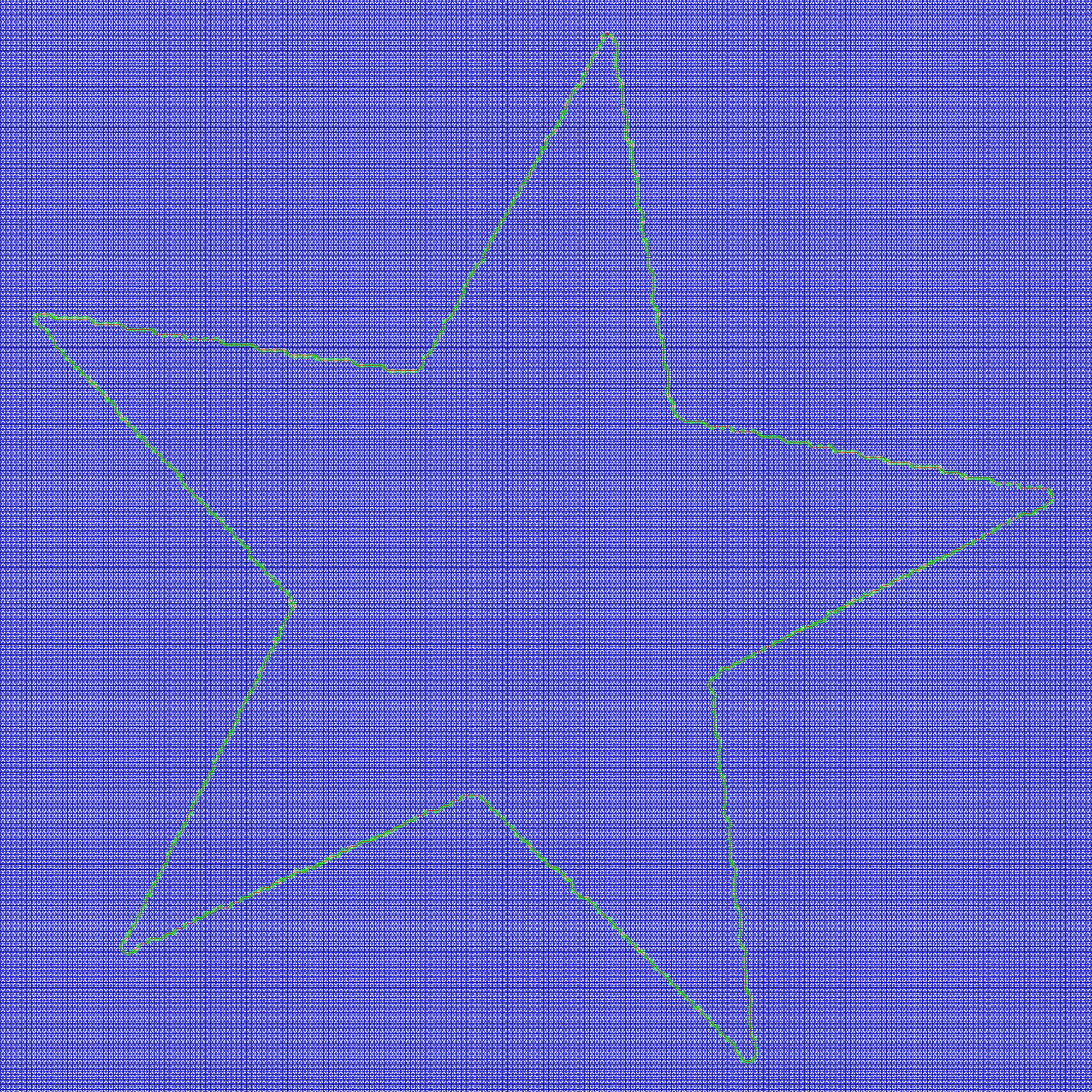}};

% ==== 右侧：放大图 ====
\node[anchor=south west, inner sep=0] (zoom) at (8,0) {\includegraphics[width=0.35\linewidth]{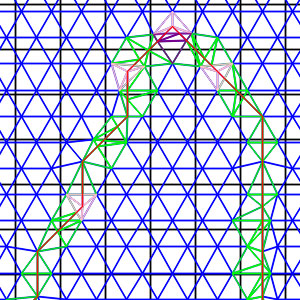}};

% ==== 红色矩形框（选区）====
\draw[red, very thick] (4.0,7.0) rectangle (4.3,7.3);  

% ==== 连线：从框角连到放大图 ====
\draw[red, thick] (4.3,7.3) -- (8,5.7);  
\draw[red, thick] (4.0,7.0) -- (8,0);  

% ==== 子图标题 ====
\node at (3.25,-0.5) {(a) Global mesh};
\node at (11.25,-0.5) {(b) Zoomed region};

\end{tikzpicture}

\caption{
Structure-aware remeshing of a synthetic star-shaped domain.
(a) Global view of the remeshed binary image. The red rectangle marks a high-curvature region near a sharp inward corner.
(b) Zoomed-in visualization of the selected region. Triangle elements remain well-shaped and align precisely with embedded boundaries.
Color coding indicates the local intersection type of each triangle, consistent with the legend in Figure~\ref{fig:zoomed_region}.
Black squares represent the underlying pixel grid, each of unit size and centered at the pixel coordinate. 
Blue lines indicate the initial regular triangular mesh, and red polylines denote embedded boundary segments.
Only $(1,1)$, $(2,2)$, and single-edge intersection types appear in this zoomed region; rarer types did not occur here.
}
\label{fig:zoomed_star}
\end{figure}

\subsection{Geometric Quality Evaluation}
\label{subsec:geom-quality}

We first assess the intrinsic geometric quality of SBMT meshes, independently of any comparison baseline. Two standard indicators are used: the minimum interior angle over all elements and the minimum triangle area. These metrics diagnose sliver formation and near-degenerate cells, which are critical for the stability of downstream PDE solvers and interpolation schemes.

Table~\ref{tab:mesh_quality_sbmt} summarizes the mesh quality of SBMT on two representative domains: an anatomical stomach cross-section and a synthetic star-shaped region. Angles are reported in degrees; areas are in pixel units.

Both domains exhibit a strictly positive safety margin above the sliver threshold of $5^\circ$, and no triangle violates this criterion. Moreover, the vast majority of elements remain essentially equilateral. This is a direct consequence of SBMT’s design: the base grid is regular and isotropic, and the local retriangulation templates are constructed to minimally perturb the equilateral pattern while embedding the input boundary segments.

\begin{figure}[htbp]
\centering
\setlength{\tabcolsep}{4pt}
\begin{tabular}{cc}
\includegraphics[width=0.47\linewidth]{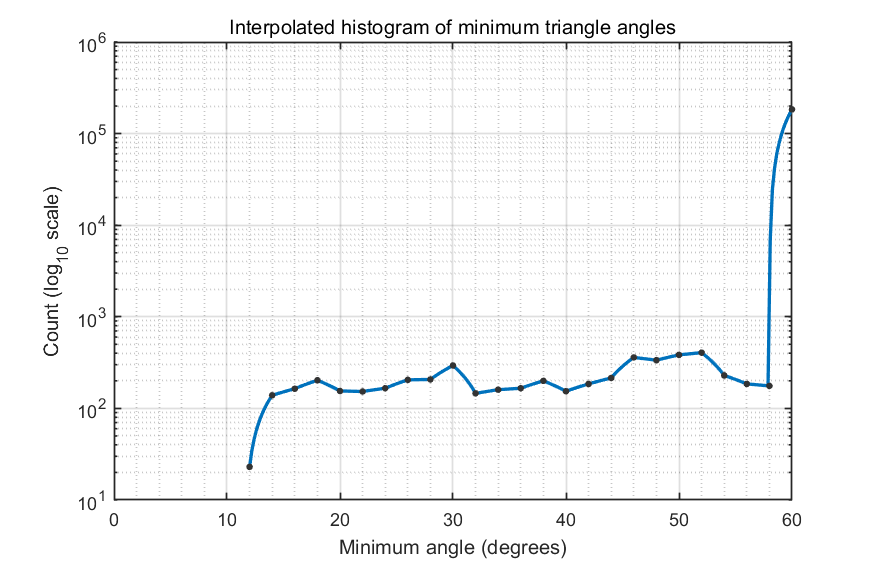} &
\includegraphics[width=0.47\linewidth]{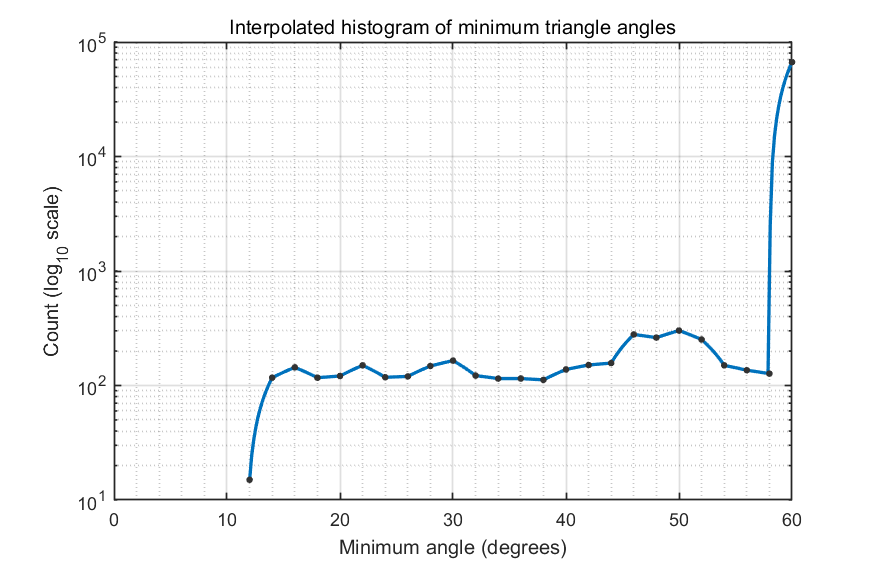} \\
\small (a) Stomach domain & \small (b) Star domain
\end{tabular}
\caption{
Smoothed histograms of minimum triangle angles for the stomach and star meshes.
Bin centers are spaced at $2^\circ$ intervals, and the vertical axis uses a logarithmic scale to reveal the low-probability tail.
In both cases, the distributions are sharply peaked near $60^\circ$, reflecting the dominance of equilateral elements. The low-angle tail remains well above the sliver threshold, with observed minima of $11.4^\circ$ for the stomach and $10.32^\circ$ for the star, in agreement with the analytic lower bound in Theorem~\ref{thm:min-angle-bound-main}.}
\label{fig:min_angle_hist}
\end{figure}

\paragraph{Theoretical comparison}
For the default thresholds $a = 0.26$, $b = 0.125$, $c = 0.183$, and triangle edge length $e = \sqrt{0.45}$, the geometric bounds in Property~\ref{prop:triangle-quality} give
\[
\begin{aligned}
\theta_{\min}
&> \min\Bigg\{
\arctan\!\left(\frac{b}{\,e + a - \sqrt{a^2 - b^2}\,}\right),
\;\arctan\!\left(\frac{c}{\,e + a\,}\right)
\Bigg\}
\;\approx\; 10.1^\circ,\\[0.3em]
A_{\min}
&>\; \tfrac{1}{2} b c
\;\approx\; 1.14\times 10^{-2}.
\end{aligned}
\]
The experimentally observed minima in Table~\ref{tab:mesh_quality_sbmt} are strictly above these thresholds, confirming that SBMT consistently produces meshes with strong geometric guarantees even in domains with high curvature or narrow features.

\begin{table}[htbp]
\centering
\caption{Geometric quality of SBMT meshes on two domains. 
Slivers are triangles with minimum angle $<5^\circ$.
Equilateral ratio is the fraction of triangles classified as equilateral by the edge-length test.}
\label{tab:mesh_quality_sbmt}
\begin{tabular}{lcccc}
\toprule
\textbf{Domain} & \textbf{Min angle} & \textbf{Min area} & \textbf{Slivers} & \textbf{Equilateral ratio} \\
\midrule
Stomach & $11.4^\circ$ & $0.0134$ & $0$ & $183{,}063 / 188{,}239 \approx 97.25\%$ \\
Star    & $10.32^\circ$ & $0.0133$ & $0$ & $66{,}540 / 70{,}292 \approx 94.66\%$ \\
\bottomrule
\end{tabular}
\end{table}

To complement these scalar indicators, Figure~\ref{fig:min_angle_hist} visualizes the full distributions of per-triangle minimum angles on both domains.

\paragraph{Resolution remark.}
SBMT is intentionally formulated in a \emph{bitmap-native} regime: the raster resolution is treated as the prescribed discretization of the domain, and---in the absence of any task-specific error indicator or weighting---pixels are regarded as an equal-weight sampling grid for downstream geometric and PDE
operators. Accordingly, throughout this section we fix a single-resolution equilateral base mesh to isolate the effect of the symbolic retriangulation and to maintain stable, uniform discrete differential stencils near boundaries.
This choice may lead to dense meshes even on simple shapes; however, SBMT targets \emph{geometry-faithful, PDE-ready discretization} rather than feature-adaptive compression or minimal-element approximation, trading special-case optimality for
generality, determinism, and predictable solver behavior.

A higher or lower sampling density can be obtained by changing the base edge length $e$ (or, equivalently, the upstream raster resolution) prior to meshing. Designing a fully adaptive, multi-level variant of SBMT (with refinement or coarsening away from the boundary) is orthogonal to our focus here and left as future work.

\subsection{Ablation Study on Threshold Rules}
\label{subsec:ablation}

To assess the contribution of each geometric preprocessing rule in SBMT, 
we ablate the three thresholds introduced in
Section~\ref{subsubsec:theoretical-properties} and
Algorithm~\ref{alg:sbmt-preprocess}, namely
\emph{vertex snapping} (threshold $a$), \emph{foot-based projection}
(threshold $c$), and \emph{edge elimination} (threshold $b$).

We then selectively disable individual rules while keeping the others active,
leading to the following five configurations:
\[
\textbf{E1:}~\text{Full (all rules on)},\quad
\textbf{E2:}~\text{No snapping},\quad
\textbf{E3:}~\text{No foot-based projection},
\]
\[
\textbf{E4:}~\text{No edge elimination},\quad
\textbf{E5:}~\text{All rules off}.
\]

Each configuration is evaluated on the stomach and star domains using three geometric indicators:
(i) minimum interior angle (in degrees),
(ii) minimum triangle area,
and (iii) the number of \emph{slivers}, defined as triangles whose smallest angle is below \(5^\circ\).

\paragraph{Findings}
The full pipeline (E1) produces sliver-free meshes with comfortable safety margins above the theoretical lower bounds in Theorems~\ref{thm:min-angle-bound-main} and~\ref{thm:min-area-bound-main}. 
Disabling snapping (E2) causes only mild degradation, mainly near tightly curved boundaries where vertices no longer align exactly with intersection points.

In contrast, removing repulsion (E3) or edge deletion (E4) leads to a sharp collapse in geometric quality. 
Minimum angles approach zero, the minimal areas drop by several orders of magnitude, and hundreds of sliver elements appear. 
This behavior is consistent with the geometric role of these rules:
repulsion prevents vertices from collapsing toward boundary segments, 
while edge deletion removes short, intersected edges that would otherwise generate highly obtuse or needle-like triangles.

\begin{table}[htbp]
\centering
\scriptsize
\setlength{\tabcolsep}{7pt}
\caption{Ablation of geometric preprocessing rules. 
E1 (full pipeline) consistently yields the highest mesh quality.}
\label{tab:ablation}
\renewcommand{\arraystretch}{1.15}
\setlength{\tabcolsep}{5pt}
\begin{tabular}{lccc|ccc}
\toprule
& \multicolumn{3}{c|}{\textbf{Stomach}} & \multicolumn{3}{c}{\textbf{Star}} \\
\textbf{Config} 
& Min angle & Min area & Slivers 
& Min angle & Min area & Slivers \\
\midrule
E1 – Full        
& \textbf{11.42} & \textbf{0.01338} & \textbf{0}
& \textbf{10.32} & \textbf{0.0133} & \textbf{0} \\
E2 – No snapping     
& 2.52  & 0.00636  & 1 
& 2.19  & 0.00104  & 4 \\
E3 – No repulsion  
& 0.0272  & $2\times 10^{-8}$ & 215 
& 0.05   & $1\times 10^{-7}$ & 190 \\
E4 – No deletion   
& 0.0331  & $4\times 10^{-5}$ & 87 
& 0.16    & $1.6\times 10^{-4}$ & 50 \\
E5 – All off      
& 0.0199  & $2\times 10^{-8}$ & 609 
& 0.001   & $1\times 10^{-7}$ & 508 \\
\bottomrule
\end{tabular}
\end{table}

When all preprocessing is disabled (E5), the mesh degenerates severely in both domains, 
with extremely small angles and nearly zero-area elements.
This confirms that the geometric filters are not cosmetic, but are essential to maintaining the integrity of the retriangulation, especially in regions with sharp or densely sampled boundaries.

Overall, the ablation study shows that the three threshold-based rules work in a strongly complementary fashion. 
Snapping improves alignment and robustness, while repulsion and edge deletion are critical for preventing pathological configurations and guaranteeing usable element quality for downstream numerical solvers.

\vspace{0.8em}
\subsection{Quantitative Comparison with Constrained Delaunay Triangulation}
\label{sec:comparative-eval}

We compare SBMT against two widely used constrained triangulation tools:
Shewchuk's Triangle~\cite{Shewchuk1996Triangle} and Gmsh.
For all three methods, the \emph{same prescribed polygonal boundary}
(PSLG extracted from the raster domain) is used as input.
Triangle is invoked with the option string
pq20a0.1948, i.e., a constrained Delaunay mesh with a
$20^\circ$ minimum-angle constraint and a maximum-area bound chosen to
keep the element scale comparable to that of SBMT.
For the Gmsh baseline, we use Gmsh~4.15.0 with the same closed PSLG
boundary and a uniform characteristic length
$l_c = e = \sqrt{0.45}$, matching the SBMT base-mesh scale.
The mesh is generated as a straight-edge planar surface using the
2D Delaunay algorithm (Mesh.Algorithm = 5); to avoid hidden
size adaptation, we fix the global size bounds to
Mesh.CharacteristicLengthMin = Mesh.CharacteristicLengthMax = $l_c$
and disable curvature- and boundary-extension-based size-field effects.

Three benchmark domains are used: a five-pointed star (sharp, nonconvex
corners), a bulged droplet (smooth, high-curvature boundary), and a
handwritten ``Y'' (branching junctions).
For each domain and method we report: minimum interior angle, minimum
area, variance of triangle areas, the total number of triangles, the
number of slivers (angle $<5^\circ$), and the number of triangles
classified as equilateral by the edge-length test.

\begin{table}[htbp]
\centering
\scriptsize
\setlength{\tabcolsep}{7pt}
\caption{Comparison between SBMT, Triangle (CDT), and Gmsh on three domains.
Sliver: triangle with minimum angle $<5^\circ$.
Equilateral: triangles classified as equilateral by the edge-length test.}
\label{tab:comparison-combined}
\begin{tabular}{llcccccc}
\toprule
\textbf{Domain} & \textbf{Method} & \textbf{Min Angle} & \textbf{Min Area} & \textbf{Area Var.} & \textbf{Count} & \textbf{Slivers} & \textbf{Equilateral} \\
\midrule
\multirow{3}{*}{Star} 
& SBMT     & $10.32^\circ$ & $0.0133$ & $0.000645$ & $70{,}292$  & $0$ & $66{,}540$ \\
& Triangle & $20.07^\circ$ & $0.0615$ & $0.00084$ & $107{,}097$ & $0$ & $3{,}438$ \\
& Gmsh     & $31.29^\circ$ & $0.0590$ & $0.00153$ & $82{,}315$ & $0$ & $0$ \\
\midrule
\multirow{3}{*}{Droplet} 
& SBMT     & $11.63^\circ$ & $0.0131$ & $0.00039$ & $84{,}271$  & $0$ & $81{,}553$ \\
& Triangle & $20.12^\circ$ & $0.0583$ & $0.00084$ & $129{,}581$ & $0$ & $4{,}249$ \\
& Gmsh     & $32.08^\circ$ & $0.0621$ & $0.00147$ & $96{,}613$ & $0$ & $0$ \\
\midrule
\multirow{3}{*}{Y-shape} 
& SBMT     & $10.76^\circ$ & $0.0132$ & $0.00124$ & $35{,}300$  & $0$ & $31{,}341$ \\
& Triangle & $20.23^\circ$ & $0.0579$ & $0.00085$ & $52{,}692$  & $0$ & $1{,}602$ \\
& Gmsh     & $31.69^\circ$ & $0.0523$ & $0.00158$ & $40{,}623$  & $0$ & $0$ \\
\bottomrule
\end{tabular}
\end{table}

\paragraph{Global statistics}
All three methods produce sliver-free meshes under the $5^\circ$
criterion.
As expected, Triangle and Gmsh achieve substantially larger
\emph{minimum} angles, since they are mature global meshing pipelines
explicitly designed to optimize worst-case element quality.
SBMT is not intended to compete with them in that regime.
Instead, it targets a different objective: exact local boundary embedding
on a fixed regular scaffold, with deterministic lookup-driven updates and
without global connectivity changes.

This distinction is reflected in the statistics of
Table~\ref{tab:comparison-combined}.
SBMT preserves a highly regular equilateral interior and perturbs it only
in a narrow boundary band required to embed the prescribed PSLG exactly.
Its smaller minimum-area elements are therefore localized boundary-band
elements introduced by the template rules, rather than evidence of a
global loss of mesh quality.
Triangle and Gmsh, by contrast, redistribute refinement and connectivity
more globally, leading to larger minimum angles and larger minimum areas,
but also to a less structured interior tessellation.

For example, in the droplet domain SBMT generates $84{,}271$ triangles,
of which $81{,}553$ are classified as equilateral, while Gmsh produces
$96{,}613$ triangles and Triangle $129{,}581$ triangles.
Thus, the distinguishing advantage of SBMT is not merely that it starts
from an equilateral background grid, but that it \emph{retains} this
regular interior structure after exact local boundary embedding through a
finite template system, without global remeshing or nonlocal refinement
cascades.
This property is particularly attractive for raster-derived domains,
where structured interiors, limited modification zones, and predictable
locality are often as important as classical mesh-quality metrics.

This behavior is further quantified by the aspect-ratio statistics in
Table~\ref{tab:ar_stats}.
We use the standard quality metric
\[
\mathrm{AR} := \frac{\text{longest edge}}{\text{shortest altitude}},
\]
for which an equilateral triangle satisfies
$\mathrm{AR} = 2/\sqrt{3} \approx 1.1547$.
We report the median, 95th percentile, and maximum aspect ratio for each
domain and method.

\begin{table}[htbp]
\centering
\scriptsize
\setlength{\tabcolsep}{7pt}
\caption{Aspect-ratio statistics ($\mathrm{AR} = \text{longest edge} / \text{shortest altitude}$)
for three domains and three meshing methods. 
$\mathrm{AR}_\mathrm{med}$: median; $\mathrm{AR}_{95}$: 95th percentile; 
$\mathrm{AR}_\mathrm{max}$: maximum.}
\label{tab:ar_stats}
\begin{tabular}{llccc}
\toprule
\textbf{Domain} & \textbf{Method} 
& $\mathbf{AR_\mathrm{med}}$ 
& $\mathbf{AR_{95}}$ 
& $\mathbf{AR_\mathrm{max}}$ \\
\midrule
\multirow{3}{*}{Star}
& SBMT     & $1.1547$ & $1.2434$ & $7.7902$ \\
& Triangle & $1.6449$ & $2.4672$ & $5.4555$ \\
& Gmsh     & $1.2586$ & $1.4980$ & $1.8879$ \\
\midrule
\multirow{3}{*}{Droplet}
& SBMT     & $1.1547$ & $1.1547$ & $7.7916$ \\
& Triangle & $1.6469$ & $2.4693$ & $5.4595$ \\
& Gmsh     & $1.2570$ & $1.4947$ & $1.8802$ \\
\midrule
\multirow{3}{*}{Y-shape}
& SBMT     & $1.1547$ & $2.0063$ & $7.7905$ \\
& Triangle & $1.6523$ & $2.4864$ & $5.4261$ \\
& Gmsh     & $1.2583$ & $1.4929$ & $1.8677$ \\
\bottomrule
\end{tabular}
\end{table}

For SBMT, the median aspect ratio is exactly $2/\sqrt{3}$ for all three
domains, and the 95th percentile remains close to this ideal value
(especially for the star and droplet cases).
This confirms the design goal of SBMT: keep the overwhelming majority of
elements near-equilateral, while confining geometric distortion to a thin
boundary strip.
Triangle and Gmsh, in contrast, achieve stronger global worst-case shape
control, but do so through broader connectivity changes and a more
heterogeneous distribution of element shapes.

From the perspective of bitmap-derived domains, this trade-off is
well motivated.
The raster image assigns equal geometric status to each pixel cell, and
many PDE-based image operators (e.g., diffusion, Poisson reconstruction,
and Hodge/Laplace solves) benefit from isotropic interior sampling and
from stable, predictable local stencils.
Maintaining an almost uniform equilateral interior is therefore not a
cosmetic preference: it is a structural property of the discretization
that is advantageous for repeated numerical processing on image-derived
domains.
Equally important, the lookup-driven SBMT updates are boundary-local and
highly parallelizable, which makes the method especially appealing for
large raster domains and for repeated geometry-aware analysis where full
global remeshing would be unnecessarily expensive.

\paragraph{Downstream elliptic and parabolic workloads}
Beyond mesh statistics, we also evaluated the three meshing strategies in
downstream PDE settings.
The holomorphic-chart reconstructions below demonstrate that SBMT meshes
support stable elliptic computations on all three benchmark domains.
In addition, Supplementary Appendix~H reports a three-way parabolic
heat-diffusion comparison under matched finite-element settings,
confirming that SBMT produces geometry-aligned diffusion fronts while
remaining numerically stable.

\paragraph{Holomorphic reconstructions}
Figure~\ref{fig:comparison_holo} shows holomorphic 1-form
reconstructions on the three benchmark domains.
For each case, we compute a conformal parameterization using the
Gu--Yau algorithm~\cite{Gu2002ComputingCS}, obtain a holomorphic
1-form on the SBMT mesh (via discrete Hodge decomposition), and map the
resulting scalar potential back to the image plane as a texture.
These experiments illustrate that SBMT meshes are not only geometrically
well behaved, but also compatible with elliptic PDE solvers used in
structure-aware image processing.

\begin{figure}[htbp]
\centering
\renewcommand{\arraystretch}{1.0}
\setlength{\tabcolsep}{2pt}
\begin{tabular}{ccc}
\includegraphics[width=0.30\linewidth]{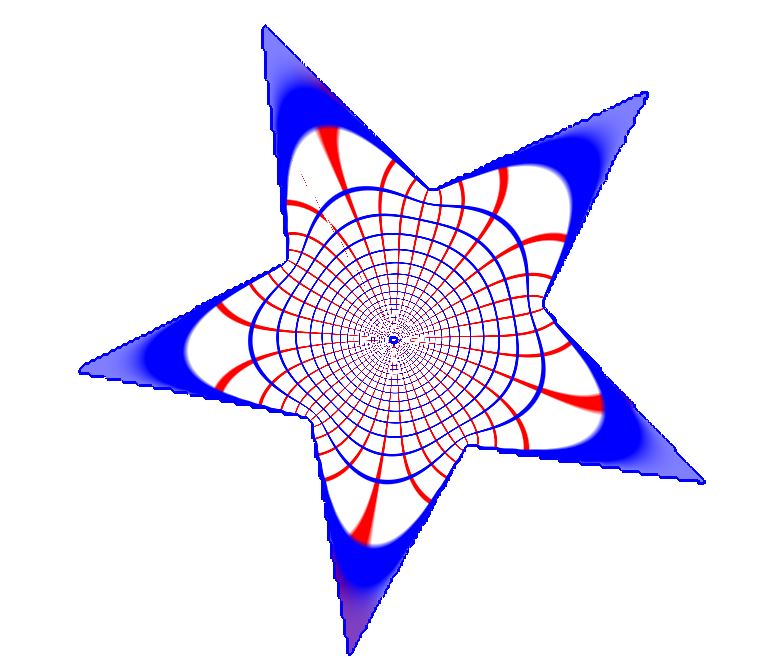} &
\includegraphics[width=0.30\linewidth]{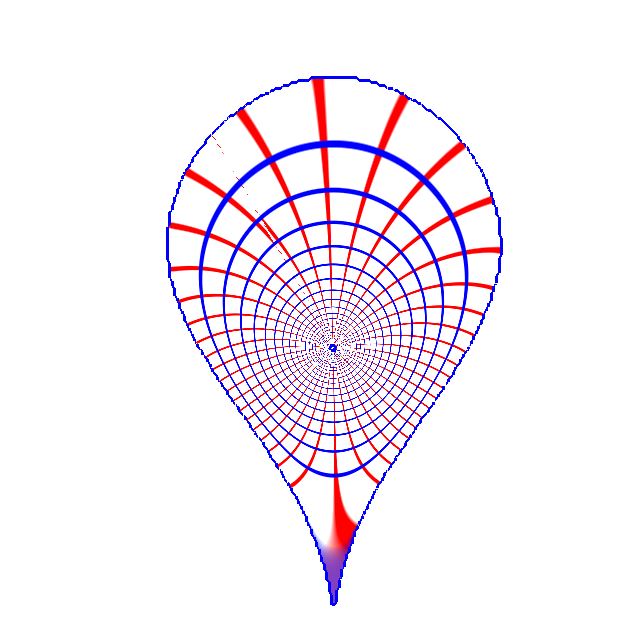} &
\includegraphics[width=0.21\linewidth]{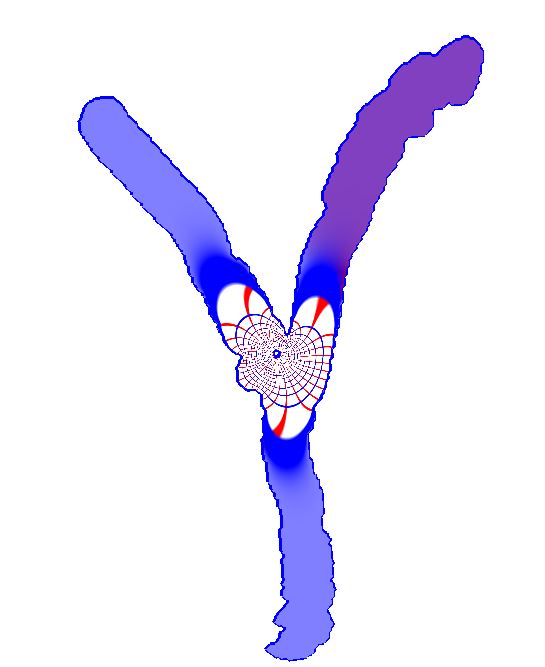} \\
\small (a) Star-shaped domain &
\small (b) Bulged droplet (blob) &
\small (c) Handwritten ``Y''
\end{tabular}
\caption{
Holomorphic 1-form reconstructions on SBMT meshes for the three benchmark domains.
The embedded boundary curves are faithfully aligned with the triangulated grid, 
and the induced scalar fields remain smooth up to the boundary, illustrating the suitability of SBMT for PDE-based image analysis.
}
\label{fig:comparison_holo}
\end{figure}

\paragraph{Local boundary behavior}
To visualize how each method behaves in the most challenging regions,
Figure~\ref{fig:local_triangulation_comparison} zooms into high-curvature
boundary tips on the star and Y-shaped domains.
In both examples, SBMT preserves an almost equilateral pattern up to the
boundary: the interior grid remains regular, and only a few triangles in
a thin boundary strip deviate from the ideal shape in order to embed the
polygonal contour exactly.
Triangle and Gmsh also satisfy the boundary constraints, but their local
meshes display more irregular connectivity and a wider range of element
shapes near the constrained edges, reflecting their global optimization
objectives.

From a geometric standpoint, the mild deviation from equilateral shapes
in this narrow boundary band is acceptable: it is spatially localized,
and the resulting small-area cells capture the fine-scale variation
induced by sharp corners or tight curvature.
For SBMT in particular, these boundary-band elements still satisfy the
theoretical quality bounds in Section~\ref{subsec:geom-quality}, and they
do not propagate into the interior.

\begin{figure}[t]
  \centering
  \includegraphics[width=\linewidth]{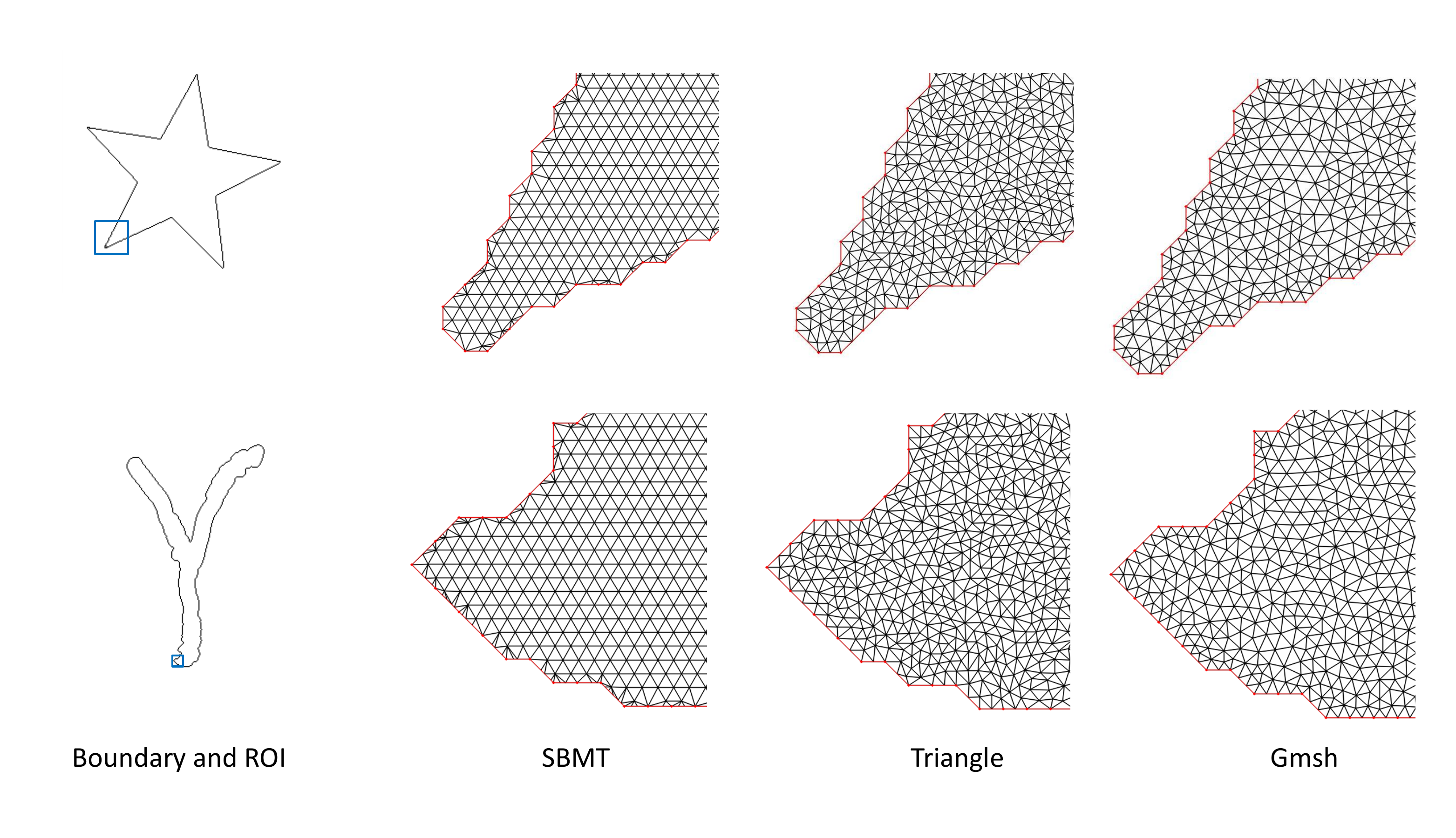}
  \caption{
  Local triangulation near high-curvature boundary tips on two canonical shapes.
  Top row: star-shaped domain; bottom row: Y-shaped domain.
  In each row, the left panel shows the binary boundary (black) and a region of interest (blue box).
  The three panels to the right display zoomed-in triangulations produced by SBMT, Triangle (constrained Delaunay), and Gmsh, respectively, with the raster boundary overlaid in red.
  SBMT preserves an almost equilateral pattern up to the boundary band, while Triangle and Gmsh introduce more irregular elements in order to satisfy their global angle and size objectives.
  }
  \label{fig:local_triangulation_comparison}
\end{figure}

\paragraph{Minimum-angle distributions}
Finally, we compare the full distributions of minimum triangle angles for
the three methods.
For each mesh, we compute the minimum interior angle of every triangle,
quantize it into bins of width $2^\circ$ over $[0^\circ,180^\circ)$, and
accumulate a histogram.
The histogram is then represented at the bin centers
($1^\circ, 3^\circ, 5^\circ, \dots$) and smoothed with a short Gaussian
kernel along the angle axis to produce a continuous curve.
Counts are displayed on a logarithmic scale to reveal both the bulk
behavior and the tails.

The resulting distributions for the three benchmark domains (star,
droplet, and Y-shape) are summarized in Figure~\ref{fig:minAngle_hist}.
In each panel, the horizontal axis is the minimum interior angle
(degrees) and the vertical axis is the triangle count (log scale), with
separate curves for SBMT, Triangle, and Gmsh.
SBMT exhibits a sharp peak near the equilateral angle ($60^\circ$),
reflecting its nearly uniform interior tiling, whereas Triangle and Gmsh
produce broader distributions with larger guaranteed minima but more
variation across the mesh.

\begin{figure}[htbp]
  \centering
  \setlength{\tabcolsep}{3pt}
  \renewcommand{\arraystretch}{1.0}
  \begin{tabular}{ccc}
    \includegraphics[width=0.32\linewidth]{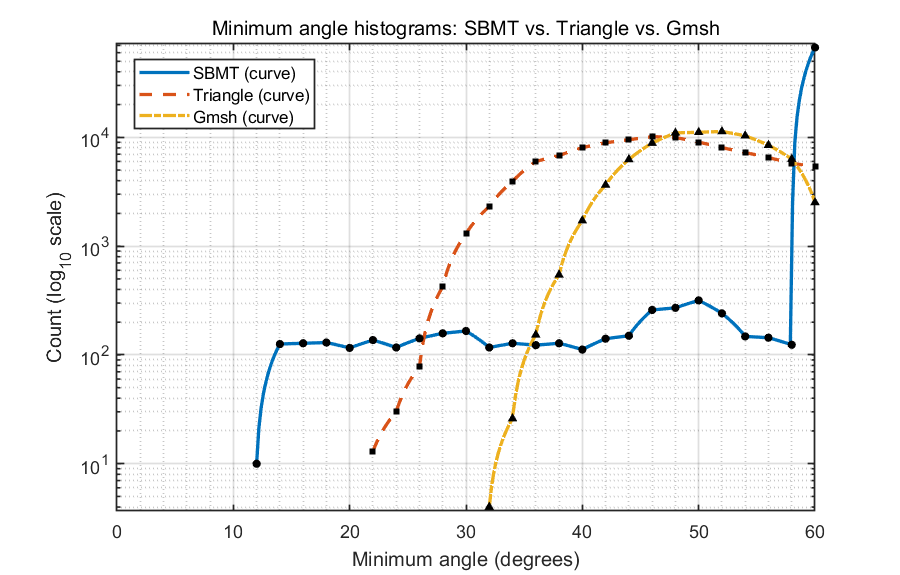} &
    \includegraphics[width=0.32\linewidth]{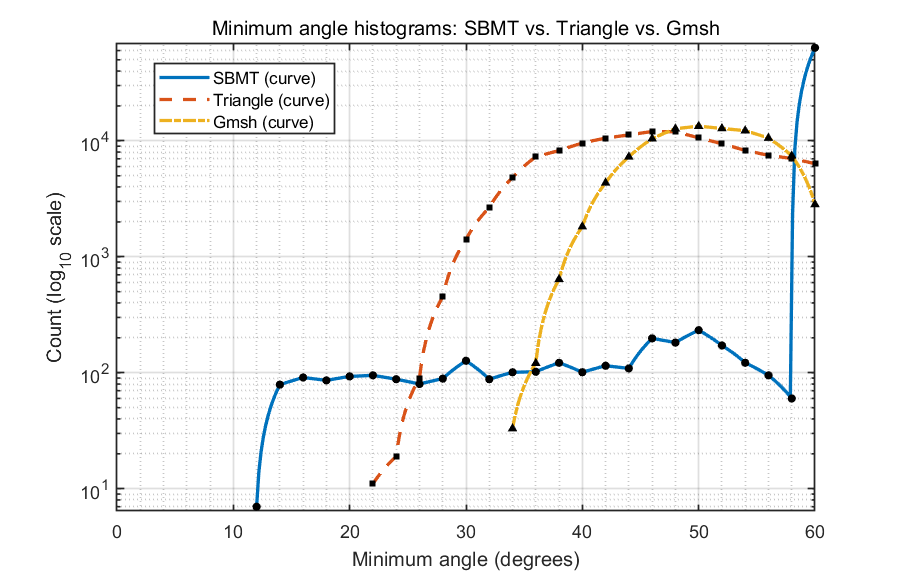} &
    \includegraphics[width=0.32\linewidth]{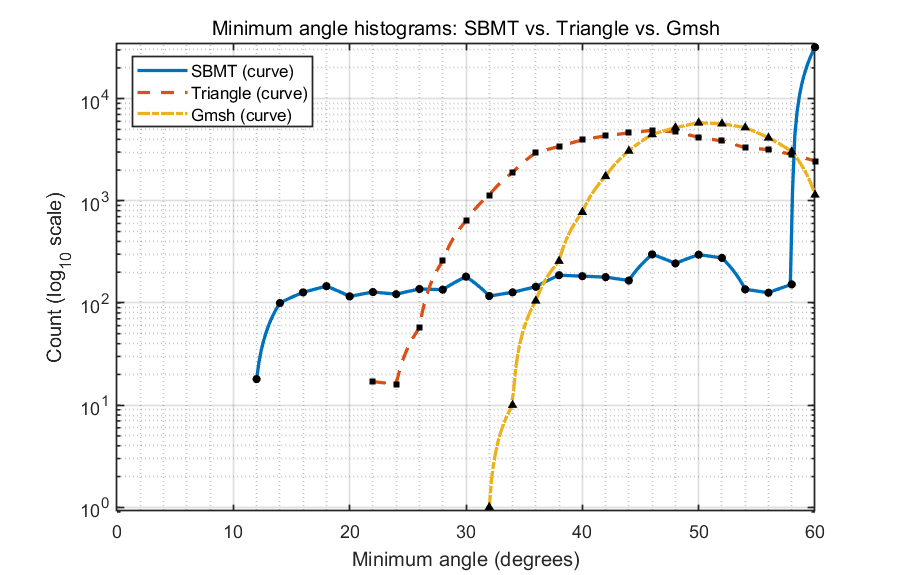} \\
    \small (a) Star shape &
    \small (b) Droplet shape &
    \small (c) Y-shape
  \end{tabular}
  \caption{
  Smoothed histograms of minimum triangle angles for the three benchmark domains:
  (a) star, (b) droplet, and (c) Y-shaped region.
  In each panel, the horizontal axis is the minimum interior angle (degrees) and the vertical axis is the triangle count (log scale). 
  Curves correspond to SBMT, Triangle, and Gmsh, as indicated in the legend.
  SBMT exhibits a sharp peak near the equilateral angle ($60^\circ$), reflecting its nearly uniform grid,
  whereas Triangle and Gmsh produce broader distributions with larger guaranteed minima but more variation across the mesh.
  }
  \label{fig:minAngle_hist}
\end{figure}

\paragraph{Summary}
Triangle and Gmsh remain strong constrained triangulation baselines, with
excellent minimum-angle guarantees and well-developed global optimization
machinery.
SBMT does not aim to outperform them on worst-case angle optimization.
Instead, it offers a complementary regime: a bitmap-native,
template-driven remeshing strategy that
(i) embeds the prescribed boundary exactly,
(ii) preserves an overwhelmingly equilateral and uniform interior
tessellation, and
(iii) achieves sliver-free meshes with explicit lower bounds on local
quality.
Its contribution is therefore not the use of an equilateral background
grid by itself, but the ability to preserve this structured scaffold
under exact local boundary embedding through a finite, deterministic, and
parallel-ready lookup system.
This makes SBMT a natural front-end for subpixel geometry and PDE-based
processing on raster images, and a suitable structural module for the
HoloRay framework developed in our companion work.
\subsection{Reproducibility}

We provide a C++ implementation of SBMT with all core modules and scripts to replicate Section~\ref{sec:experiments}, available at: \url{https://github.com/monge-ampere/SBMT}. The code is a single-threaded prototype for validating the method.

\section{Conclusions}
\label{sec:conclusions}
We introduced Structured Bitmap-to-Mesh Triangulation (SBMT), a template-driven, topology-preserving remeshing framework that converts raster-derived domains into high-quality triangular meshes with exact boundary embedding. By combining an equilateral base mesh with a finite symbolic retriangulation table, SBMT ensures bounded angles, local composability, and stateless execution without global connectivity updates.

Unlike CDT-based methods, SBMT resolves boundary intersections through conflict-free local templates, enhancing mesh regularity and geometric fidelity. Experiments on synthetic, anatomical, and real-world domains demonstrate that SBMT produces sliver-free, boundary-aligned meshes and more coherent isotherms in heat diffusion, with slower yet more geometry-respecting propagation compared to Triangle.

SBMT’s symbolic, parallel-ready architecture enables scalable implementations, including GPU acceleration. It naturally supports multiply connected and open-boundary domains, and extends to level sets and polygonal regions.

Future directions include coupling SBMT with elliptic and transport PDE solvers, introducing adaptive refinement under geometric constraints, and formalizing the template space algebraically. SBMT thus offers a robust, structure-aware meshing strategy for image-based simulations, with clear benefits in accuracy, efficiency, and scalability.

\appendix
\section{Categorical and Algebraic Structure of SBMT}
\label{appendix:category}

This appendix formalizes the template-driven structure of SBMT using categorical language and algebraic constructions~\cite{maclane1998}. These abstractions
clarify the uniqueness and composability of lookup templates, and formalize a qualified order-independence (path independence under independent reordering), thereby positioning SBMT as a symbolic, deterministic, and conflict-free remeshing system.

We organize this section into five parts:
\begin{enumerate}
  \item Defining a symbolic configuration system (objects and schedules) induced by atomic intersection types;
  \item Interpreting the deterministic lookup mapping (from configuration keys to patches) in categorical terms when convenient;
  \item Proving uniqueness/injectivity of the lookup mapping up to triangle symmetries;
  \item Introducing the induced word algebra of template actions (a free monoid/semigroup), with commutation relations from disjoint supports;
  \item Proving qualified path independence as trace equivalence under independent reorderings, under the fixed registry/tolerance conventions.
\end{enumerate}

\subsection{Symbolic Configuration Category}

Let \(\mathcal{C}\) be the finite set of all atomic intersection types between segment paths and triangle interiors. Each element \(c \in \mathcal{C}\) represents a local configuration (e.g., which triangle edges are intersected, their order and orientation). We define a small category \(\mathsf{Config}\) as follows:

\begin{itemize}
  \item \textbf{Objects:} Finite ordered subsets of \(\mathcal{C}\), representing interaction sequences along a boundary path;
  \item \textbf{Morphisms:} A morphism \(f: C_1 \to C_2\) corresponds to a symbolic transition between configuration sequences, typically reducing complexity (e.g., moving toward a canonical empty configuration).
\end{itemize}

The symbolic configuration flow of SBMT is fully captured within this category, with compositional morphisms reflecting rewrites or reductions in the symbolic intersection domain.

\subsection{Functorial Lookup Mechanism}

Let \(\mathsf{Template}\) be the category of retriangulation templates, where:

\begin{itemize}
  \item \textbf{Objects:} Finite ordered subsets of retriangulation templates, each assigned to a configuration in \(\mathcal{C}\);
  \item \textbf{Morphisms:} Symbolic operations such as template concatenation, commutation, and simplification.
\end{itemize}

Then, the SBMT symbolic lookup defines a functor:
\[
F: \mathsf{Config} \to \mathsf{Template}
\]
This functor assigns to each atomic configuration \(c \in \mathcal{C}\) a deterministic retriangulation template \(F(c)\). For any morphism \(f: C_1 \to C_2\), the functor maps it to a composite morphism \(F(f): F(C_1) \to F(C_2)\), preserving the structure of template-driven remeshing.

\subsection{Injectivity and Consistency Properties}

\begin{lemma}[Injectivity up to Symmetry]
\label{lemma:injectivity}
Let \(\sim\) denote an equivalence relation over \(\mathcal{C}\) induced by triangle symmetry group actions (e.g., edge permutations). Then \(F\) is injective on \(\mathcal{C}/\sim\):
\[
c_1 \not\sim c_2 \quad \Rightarrow \quad F(c_1) \neq F(c_2)
\]
\end{lemma}

\begin{proof}
SBMT explicitly enumerates all geometrically distinct atomic configurations in \(\mathcal{C}/\sim\), assigning to each a unique symbolic identifier and retriangulation template. Because symmetry-equivalent classes are collapsed and all remaining distinctions preserved, the mapping is injective on the quotient.
\end{proof}

\begin{lemma}[Floating-Point Consistency of Template Mapping]
\label{lemma:precision}
Assume all geometric predicates in the lookup procedure are evaluated using a fixed global tolerance \(\epsilon > 0\). Then for any configuration \(c \in \mathcal{C}\), its image \(F(c)\) produces numerically conforming template subdivisions across adjacent triangles up to \(\mathcal{O}(\epsilon)\).
\end{lemma}

\begin{proof}
Let \(T_i\) and \(T_j\) be adjacent triangles sharing edge \(e\), intersected by a boundary segment. To avoid inconsistencies due to floating-point deviation, SBMT enforces:

\begin{enumerate}
  \item Deterministic traversal and evaluation order of triangle edges;
  \item Global registry of segment--edge intersections indexed by edge identity and position, with numerical equality defined up to \(\epsilon\).
\end{enumerate}

This ensures that both \(T_i\) and \(T_j\) identify and interpret the same logical intersection configuration, invoking identical retriangulation templates. Resulting subdivisions align within a numerical envelope \(\mathcal{O}(\epsilon)\), maintaining mesh watertightness.
\end{proof}

\subsection{Word Algebra of Template Actions (Free Monoid Envelope)}
\label{sec:free-group}

To capture the compositional semantics of SBMT retriangulation, we construct a free group over symbolic morphisms in the retriangulation template category \(\mathsf{Template}\). While objects in \(\mathsf{Template}\) represent finite sequences of local triangulation states, it is the morphisms---symbolic transitions between these states---that encode the remeshing logic.

Let \(\mathrm{Mor}(\mathsf{Template})\) denote the set of symbolic morphisms between template sequences, including operations such as:
\begin{itemize}
  \item Concatenation of adjacent template actions;
  \item Commutation of non-overlapping substructures;
  \item Cancellation of inverse template pairs;
  \item Reduction via geometric symmetry.
\end{itemize}

We then define the free group:
\[
\mathbb{F}_{\mathsf{Template}} := \mathrm{Free}\left(\mathrm{Mor}(\mathsf{Template})\right)
\]
Each element in \(\mathbb{F}_{\mathsf{Template}}\) corresponds to a symbolic composition of template actions along a segment--triangle interaction path. These symbolic group words track the remeshing process globally while preserving reversibility and statelessness.

\paragraph{Remark (current prototype vs.\ group envelope)}
While the above free-group construction provides a convenient symbolic envelope to articulate reversible and coarsening extensions, the current SBMT prototype operates primarily in the forward (refinement/reshape) direction and is therefore more faithfully modeled by a free semigroup/monoid of template words (a word algebra) generated by atomic template actions, with commutation relations induced only by disjoint supports.

\subsection{Path-Independence Theorem (Sketch)}
\label{thm:path-independence-simplified}

\begin{theorem}[Path Independence of Lookup-Driven SBMT Retriangulation]
\label{thm:path-independence}
Assume the frozen-registry and single-valued lookup setting in
Assumption~\ref{assump:frozen-lookup}, and the edge-consistency guarantee of
Lemma~\ref{lemma:precision}. Let $\Gamma_1,\Gamma_2$ be two schedules that differ
only in the execution order of per-triangle patch-generation actions (with the
registry $\mathcal{R}$ kept fixed). Then
\[
F(\Gamma_1)\equiv F(\Gamma_2),
\]
i.e., they generate the same patch family $\{P_K\}_{K\in\mathcal{T}}$ and yield the
same assembled global mesh after deterministic stitching, up to the fixed
symmetry-breaking/tolerance conventions.
\end{theorem}

\begin{proof}[Sketch]
For each $K\in\mathcal{T}$, let $g_K$ denote the out-of-place action that produces the
unique patch $P_K=\mathcal{L}(\kappa(K))$.
By Assumption~\ref{assump:frozen-lookup}, $g_K$ is a deterministic function of the
frozen registry $\mathcal{R}$ and the local key $\kappa(K)$, and writes only to a
private buffer; hence for any $K\neq K'$ the patch-generation actions commute:
\[
g_K\circ g_{K'} \;=\; g_{K'}\circ g_K,
\]
so any reordering yields the same set of generated patches.

By Lemma~\ref{lemma:precision}, adjacent triangles share the same canonicalized
edge--intersection records along each base edge, so the induced edge traces of their
patches coincide and the stitching stage is deterministic. Therefore the assembled
mesh is independent of the schedule order, proving $F(\Gamma_1)\equiv F(\Gamma_2)$.
\end{proof}

\paragraph{Remark (discrete harmonicity / trivial-holonomy intuition).}
It is often useful to interpret the frozen registry $\mathcal{R}$ as prescribing a \emph{single-valued edge trace} (a discrete ``connection'') on the dual mesh: for every shared base edge $e=K\cap K'$, the two incident triangles read the same ordered
intersection records and thus induce the same subdivision of $e$.
Define the \emph{edge holonomy defect} (a discrete curl surrogate)
\[
\mathrm{hol}(e)\;:=\;\mathrm{trace}(P_K|_e)\;-\;\mathrm{trace}(P_{K'}|_e).
\]
Lemma~\ref{lemma:precision} gives $\mathrm{hol}(e)\equiv 0$ for all shared edges, so the accumulated defect along any dual cycle vanishes (trivial holonomy). This provides the geometric/PDE intuition behind schedule-independence: the lookup-driven remeshing field is curl-free in precisely the sense needed for deterministic patch assembly.

\subsection{Implications}
Together, these constructions imply that SBMT is (i) functorial and deterministic (stateless, with reversible composition in the symbolic envelope), (ii) schedule-independent, and (iii) amenable to algebraic verification and extension.

\subsection*{Outlook: Adaptive and Invertible Template Families}

The categorical viewpoint clarifies the status of the current SBMT prototype and its natural extensions. At the implementation level, SBMT uses only forward, refinement-type local templates on a fixed-resolution equilateral scaffold: each atomic configuration $c\in\mathcal{C}$ is mapped to a deterministic action $F(c)$ that refines or reshapes a bounded patch, without attempting to invert or coarsen previous updates.

Symbolically, the current prototype operates in the forward template-word model clarified above. Within the same framework, one can adjoin additional local generators---coarsening templates, multiresolution transfers between nested grids,
or quality-driven updates (e.g., Delaunay-like edge flips)---provided they remain local and respect the symmetry and numerical-consistency constraints encoded by the lookup functor $F:\mathsf{Config}\to\mathsf{Template}$.

Under these assumptions, the disjoint-support commutativity and the qualified order-independence results (path independence under independent reordering) extend to the enlarged generator set. Thus, the categorical formalization is not merely a reinterpretation of the current lookup table, but a structural basis for future adaptive and multiresolution SBMT variants that retain locality, parallel composability, and deterministic behavior.

\section*{Acknowledgments}
This work was supported in part by the National Natural Science Foundation of China (Grant Nos. 62571503 and 62171421) and in part by the TaiShan Scholar Youth Expert Program of Shandong Province (Grant No. tsqn202306096).

The authors would like to express their sincere gratitude to Prof.~Xianfeng Gu for his publicly available lectures and course materials on computational conformal geometry. Portions of the visualizations and code related to holomorphic 1-forms in this work were adapted from his instructional resources.

\bibliographystyle{elsarticle-num}  % or elsarticle-harv / elsarticle-num-names
\bibliography{references}

@article{shannon1949,
  author  = {Claude E. Shannon},
  title   = {Communication in the Presence of Noise},
  journal = {Proc. IRE},
  volume  = {37},
  number  = {1},
  pages   = {10--21},
  year    = {1949}
}

@book{paley1934,
  author    = {R. E. A. C. Paley and N. Wiener},
  title     = {Fourier Transforms in the Complex Domain},
  series    = {Amer. Math. Soc. Colloq. Publ.},
  volume    = {19},
  publisher = {American Mathematical Society},
  address   = {New York},
  year      = {1934}
}

@article{Gu2002ComputingCS,
  author  = {Xianfeng Gu and Shing-Tung Yau},
  title   = {Computing Conformal Structure of Surfaces},
  journal = {Commun. Inf. Syst.},
  volume  = {2},
  number  = {2},
  pages   = {121--146},
  year    = {2002}
}

@article{blelloch1999design,
  author  = {G.~E. Blelloch and G.~L. Miller and J.~C. Hardwick and D. Talmor},
  title   = {Design and implementation of a practical parallel {D}elaunay algorithm},
  journal = {Algorithmica},
  volume  = {24},
  pages   = {243--269},
  year    = {1999}
}

@inproceedings{parallelrefinement,
  author    = {Andrey Chernikov and Nikos Chrisochoides},
  title     = {Practical and Efficient Point Insertion Scheduling Method for Parallel Guaranteed Quality Delaunay Refinement},
  booktitle = {Proc. Int. Conf. Supercomputing (ICS)},
  pages     = {48--57},
  year      = {2004},
  doi       = {10.1145/1006209.1006217}
}

@article{chernikov2006parallel,
  author  = {Andrey Chernikov and Nikos Chrisochoides},
  title   = {Parallel Guaranteed Quality Delaunay Uniform Mesh Refinement},
  journal = {SIAM J. Sci. Comput.},
  volume  = {28},
  pages   = {1907--1926},
  year    = {2006}
}

@article{algo872,
  author  = {Andrey Chernikov and Nikos Chrisochoides},
  title   = {Algorithm 872: Parallel 2D Constrained Delaunay Mesh Generation},
  journal = {ACM Trans. Math. Softw.},
  volume  = {34},
  number  = {1},
  year    = {2008}
}

@article{DDM,
  author  = {Leonidas Linardakis and Nikos Chrisochoides},
  title   = {Delaunay Decoupling Method for Parallel Guaranteed Quality Planar Mesh Refinement},
  journal = {SIAM J. Sci. Comput.},
  volume  = {27},
  number  = {4},
  pages   = {1394--1423},
  year    = {2006}
}

@article{held2001fist,
  author  = {Martin Held},
  title   = {FIST: Fast Industrial-Strength Triangulation of Polygons},
  journal = {Algorithmica},
  volume  = {30},
  pages   = {563--596},
  year    = {2001}
}

@inproceedings{Shewchuk1996Triangle,
  author    = {Jonathan Richard Shewchuk},
  title     = {Triangle: Engineering a 2D Quality Mesh Generator and Delaunay Triangulator},
  booktitle = {Proc. Workshop Appl. Comput. Geom.},
  pages     = {203--222},
  year      = {1996}
}

@article{shewchuk2002delaunay,
  author  = {Jonathan Richard Shewchuk},
  title   = {Delaunay Refinement Algorithms for Triangular Mesh Generation},
  journal = {Comput. Geom.},
  volume  = {22},
  number  = {1--3},
  pages   = {21--74},
  year    = {2002}
}

@book{gonzales1987digital,
  author    = {Rafael C. Gonzales and Paul Wintz},
  title     = {Digital Image Processing},
  publisher = {Addison--Wesley},
  year      = {1987}
}

@inproceedings{shewchuk2002good,
  author    = {Jonathan Richard Shewchuk},
  title     = {What is a Good Linear Element? Interpolation, Conditioning, and Quality Measures},
  booktitle = {Proc. 11th Int. Meshing Roundtable (IMR)},
  year      = {2002}
}

@article{goodchild1992geographical,
  author  = {Michael F. Goodchild},
  title   = {Geographical Data Modeling},
  journal = {Comput. Geosci.},
  volume  = {18},
  number  = {4},
  pages   = {401--408},
  year    = {1992}
}

@book{leveque2007finite,
  author    = {Randall J. LeVeque},
  title     = {Finite Difference Methods for Ordinary and Partial Differential Equations: Steady-State and Time-Dependent Problems},
  publisher = {SIAM},
  year      = {2007}
}

@article{ruppert1995delaunay,
  author  = {J. Ruppert},
  title   = {A Delaunay Refinement Algorithm for Quality 2-Dimensional Mesh Generation},
  journal = {J. Algorithms},
  volume  = {18},
  number  = {3},
  pages   = {548--585},
  year    = {1995}
}

@inproceedings{chew1987constrained,
  author    = {L. Paul Chew},
  title     = {Constrained Delaunay Triangulations},
  booktitle = {Proc. 3rd Annu. Symp. Comput. Geom.},
  pages     = {215--222},
  year      = {1987}
}

@inproceedings{meyer2003discrete,
  author    = {Mark Meyer and Mathieu Desbrun and Peter Schr{\"o}der and Alan H. Barr},
  title     = {Discrete Differential-Geometry Operators for Triangulated 2-Manifolds},
  booktitle = {Vis. Math. III},
  pages     = {35--57},
  year      = {2003}
}

@book{li2004digital,
  author    = {Zhilin Li and Christopher Zhu and Chris Gold},
  title     = {Digital Terrain Modeling: Principles and Methodology},
  publisher = {CRC Press},
  year      = {2004}
}

@book{bankman2008handbook,
  author    = {Isaac Bankman},
  title     = {Handbook of Medical Image Processing and Analysis},
  publisher = {Elsevier},
  year      = {2008}
}

@article{chazelle1991triangulating,
  author  = {Bernard Chazelle},
  title   = {Triangulating a Simple Polygon in Linear Time},
  journal = {Discrete Comput. Geom.},
  volume  = {6},
  number  = {3},
  pages   = {485--524},
  year    = {1991}
}

@article{de1992line,
  author  = {Leila De Floriani and Enrico Puppo},
  title   = {An On-line Algorithm for Constrained Delaunay Triangulation},
  journal = {CVGIP Graph. Models Image Process.},
  volume  = {54},
  number  = {4},
  pages   = {290--300},
  year    = {1992}
}

@incollection{bern1995mesh,
  author    = {Marshall Bern and David Eppstein},
  title     = {Mesh Generation and Optimal Triangulation},
  booktitle = {Comput. Euclidean Geom.},
  pages     = {47--123},
  publisher = {World Scientific},
  year      = {1995}
}

@article{xiao2022image,
  author  = {Yanyang Xiao and Juan Cao and Zhonggui Chen},
  title   = {Image Representation on Curved Optimal Triangulation},
  journal = {Comput. Graph. Forum},
  volume  = {41},
  number  = {6},
  pages   = {23--36},
  year    = {2022}
}

@incollection{lorensen1998marching,
  author    = {William E. Lorensen and Harvey E. Cline},
  title     = {Marching Cubes: A High Resolution 3D Surface Construction Algorithm},
  booktitle = {Seminal Graph.},
  pages     = {347--353},
  publisher = {ACM},
  year      = {1998}
}

@article{treece1999regularised,
  author  = {Graham M. Treece and Richard W. Prager and Andrew H. Gee},
  title   = {Regularised Marching Tetrahedra: Improved Iso-Surface Extraction},
  journal = {Comput. Graph.},
  volume  = {23},
  number  = {4},
  pages   = {583--598},
  year    = {1999}
}

@inproceedings{hilton1996marching,
  author    = {Adrian Hilton and Andrew J. Stoddart and John Illingworth and Terry Windeatt},
  title     = {Marching Triangles: Range Image Fusion for Complex Object Modelling},
  booktitle = {Proc. 3rd IEEE Int. Conf. Image Process. (ICIP)},
  volume    = {2},
  pages     = {381--384},
  year      = {1996}
}

@article{leyva2023satellite,
  author  = {Israel Leyva-Mayorga and Marc Martinez-Gost and Marco Moretti and Ana Pérez-Neira and Miguel Ángel Vázquez and Petar Popovski and Beatriz Soret},
  title   = {Satellite Edge Computing for Real-Time and Very-High Resolution Earth Observation},
  journal = {IEEE Trans. Commun.},
  volume  = {71},
  number  = {10},
  pages   = {6180--6194},
  year    = {2023}
}

@article{hu2019triwild,
  author  = {Yixin Hu and Teseo Schneider and Xifeng Gao and Qingnan Zhou and Alec Jacobson and Denis Zorin and Daniele Panozzo},
  title   = {TriWild: Robust Triangulation with Curve Constraints},
  journal = {ACM Trans. Graph.},
  volume  = {38},
  number  = {4},
  pages   = {1--15},
  year    = {2019}
}

@book{demmel1997applied,
  author    = {James W. Demmel},
  title     = {Applied Numerical Linear Algebra},
  publisher = {SIAM},
  year      = {1997}
}

@phdthesis{kadow2004parallel,
  author       = {Clemens Martin Joachim Kadow},
  title        = {Parallel Delaunay Refinement Mesh Generation},
  school       = {Carnegie Mellon University},
  year         = {2004},
  type         = {Ph.D. dissertation}
}

@article{spielman2007parallel,
  author  = {Daniel A. Spielman and Shang-Hua Teng and Alper Üngör},
  title   = {Parallel Delaunay Refinement: Algorithms and Analyses},
  journal = {Int. J. Comput. Geom. Appl.},
  volume  = {17},
  number  = {1},
  pages   = {1--30},
  year    = {2007}
}

@article{KOHOUT2005491,
  author  = {Josef Kohout and Ivana Kolingerová and Jiří Žára},
  title   = {Parallel Delaunay Triangulation in $\mathbb{E}^2$ and $\mathbb{E}^3$ for Computers with Shared Memory},
  journal = {Parallel Comput.},
  volume  = {31},
  number  = {5},
  pages   = {491--522},
  year    = {2005}
}

@inproceedings{hoppe1993mesh,
  author    = {Hugues Hoppe and Tony DeRose and Tom Duchamp and John McDonald and Werner Stuetzle},
  title     = {Mesh Optimization},
  booktitle = {Proc. 20th Annu. Conf. Comput. Graph. Interact. Tech. (SIGGRAPH)},
  pages     = {19--26},
  year      = {1993}
}

@inproceedings{nealen2006laplacian,
  author    = {Andrew Nealen and Takeo Igarashi and Olga Sorkine and Marc Alexa},
  title     = {Laplacian Mesh Optimization},
  booktitle = {Proc. 4th Int. Conf. Comput. Graph. Interact. Tech. Australasia Southeast Asia},
  pages     = {381--389},
  year      = {2006}
}

@incollection{lawson1977software,
  author    = {Charles L. Lawson},
  title     = {Software for {C1} Surface Interpolation},
  booktitle = {Mathematical Software III},
  editor    = {John R. Rice},
  publisher = {Academic Press},
  year      = {1977},
  pages     = {161--194}
}

@book{maclane1998,
  author       = {Saunders Mac Lane},
  title        = {Categories for the Working Mathematician},
  series       = {Graduate Texts in Mathematics},
  volume       = {5},
  publisher    = {Springer},
  year         = {1998}
}

@article{idelsohn2006mesh,
  author  = {S.R. Idelsohn and E. O{\~n}ate},
  title   = {To mesh or not to mesh. That is the question{\dots}},
  journal = {Comput. Methods Appl. Mech. Engrg.},
  volume  = {195},
  number  = {37--40},
  pages   = {4681--4696},
  year    = {2006},
  publisher = {Elsevier}
}

@article{Ortiz1991,
  author    = {M. Ortiz and J. J. Quigley},
  title     = {Adaptive mesh refinement in strain localization problems},
  journal   = {Comput. Methods Appl. Mech. Engrg.},
  year      = {1991},
  volume    = {90},
  number    = {1--3},
  pages     = {781--804},
  publisher = {Elsevier}
}

@article{li2013tuned,
  author  = {P. Li and M. D. Adams},
  title   = {A tuned mesh-generation strategy for image representation based on data-dependent triangulation},
  journal = {IEEE Trans. Image Process.},
  volume  = {22},
  number  = {5},
  pages   = {2004--2018},
  year    = {2013}
}

@article{peyre2011review,
  author  = {G. Peyr{\'e}},
  title   = {A Review of Adaptive Image Representations},
  journal = {IEEE J. Sel. Top. Signal Process.},
  volume  = {5},
  number  = {5},
  pages   = {896--911},
  year    = {2011}
}

@article{lawonn2019stylized,
  author  = {K. Lawonn and T. G{\"u}nther},
  title   = {Stylized Image Triangulation},
  journal = {Comput. Graph. Forum},
  volume  = {38},
  number  = {1},
  pages   = {221--234},
  year    = {2019}
}

@article{Rippa92,
  author  = {S. Rippa},
  title   = {Adaptive approximation by piecewise linear polynomials on triangulations of subsets of scattered data},
  journal = {SIAM J. Sci. Stat. Comput.},
  volume  = {13},
  number  = {5},
  pages   = {1123--1141},
  year    = {1992},
  doi     = {10.1137/0913065}
}

\clearpage
\appendix
\section*{Appendix / Supplementary Material}
\addcontentsline{toc}{section}{Appendix / Supplementary Material}

\section*{Appendix B\\Proof of Lookup Table Completeness}
\label{app:lookup-completeness}

This appendix provides the proof of
Theorem~3.2 on the completeness of the
SBMT retriangulation lookup table.

\begin{proof}[Proof of Theorem~3.2]
By Lemma~3.1, each base triangle $T$ is
intersected by at most two consecutive boundary segments, say
$S_1$ and $S_2$, and the geometric protocol of
Section~3.2 bounds the number of embedded
intersection points on $\partial T$.

Under the intersection-counting convention of
Algorithm~1, each segment contributes an
integer count between $0$ and $3$: it may be disjoint from $T$ (count
$0$), meet one or two edges (counts $1$ or $2$, with vertex-on-segment
events counted twice), or intersect all three edges (count $3$).
Thus, in the $(k,\ell)$ notation, we necessarily have
$k,\ell \in \{0,1,2,3\}$ after merging coincident events induced by the
snapping and preprocessing rules.

The pair $(0,0)$ corresponds to the trivial case with no intersection
and therefore does not trigger any retriangulation.
Conversely, a configuration of type $(3,3)$ cannot occur: if both
$S_1$ and $S_2$ generated three intersections on $\partial T$, then
either there would be at least four distinct intersection points on the
three edges of $T$, contradicting Lemma~3.1,
or the two segments would share three distinct points and hence coincide
locally, which is incompatible with the boundary-chain model.

It follows that any nontrivial, remeshing-relevant configuration must
have $(k,\ell)$ among
\[
(0,1), (0,2), (0,3),\;
(1,1), (1,2), (2,2), (1,3),\;
\text{or } (2,3).
\]
Appendix~F gives an explicit enumeration (up to symmetry) of all
segment--triangle patterns consistent with
Lemma~3.1,
Algorithm~1, and the preprocessing rules, and
shows that each such pattern is represented by exactly one of the
canonical types listed above.
For each canonical type, Appendix~G provides at least one retriangulation
template, and the lookup table selects a template solely from the
symbolic label.

Hence every admissible, nontrivial interaction of boundary segments
with a triangle $T$ is covered by at least one entry of the lookup
table, establishing the claimed completeness.
\end{proof}

\section*{Appendix C\\Time Complexity Analysis}
\label{subsec:time-complexity}

We conclude the method section by analyzing the runtime complexity of the SBMT pipeline, which consists of three stages: (1) geometric preprocessing, (2) localized retriangulation, and (3) image-based interpolation.

Let \(n\) be the number of input boundary segments, \(N\) the number of base-mesh triangles, \(N_t\) the number of triangles intersected by boundaries, and \(m\) the average number of triangles influenced by a segment (typically constant due to local thresholds).

\begin{lemma}[Boundary-Conforming Preprocessing]
The runtime of geometric preprocessing, including vertex snapping, vertex repulsion, and edge elimination, is bounded by \(\mathcal{O}(n \log N + nm)\).
\end{lemma}

\begin{proof}
Each boundary segment initiates a proximity query to locate nearby vertices or edges using spatial thresholds \(a\), \(b\), and \(c\). These queries are executed via an auxiliary spatial index (e.g., a KD-tree) built on the \(N\) mesh triangles. The query cost is \(\mathcal{O}(\log N)\) per segment, totaling \(\mathcal{O}(n \log N)\).

Post-query, each segment affects \(m\) nearby triangles on average. Each triangle triggers a constant-time geometric update (snap, repel, or mark for deletion), costing \(\mathcal{O}(nm)\) overall.

Thus, the total preprocessing cost is \(\mathcal{O}(n \log N + nm)\).
\end{proof}

\begin{lemma}[Boundary-Driven Triangulation]
\label{lem:boundary-triangulation}
The full meshing process—including preprocessing and retriangulation—runs in
\[
\mathcal{O}(n \log N + nm + N_t),
\]
where \(N_t\) is the number of triangles intersected by the boundary.
\end{lemma}

\begin{proof}
Preprocessing requires \(\mathcal{O}(n \log N + nm)\) (Lemma 1). Then, for each of the triangles interposed with the boundaries \(N_t\), a pre-computed retriangulation template is applied via table lookup and constant-time updates. Thus, retriangulation costs \(\mathcal{O}(N_t)\), yielding the total complexity.
\end{proof}

\begin{theorem}[Overall Time Complexity of SBMT]
The total runtime of the SBMT pipeline, including mesh generation and interpolation, is
\[
\mathcal{O}(n \log N + nm + N_t + N).
\]
\end{theorem}

\begin{proof}
From Lemma~\ref{lem:boundary-triangulation}, the meshing stage runs in \(\mathcal{O}(n \log N + nm + N_t)\).

After meshing, signal interpolation (e.g., intensity or scalar field) is applied per triangle, typically via linear or barycentric schemes. Since each triangle requires constant time, the interpolation cost is \(\mathcal{O}(N)\).

Summing both stages yields:
\[
\mathcal{O}(n \log N + nm + N_t + N),
\]
reflecting the combined cost of spatial filtering, retriangulation, and numerical post-processing.
\end{proof}

\section*{Appendix D\\Additional Theorems: Minimal Angle and Area Bounds}

\begin{theorem}[Minimal Angle Bound in SBMT]
\label{thm:min-angle-bound}
Given geometric thresholds $a$, $b$, and $c$ satisfying:
\[
b < \frac{a}{2}, \qquad c < \frac{a}{\sqrt{2}},
\]
every sub-triangle generated through SBMT retriangulation has a minimal interior angle bounded below by:
\[
\theta_{\min} > \min\left\{
\arctan\left( \frac{b}{e + a - \sqrt{a^2 - b^2}} \right),\; 
\arctan\left( \frac{c}{e + a} \right)
\right\},
\]
where $e$ denotes the edge length of the original equilateral triangle in the base mesh.
\end{theorem}

\begin{figure}[htbp]
\centering
\begin{tikzpicture}[scale=1.5]

% ------- 图 (a) -------
\node at (0.5, 2.8) {(a)};
\coordinate (A) at (1,0);
\coordinate (B) at (3,0);
\coordinate (C) at (2,1.732);
\draw[thick] (A) -- (B) -- (C) -- cycle;
\node[below left] at (A) {$A$};
\node[below right] at (B) {$B$};
\node[above] at (C) {$C$};

\coordinate (P1) at (0,-0.5);
\coordinate (P2) at (2,0.5);
\coordinate (P3) at (1.3,2.3); % 更靠近三角形的 P3
\draw[red, thick] (P1) -- (P2);
\draw[red, thick] (P2) -- (P3);
\filldraw[red] (P1) circle (0.03) node[left] {$P_1$};
\filldraw[red] (P2) circle (0.03) node[below right] {$P_2$};
\filldraw[red] (P3) circle (0.03) node[right] {$P_3$};

\coordinate (Q) at (1.709, 1.229); % 交点 Q = P2P3 与 AC
\filldraw[blue] (Q) circle (0.03) node[left] {$Q$};

% 简洁圆弧表示 \theta_1，在点 A 附近
\draw (1,0) ++(10:0.2) arc[start angle=10, end angle=70, radius=0.2];
\node at (1.18,0.15) {\scriptsize$\theta_1$};

% ------- 图 (b) -------
\begin{scope}[xshift=4.5cm]
\node at (0.5, 2.8) {(b)};
\coordinate (A2) at (1,0);
\coordinate (B2) at (3,0);
\coordinate (C2) at (2,1.732);
\draw[thick] (A2) -- (B2) -- (C2) -- cycle;
\node[below left] at (A2) {$A$};
\node[below right] at (B2) {$B$};
\node[above] at (C2) {$C$};

\coordinate (P21) at (0,1.0);
\coordinate (P22) at (3,1.0);
\draw[red, thick] (P21) -- (P22);
\filldraw[red] (P21) circle (0.03) node[left] {$P_1$};
\filldraw[red] (P22) circle (0.03) node[below right] {$P_2$};

\coordinate (Q1) at (2.423, 1.0);  % 与 BC 相交点
\coordinate (Q2) at (1.578, 1.0);  % 与 AB 相交点
\filldraw[blue] (Q1) circle (0.03) node[above right] {$Q_1$};
\filldraw[blue] (Q2) circle (0.03) node[above left] {$Q_2$};

% 添加 Q1 -- A 的虚线
\draw[dashed, gray] (Q1) -- (A2);

% 简洁圆弧表示 θ₂，在点 A 附近
\draw (1,0) ++(5:0.2) arc[start angle=5, end angle=35, radius=0.2];
\node at (1.15,0.12) {\scriptsize$\theta_2$};

\end{scope}

\end{tikzpicture}

\caption{
Worst-case angle configurations used to derive lower bounds for triangle angles. 
(a)~Angle $\theta_1 = \angle QAP_2$ is bounded below by $\arctan\left(\frac{b}{e + a - \sqrt{a^2 - b^2}}\right)$, where $b$ is the edge threshold. 
(b)~Angle $\theta_2 = \angle Q_1AB$ is bounded below by $\arctan\left(\frac{c}{e + a}\right)$, where $c$ is the projection threshold from a vertex to a boundary segment.
}
\label{fig:angle_bounds}
\end{figure}
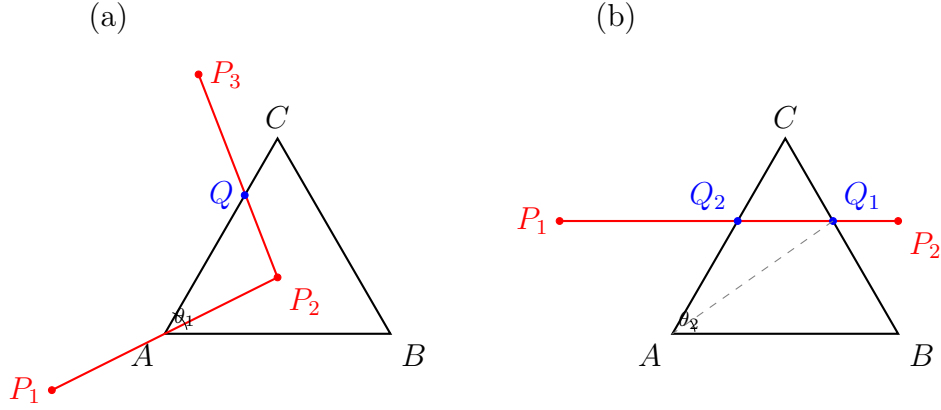

\begin{proof}
We analyze two canonical worst-case configurations illustrated in Figure~\ref{fig:angle_bounds}.

\textbf{Case 1 (Edge elimination near a snapped vertex, as illustrated in the left diagram of Figure~\ref{fig:angle_bounds}):}

Consider a triangle with base edge $(C,A)$ of original length $e$. 
Suppose vertex $A$ is snapped toward a boundary point $s$ located along the extension of edge $CA$, 
such that $\|A - s\| = a$. After snapping, the edge $AC$ becomes length $e + a$.

Assume a boundary segment intersects this extended edge at a point $Q$ within perpendicular distance $b$ of the original edge direction. 
To maintain snapping eligibility, the intersection must lie within a circle of radius $a$ centered at $A$. 
The minimal possible distance from $Q$ to vertex $C$ is then lower-bounded by:
\[
CQ \geq \sqrt{a^2 - b^2}.
\]
This implies the maximal distance from $A$ to $Q$ is:
\[
AQ \leq e + a - \sqrt{a^2 - b^2}.
\]

The triangle $\triangle QAP_2$ thus formed has its minimal angle at vertex $A$, bounded below by:
\[
\theta_1 > \arctan\left( \frac{b}{e + a - \sqrt{a^2 - b^2}} \right).
\]

\textbf{Case 2 (Vertex repulsion from a boundary segment, as illustrated in the right diagram of Figure~\ref{fig:angle_bounds}):}

Consider triangle $\triangle Q_1AB$ formed after vertex $A$ is repelled from a nearby boundary segment $P_1P_2$ along its normal direction. 
The repulsion places $A$ at a perpendicular distance of $c$ from the boundary segment, with $Q_1$ denoting the foot of the perpendicular.

To establish the worst-case scenario, we consider $A$ to have been previously snapped toward a boundary point $s$ located along the extension of edge $BA$, 
with $\|A - s\| = a$. In this case, the effective length of edge $AB$ becomes:
\[
AB \leq e + a.
\]
Since the vertical offset (height) of $A$ from the segment is exactly $c$, the resulting angle at vertex $B$ satisfies:
\[
\theta_2 > \arctan\left( \frac{c}{e + a} \right).
\]
\end{proof}

\begin{remark}
The two worst-case configurations shown in Figure~\ref{fig:angle_bounds} fully capture the minimal-angle-generating scenarios under the SBMT framework. 
Any other intersection pattern either reduces to one of these two via symmetry or fails to produce an angle as small due to geometric constraints.

Specifically, all minimal angles in retriangulated sub-triangles arise from:
\begin{itemize}
    \item \textbf{Angle slicing:} a boundary segment intersects a triangle vertex and creates two acute angles;
    \item \textbf{Edge flattening:} a boundary segment compresses the triangle from two edges, displacing the opposite vertex via foot projection.
\end{itemize}

These two classes represent the only configurations where two adjacent edges can be sufficiently elongated and aligned (within thresholds $a$, $b$, $c$) to produce a sharp internal angle. Configurations involving segment penetration through triangle interiors cannot achieve comparable stretching and thus yield larger angles.

Hence, the lower bounds in Theorem~\ref{thm:min-angle-bound} are tight within the design space of SBMT.
\end{remark}

\begin{theorem}[Minimal Area Bound in SBMT]
\label{thm:min-area-bound}
Let $b$ and $c$ denote the edge elimination and foot projection thresholds, respectively. Then every sub-triangle generated through SBMT retriangulation satisfies the minimal area constraint:
\[
A_{\min} > \frac{1}{2} b \cdot c.
\]
\end{theorem}

\begin{proof}
The worst-case minimal-area configuration arises when:
\begin{itemize}
  \item A boundary segment intersects a triangle edge at a perpendicular distance $b$ from a vertex;
  \item A nearby vertex is repelled orthogonally from another boundary segment at distance $c$.
\end{itemize}

In this configuration, a right-angled triangle is formed where the base and height are bounded above by $b$ and $c$, respectively. Thus, its area satisfies:
\[
A = \frac{1}{2} b \cdot c.
\]

Since SBMT enforces local non-degeneracy through snapping and repulsion thresholds, and avoids overlapping retriangulation zones, all actual triangle areas strictly exceed this lower bound:
\[
A_{\min} > \frac{1}{2} b \cdot c.
\]

A geometric illustration of this worst-case construction is provided in Figure~\ref{fig:area_bound_correct}.
\end{proof}

\begin{remark}
Unlike the minimal-angle bound, which involves elongated triangle sides and acute angle formations due to alignment or slicing, the minimal-area configuration arises when both the triangle's base and height are simultaneously constrained by thresholds.

In particular, the worst-case scenario corresponds to a right-angled triangle with:
\begin{itemize}
  \item Base $\leq b$: induced by edge proximity under the edge-elimination threshold;
  \item Height $\leq c$: induced by orthogonal vertex repulsion from a nearby boundary segment.
\end{itemize}

This configuration uniquely minimizes area among all local retriangulation patterns permitted by SBMT. Any deviation from right-angle geometry or increase in either base or height yields strictly larger area, hence the bound is tight.

See Figure~\ref{fig:area_bound_correct} for a schematic illustration.
\end{remark}

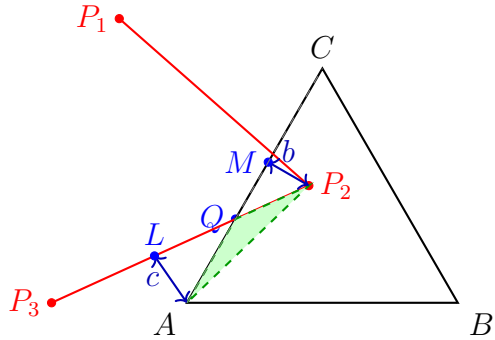
\begin{figure}[htbp]
\centering
\begin{tikzpicture}[scale=1.8]

% 点坐标
\coordinate (A) at (1,0);
\coordinate (C) at (2,1.732);
\coordinate (P2) at (1.9,0.866);
\coordinate (P3) at (0,0);

% 辅助：AC直线参数
% 向量 v1 = C - A = (1, 1.732)
% 向量 v2 = P2 - A = (0.9, 0.866)
% 投影系数：dot(v1,v2)/|v1|^2 ≈ (0.9*1 + 0.866*1.732)/(1^2 + 1.732^2)
% => ≈ (0.9 + 1.5)/4 ≈ 0.6
\coordinate (Bfoot) at ($(A)!0.6!(C)$);  % 垂足：P2 到 AC 的垂线

% 垂足：A 到 P2P3 的垂足
% 向量 w1 = P3 - P2 = (-1.9, -0.866)
% 向量 w2 = A - P2 = (-0.9, -0.866)
% 投影系数 = dot(w1,w2)/|w1|^2 ≈ (1.71 + 0.75)/4.2 ≈ 0.6
\coordinate (Q2) at ($(P2)!0.6!(P3)$);

\coordinate (Q) at (1.358, 0.619);
\filldraw[blue] (Q) circle (0.03) node[left] {$Q$};

% 原始三角形 ABC
\coordinate (B) at (3,0);
\draw[thick] (A) -- (B) -- (C) -- cycle;
\node[below left] at (A) {$A$};
\node[below right] at (B) {$B$};
\node[above] at (C) {$C$};

% 红折线段 P1-P2-P3
\coordinate (P1) at (0.5,2.1);
\draw[red, thick] (P1) -- (P2);
\draw[red, thick] (P2) -- (P3);
\filldraw[red] (P1) circle (0.03) node[left] {$P_1$};
\filldraw[red] (P2) circle (0.03) node[right] {$P_2$};
\filldraw[red] (P3) circle (0.03) node[left] {$P_3$};

% 三角形 AQ2P2 着色
\fill[green!20, opacity=1.0] (A) -- (Q) -- (P2) -- cycle;
\draw[green!60!black, thick, dashed] (A) -- (Q) -- (P2) -- cycle;

% c：A 到 P2P3 的垂线
\draw[<->, thick, blue!70!black] (A) -- node[left] {$c$} (Q2);

% b：P2 到 AC 的垂线
\draw[<->, thick, blue!70!black] (P2) -- node[above] {$b$} (Bfoot);

% 垂足点标注
\filldraw[blue] (Q2) circle (0.03) node[above]{$L$};
\filldraw[blue] (Bfoot) circle (0.03) node[left]{$M$};

% 辅助线
\draw[dashed, gray] (A) -- (C);

\end{tikzpicture}
\caption{
A worst-case triangle $\triangle AQP_2$ is formed when vertex $A$ is projected orthogonally onto the boundary segment $P_2P_3$, producing foot $L$. Meanwhile, the point $P_2$ is projected onto edge $AC$ to form foot $M$. If $|AL| > c$ and $|P_2M| > b$, then the resulting triangle satisfies $\text{Area}(AQP_2) > \frac{1}{2} b \cdot c$, guaranteeing a lower bound for sub-triangle area during retriangulation.
}
\label{fig:area_bound_correct}
\end{figure}

\section*{Appendix E\\Spectrally Accurate Interpolation over Equilateral Meshes}

Bitmap images are discretized observations of physical signals, which are rarely strictly bandlimited~\cite{shannon1949} due to edges, discontinuities, and noise. By the Paley--Wiener theorem~\cite{paley1934}, any strictly bandlimited function must be entire (analytic across the complex plane), and thus cannot exhibit sharp transitions. In contrast, real-world signals are typically \emph{effectively bandlimited}: their Fourier spectra decay rapidly enough that the energy beyond a certain frequency becomes negligible.

Moreover, most physical acquisition systems—such as optical lenses or CCD sensors—introduce inherent low-pass filtering, smoothing high-frequency components before digital sampling.

Let \( f(x, y) \) be a compactly supported image function whose Fourier transform \( \hat{f}(\xi, \eta) \) satisfies
\[
|\hat{f}(\xi, \eta)| < \delta, \quad \text{for } \sqrt{\xi^2 + \eta^2} > \Omega,
\]
for some effective cutoff frequency \( \Omega > 0 \) and small tail amplitude \( \delta \). That is, \( f \) is effectively bandlimited to the disk \( \|\boldsymbol{\xi}\| \leq \Omega \).

We now establish a sufficient sampling condition for piecewise interpolation over equilateral triangular meshes to approximate \( f \) with spectral accuracy.

\begin{theorem}[Spectrally Accurate Interpolation over Equilateral Meshes]
\label{thm:spectral_interp}
Let \( f \) be a compactly supported function that is effectively bandlimited with cutoff frequency \( \Omega \), and let \( \mathcal{T}_h \) be a structured triangular mesh composed of equilateral triangles of edge length \( h \). Then, if
\[
h < \frac{1.86}{\Omega},
\]
any piecewise polynomial interpolation \( \tilde{f} \) over \( \mathcal{T}_h \) (e.g., linear or quadratic) satisfies
\[
\| \tilde{f} - f \|_{L^\infty(\Omega)} < \varepsilon,
\]
where \( \varepsilon \) depends on the high-frequency tail of \( \hat{f} \) beyond \( \Omega \), and vanishes as \( \Omega \to \infty \) or \( h \to 0 \).
\end{theorem}

\begin{proof}
Each equilateral triangle in the mesh has area \( A_\triangle = \frac{\sqrt{3}}{4} h^2 \). In a regular tiling, each vertex is shared by six such triangles, so the average sampling density corresponds to one sample per
\[
A_\text{vertex} = \frac{1}{6} A_\triangle = \frac{\sqrt{3}}{24} h^2.
\]
This is equivalent to a point spacing of
\[
\delta = \sqrt{A_\text{vertex}} = \sqrt{ \frac{\sqrt{3}}{24} }\, h \approx 0.269\, h.
\]
By the Nyquist sampling criterion, accurate interpolation of a bandlimited function requires a spacing \( \delta \leq \frac{1}{2\Omega} \). Solving for \( h \) yields:
\[
h \leq \frac{1}{2\Omega \cdot 0.269} \approx \frac{1.86}{\Omega}.
\]

Under this condition, the mesh resolution is sufficient to capture all meaningful spectral components of \( f \), and the interpolation error is dominated by the spectral tail \( \hat{f}(\xi, \eta) \) for \( \|\boldsymbol{\xi}\| > \Omega \). Since \( f \) is assumed effectively bandlimited, this tail is negligible, and interpolation converges uniformly:
\[
\tilde{f} \to f \quad \text{as} \quad h \to 0.
\]
\end{proof}

\begin{remark}
This theorem provides a rigorous foundation for using SBMT-generated equilateral meshes in interpolation-based image reconstruction. Because SBMT preserves the regularity of the base mesh even after boundary-aware retriangulation, it enables spatially consistent interpolation that respects both signal geometry and sampling theory.

In practical terms, this guarantees that for sufficiently smooth or low-pass signals, the SBMT mesh supports near-exact interpolation without ringing artifacts, aliasing, or degradation near boundaries—critical for applications in simulation, restoration, and geometry-aware filtering.
\end{remark}

\begin{lemma}[Error Bound for Triangular Interpolation]
Let $f : \mathbb{R}^2 \to \mathbb{R}$ be a continuous function with rapidly decaying Fourier transform $\hat{f}(\xi, \eta)$, and suppose there exists a cutoff frequency $\Omega > 0$ such that
\[
|\hat{f}(\xi, \eta)| < \varepsilon \quad \text{for all } \sqrt{\xi^2 + \eta^2} > \Omega.
\]
Let $\mathcal{T}_h$ be a structured equilateral triangular mesh with edge length $h < \frac{1.86}{\Omega}$. Then any piecewise polynomial interpolation $\tilde{f}$ over $\mathcal{T}_h$ satisfies:
\[
\| \tilde{f} - f \|_{L^\infty} < \varepsilon.
\]
\end{lemma}

\begin{proof}[Proof Sketch]
This follows directly from Theorem~\ref{thm:spectral_interp}, which shows that interpolation over equilateral triangular meshes achieves spectral accuracy when the resolution satisfies $h < \frac{1.86}{\Omega}$. Since the spectral tail of $f$ beyond $\Omega$ is bounded by $\varepsilon$, the resulting interpolation error is also bounded by $\varepsilon$.
\end{proof}

\begin{corollary}[Exponential Accuracy for Analytic Functions]
If $f$ is real-analytic—common in physical simulations and solutions to elliptic PDEs—then its Fourier spectrum decays exponentially. Consequently, for sufficiently small $h$, the interpolation error becomes exponentially small, rendering it negligible in practical applications.
\end{corollary}

\begin{proof}[Sketch of Justification]
Real-analytic functions have Fourier transforms that decay faster than any inverse polynomial, often exponentially as $|\hat{f}(\xi,\eta)| \lesssim e^{-\alpha \sqrt{\xi^2 + \eta^2}}$ for some $\alpha > 0$. As $h$ decreases, the cutoff frequency $\Omega$ effectively increases, pushing the spectral tail beyond $\Omega$ to be exponentially small. Thus, the interpolation error decays exponentially in $1/h$.
\end{proof}

\begin{proposition}[Recommended Edge Lengths for Structured Sampling]
Let $\mathcal{T}_h$ be a structured equilateral triangular mesh used to interpolate an effectively bandlimited signal with cutoff frequency $\Omega$. The following practical bounds on $h$ ensure spectral interpolation accuracy for typical signal types.
\end{proposition}

\begin{table}[H]
\centering
\caption{Recommended edge-length bounds $h$ for SBMT triangulation under typical signal conditions.}
\label{tab:h-guidelines}
\begin{tabular}{@{}lcc@{}}
\toprule
\textbf{Signal Condition} & \textbf{Cutoff Frequency $\Omega$} & \textbf{Recommended $h$} \\
\midrule
Sharp transitions (e.g., checkerboard, alias-prone) & $\pi$       & $h \lesssim 0.59$ \\
Natural images (antialiased or smoothed)            & $\pi/2$     & $h \lesssim 1.18$ \\
Conservative safety margin                           & N/A         & $h \le 0.5$ \\
\bottomrule
\end{tabular}
\end{table}

\begin{remark}[Interpretation]
The theoretical bound $h < \frac{1.86}{\Omega}$ ensures that the average sampling density over equilateral triangular meshes meets the Nyquist condition. 

For worst-case high-frequency inputs with discontinuities (e.g., a black-and-white checkerboard), the effective $\Omega \approx \pi$ requires $h \lesssim 0.59$. However, most real-world signals—natural images, sensor data, physical fields—are smooth or have been prefiltered, resulting in lower effective $\Omega$.

Thus, setting $h \le 0.5$ is a conservative yet robust choice, ensuring compatibility with a broad range of input signals. This makes the SBMT framework especially suitable for subpixel-accurate interpolation of bitmap-derived domains while retaining both geometric alignment and sampling-theoretic fidelity.
\end{remark}

\vspace{1em}

\section*{Appendix F\\Lookup Tables in SBMT}
\label{app:lookup_table}

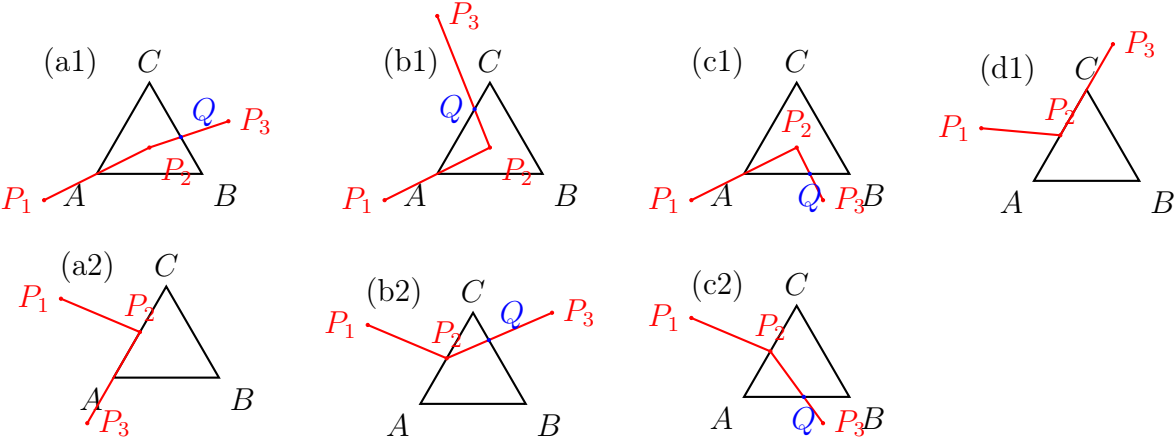
\begin{figure}[H]
\begin{center}
\renewcommand{\arraystretch}{1.5}
\begin{tabular}{cccc}
% 图 (a)
\begin{tikzpicture}[scale=0.7]
\node at (0.5, 2.1) {(a1)};  % 编号
\coordinate (A) at (1,0);
\coordinate (B) at (3,0);
\coordinate (C) at (2,1.732);
\draw[thick] (A) -- (B) -- (C) -- cycle;
\node[below left] at (A) {$A$};
\node[below right] at (B) {$B$};
\node[above] at (C) {$C$};

% 正确的 P1，使 P1P2 经过 A
\coordinate (P1) at (0,-0.5);
\coordinate (P2) at (2,0.5);
\coordinate (P3) at (3.5,1);
\draw[red, thick] (P1) -- (P2);
\draw[red, thick] (P2) -- (P3);
\filldraw[red] (P1) circle (0.03) node[left] {$P_1$};
\filldraw[red] (P2) circle (0.03) node[below right] {$P_2$};
\filldraw[red] (P3) circle (0.03) node[right] {$P_3$};

% 交点（根据原折线位置保持不变）
\coordinate (Q) at (2.597, 0.699);
\filldraw[blue] (Q) circle (0.03) node[above right] {$Q$};
\end{tikzpicture}
&
% 图 (b)
\begin{tikzpicture}[scale=0.7]
\node at (0.5, 2.1) {(b1)};  % 编号
\coordinate (A) at (1,0);
\coordinate (B) at (3,0);
\coordinate (C) at (2,1.732);
\draw[thick] (A) -- (B) -- (C) -- cycle;
\node[below left] at (A) {$A$};
\node[below right] at (B) {$B$};
\node[above] at (C) {$C$};

% 红折线段
\coordinate (P1) at (0,-0.5);
\coordinate (P2) at (2,0.5);
\coordinate (P3) at (1,3);
\draw[red, thick] (P1) -- (P2);
\draw[red, thick] (P2) -- (P3);
\filldraw[red] (P1) circle (0.03) node[left] {$P_1$};
\filldraw[red] (P2) circle (0.03) node[below right] {$P_2$};
\filldraw[red] (P3) circle (0.03) node[right] {$P_3$};

% 标注交点 Q
\coordinate (Q) at (1.709, 1.229);
\filldraw[blue] (Q) circle (0.03) node[left] {$Q$};
\end{tikzpicture}

&
% 图 (c)
\begin{tikzpicture}[scale=0.7]
\node at (0.5, 2.1) {(c1)};  % 编号
\coordinate (A) at (1,0);
\coordinate (B) at (3,0);
\coordinate (C) at (2,1.732);
\draw[thick] (A) -- (B) -- (C) -- cycle;
\node[below left] at (A) {$A$};
\node[below right] at (B) {$B$};
\node[above] at (C) {$C$};

% 红折线段
\coordinate (P1) at (0,-0.5);
\coordinate (P2) at (2,0.5);
\coordinate (P3) at (2.5,-0.5);
\draw[red, thick] (P1) -- (P2);
\draw[red, thick] (P2) -- (P3);
\filldraw[red] (P1) circle (0.03) node[left] {$P_1$};
\filldraw[red] (P2) circle (0.03) node[above] {$P_2$};
\filldraw[red] (P3) circle (0.03) node[right] {$P_3$};

% 交点 Q
\coordinate (Q) at (2.25, 0);
\filldraw[blue] (Q) circle (0.03) node[below] {$Q$};
\end{tikzpicture}
&
\begin{tikzpicture}[scale=0.7]
\node at (0.5, 2.1) {(d1)};  % 编号
\coordinate (A) at (1,0);
\coordinate (B) at (3,0);
\coordinate (C) at (2,1.732);
\draw[thick] (A) -- (B) -- (C) -- cycle;
\node[below left] at (A) {$A$};
\node[below right] at (B) {$B$};
\node[above] at (C) {$C$};

% 红折线段
\coordinate (P1) at (0,1);
\coordinate (P2) at (1.5,0.866);
\coordinate (P3) at (2.5,2.598);
\draw[red, thick] (P1) -- (P2);
\draw[red, thick] (P2) -- (P3);
\filldraw[red] (P1) circle (0.03) node[left] {$P_1$};
\filldraw[red] (P2) circle (0.03) node[above] {$P_2$};
\filldraw[red] (P3) circle (0.03) node[right] {$P_3$};

\end{tikzpicture}

\\
\begin{tikzpicture}[scale=0.7]
\node at (0.5, 2.1) {(a2)};  % 编号
\coordinate (A) at (1,0);
\coordinate (B) at (3,0);
\coordinate (C) at (2,1.732);
\draw[thick] (A) -- (B) -- (C) -- cycle;
\node[below left] at (A) {$A$};
\node[below right] at (B) {$B$};
\node[above] at (C) {$C$};

% 红折线段
\coordinate (P1) at (0,1.5);
\coordinate (P2) at (1.5,0.866);
\coordinate (P3) at (0.5,-0.866);
\draw[red, thick] (P1) -- (P2);
\draw[red, thick] (P2) -- (P3);
\filldraw[red] (P1) circle (0.03) node[left] {$P_1$};
\filldraw[red] (P2) circle (0.03) node[above] {$P_2$};
\filldraw[red] (P3) circle (0.03) node[right] {$P_3$};

\end{tikzpicture}
&
% 图 (b)
\begin{tikzpicture}[scale=0.7]
\node at (0.5, 2.1) {(b2)};  % 编号
\coordinate (A) at (1,0);
\coordinate (B) at (3,0);
\coordinate (C) at (2,1.732);
\draw[thick] (A) -- (B) -- (C) -- cycle;
\node[below left] at (A) {$A$};
\node[below right] at (B) {$B$};
\node[above] at (C) {$C$};

% 红折线段
\coordinate (P1) at (0,1.5);
\coordinate (P2) at (1.5,0.866);
\coordinate (P3) at (3.5,1.732);
\draw[red, thick] (P1) -- (P2);
\draw[red, thick] (P2) -- (P3);
\filldraw[red] (P1) circle (0.03) node[left] {$P_1$};
\filldraw[red] (P2) circle (0.03) node[above] {$P_2$};
\filldraw[red] (P3) circle (0.03) node[right] {$P_3$};

% 交点 Q
\coordinate (Q) at (2.3, 1.212);
\filldraw[blue] (Q) circle (0.03) node[above right] {$Q$};
\end{tikzpicture}

&
% 图 (c)
\begin{tikzpicture}[scale=0.7]
\node at (0.5, 2.1) {(c2)};  % 编号
\coordinate (A) at (1,0);
\coordinate (B) at (3,0);
\coordinate (C) at (2,1.732);
\draw[thick] (A) -- (B) -- (C) -- cycle;
\node[below left] at (A) {$A$};
\node[below right] at (B) {$B$};
\node[above] at (C) {$C$};

% 红折线段
\coordinate (P1) at (0,1.5);
\coordinate (P2) at (1.5,0.866);
\coordinate (P3) at (2.5,-0.5);
\draw[red, thick] (P1) -- (P2);
\draw[red, thick] (P2) -- (P3);
\filldraw[red] (P1) circle (0.03) node[left] {$P_1$};
\filldraw[red] (P2) circle (0.03) node[above] {$P_2$};
\filldraw[red] (P3) circle (0.03) node[right] {$P_3$};

% 交点 Q
\coordinate (Q) at (2.134, 0);
\filldraw[blue] (Q) circle (0.03) node[below] {$Q$};
\end{tikzpicture}

&

\end{tabular}
\end{center}

\caption{
Configurations of type \textbf{(1,2)}, where two boundary segments intersect a triangle with one and two points respectively. These patterns feature three embedded intersection points and support unified retriangulation rules across different edge locations.
}
\label{fig:1_2}
\end{figure}

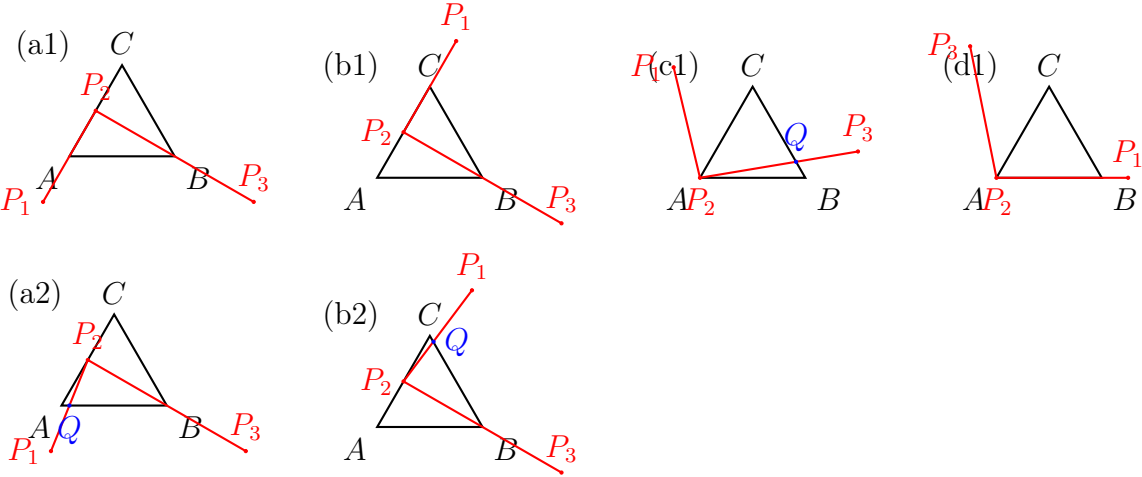
\begin{figure}[H]
\begin{center}
\renewcommand{\arraystretch}{1.5}
\begin{tabular}{cccc}
% 图 (a)
\begin{tikzpicture}[scale=0.7]
\node at (0.5, 2.1) {(a1)};  % 编号
\coordinate (A) at (1,0);
\coordinate (B) at (3,0);
\coordinate (C) at (2,1.732);
\draw[thick] (A) -- (B) -- (C) -- cycle;
\node[below left] at (A) {$A$};
\node[below right] at (B) {$B$};
\node[above] at (C) {$C$};

% 红折线段（调整后的）
\coordinate (P1) at (0.5,-0.866);  % 使 P1P2 经过 A
\coordinate (P2) at (1.5,0.866);
\coordinate (P3) at (4.5,-0.866);  % 使 P2P3 经过 B
\draw[red, thick] (P1) -- (P2);
\draw[red, thick] (P2) -- (P3);
\filldraw[red] (P1) circle (0.03) node[left] {$P_1$};
\filldraw[red] (P2) circle (0.03) node[above] {$P_2$};
\filldraw[red] (P3) circle (0.03) node[above] {$P_3$};
\end{tikzpicture}
&
% 图 (b)
\begin{tikzpicture}[scale=0.7]
\node at (0.5, 2.1) {(b1)};  % 编号
\coordinate (A) at (1,0);
\coordinate (B) at (3,0);
\coordinate (C) at (2,1.732);
\draw[thick] (A) -- (B) -- (C) -- cycle;
\node[below left] at (A) {$A$};
\node[below right] at (B) {$B$};
\node[above] at (C) {$C$};

% 红折线段：C 在 P1 和 P2 之间
\coordinate (P1) at (2.5,2.598);     % 调整后
\coordinate (P2) at (1.5,0.866);     % 保持不变
\coordinate (P3) at (4.5,-0.866);    % 保持不变
\draw[red, thick] (P1) -- (P2);
\draw[red, thick] (P2) -- (P3);
\filldraw[red] (P1) circle (0.03) node[above] {$P_1$};
\filldraw[red] (P2) circle (0.03) node[left] {$P_2$};
\filldraw[red] (P3) circle (0.03) node[above] {$P_3$};
\end{tikzpicture}
&
% 图 (c)
\begin{tikzpicture}[scale=0.7]
\node at (0.5, 2.1) {(c1)};  % 编号
\coordinate (A) at (1,0);
\coordinate (B) at (3,0);
\coordinate (C) at (2,1.732);
\draw[thick] (A) -- (B) -- (C) -- cycle;
\node[below left] at (A) {$A$};
\node[below right] at (B) {$B$};
\node[above] at (C) {$C$};

% 红线段
\coordinate (P1) at (0.5,2.1);
\coordinate (P2) at (1.0,0.0);
\coordinate (P3) at (4.0,0.5);
\draw[red, thick] (P1) -- (P2);
\draw[red, thick] (P3) -- (P2);
\filldraw[red] (P1) circle (0.03) node[left] {$P_1$};
\filldraw[red] (P2) circle (0.03) node[below] {$P_2$};
\filldraw[red] (P3) circle (0.03) node[above] {$P_3$};

% 标注交点 Q
\coordinate (Q) at (2.824, 0.304);
\filldraw[blue] (Q) circle (0.03) node[above] {$Q$};
\end{tikzpicture}
&
\begin{tikzpicture}[scale=0.7]
\node at (0.5, 2.1) {(d1)};  % 编号
\coordinate (A) at (1,0);
\coordinate (B) at (3,0);
\coordinate (C) at (2,1.732);
\draw[thick] (A) -- (B) -- (C) -- cycle;
\node[below left] at (A) {$A$};
\node[below right] at (B) {$B$};
\node[above] at (C) {$C$};

% 红线段
\coordinate (P1) at (3.5,0);
\coordinate (P2) at (1.0,0.0);
\coordinate (P3) at (0.5,2.5);
\draw[red, thick] (P1) -- (P2);
\draw[red, thick] (P3) -- (P2);
\filldraw[red] (P1) circle (0.03) node[above] {$P_1$};
\filldraw[red] (P2) circle (0.03) node[below] {$P_2$};
\filldraw[red] (P3) circle (0.03) node[left] {$P_3$};

\end{tikzpicture}

\\

\begin{tikzpicture}[scale=0.7]
\node at (0.5, 2.1) {(a2)};  % 编号
\coordinate (A) at (1,0);
\coordinate (B) at (3,0);
\coordinate (C) at (2,1.732);
\draw[thick] (A) -- (B) -- (C) -- cycle;
\node[below left] at (A) {$A$};
\node[below right] at (B) {$B$};
\node[above] at (C) {$C$};

% 红折线段（调整后的）
\coordinate (P1) at (0.8,-0.866);  % 使 P1P2 经过 A
\coordinate (P2) at (1.5,0.866);
\coordinate (P3) at (4.5,-0.866);  % 使 P2P3 经过 B
\draw[red, thick] (P1) -- (P2);
\draw[red, thick] (P2) -- (P3);
\filldraw[red] (P1) circle (0.03) node[left] {$P_1$};
\filldraw[red] (P2) circle (0.03) node[above] {$P_2$};
\filldraw[red] (P3) circle (0.03) node[above] {$P_3$};

% 添加交点标注
\coordinate (Q) at (1.15, 0);
\filldraw[blue] (Q) circle (0.03) node[below] {$Q$};
\end{tikzpicture}
&
% 图 (b)
\begin{tikzpicture}[scale=0.7]
\node at (0.5, 2.1) {(b2)};  % 编号
\coordinate (A) at (1,0);
\coordinate (B) at (3,0);
\coordinate (C) at (2,1.732);
\draw[thick] (A) -- (B) -- (C) -- cycle;
\node[below left] at (A) {$A$};
\node[below right] at (B) {$B$};
\node[above] at (C) {$C$};

% 红折线段：C 在 P1 和 P2 之间
\coordinate (P1) at (2.8,2.598);     % 调整后
\coordinate (P2) at (1.5,0.866);     % 保持不变
\coordinate (P3) at (4.5,-0.866);    % 保持不变
\draw[red, thick] (P1) -- (P2);
\draw[red, thick] (P2) -- (P3);
\filldraw[red] (P1) circle (0.03) node[above] {$P_1$};
\filldraw[red] (P2) circle (0.03) node[left] {$P_2$};
\filldraw[red] (P3) circle (0.03) node[above] {$P_3$};

% 添加交点标注
\coordinate (Q) at (2.07, 1.62);
\filldraw[blue] (Q) circle (0.03) node[right] {$Q$};
\end{tikzpicture}

\end{tabular}
\end{center}

\caption{
Typical (2,3) intersection patterns involving two boundary segments and their interaction with a single triangle. Classification is based on geometry to support complete, template-driven retriangulation.
}

\label{fig:2_3}
\end{figure}

\begin{figure}[H]
\begin{center}
\renewcommand{\arraystretch}{1.5}
\begin{tabular}{cccc}
% 图 (a)
\begin{tikzpicture}[scale=0.7]
\node at (0.5, 2.1) {(a1)};  % 编号
\coordinate (A) at (1,0);
\coordinate (B) at (3,0);
\coordinate (C) at (2,1.732);
\draw[thick] (A) -- (B) -- (C) -- cycle;
\node[below left] at (A) {$A$};
\node[below right] at (B) {$B$};
\node[above] at (C) {$C$};

% 红折线段，调整后的 P3，使 P2P3 经过 B
\coordinate (P1) at (0,-0.5);
\coordinate (P2) at (2,0.5);
\coordinate (P3) at (4,-0.5);
\draw[red, thick] (P1) -- (P2);
\draw[red, thick] (P2) -- (P3);
\filldraw[red] (P1) circle (0.03) node[left] {$P_1$};
\filldraw[red] (P2) circle (0.03) node[above] {$P_2$};
\filldraw[red] (P3) circle (0.03) node[right] {$P_3$};

\end{tikzpicture}
&
% 图 (b)
\begin{tikzpicture}[scale=0.7]
\node at (0.5, 2.1) {(b1)};  % 编号
\coordinate (A) at (1,0);
\coordinate (B) at (3,0);
\coordinate (C) at (2,1.732);
\draw[thick] (A) -- (B) -- (C) -- cycle;
\node[below left] at (A) {$A$};
\node[below right] at (B) {$B$};
\node[above] at (C) {$C$};

% 红折线段
\coordinate (P1) at (2,-0.5);
\coordinate (P2) at (1,1.0);
\coordinate (P3) at (3,1.0);
\draw[red, thick] (P1) -- (P2);
\draw[red, thick] (P2) -- (P3);
\filldraw[red] (P1) circle (0.03) node[below] {$P_1$};
\filldraw[red] (P2) circle (0.03) node[left] {$P_2$};
\filldraw[red] (P3) circle (0.03) node[right] {$P_3$};

% 四个交点
\coordinate (Q1) at (1.6667, 0);
\coordinate (Q2) at (1.309, 0.537);
\coordinate (Q3) at (2.423, 1.0);
\coordinate (Q4) at (1.577, 1.0);
\filldraw[blue] (Q1) circle (0.03) node[below] {$Q_1$};
\filldraw[blue] (Q2) circle (0.03) node[left] {$Q_2$};
\filldraw[blue] (Q3) circle (0.03) node[above right] {$Q_3$};
\filldraw[blue] (Q4) circle (0.03) node[above] {$Q_4$};
\end{tikzpicture}

&
% 图 (c)
\begin{tikzpicture}[scale=0.7]
\node at (0.5, 2.1) {(c1)};  % 编号
\coordinate (A) at (1,0);
\coordinate (B) at (3,0);
\coordinate (C) at (2,1.732);
\draw[thick] (A) -- (B) -- (C) -- cycle;
\node[below left] at (A) {$A$};
\node[below right] at (B) {$B$};
\node[above] at (C) {$C$};

% 红折线段
\coordinate (P1) at (2,-1);
\coordinate (P2) at (1.5,0.866);
\coordinate (P3) at (3.5,1);
\draw[red, thick] (P1) -- (P2);
\draw[red, thick] (P2) -- (P3);
\filldraw[red] (P1) circle (0.03) node[below] {$P_1$};
\filldraw[red] (P2) circle (0.03) node[above] {$P_2$};
\filldraw[red] (P3) circle (0.03) node[right] {$P_3$};

% 交点 Q1 与 AB
\coordinate (Q1) at (1.732, 0);
\filldraw[blue] (Q1) circle (0.03) node[below] {$Q_1$};

% 交点 Q2 与 BC
\coordinate (Q2) at (2.462, 0.930);
\filldraw[blue] (Q2) circle (0.03) node[above right] {$Q_2$};
\end{tikzpicture}
&
\begin{tikzpicture}[scale=0.7]
\node at (0.5, 2.1) {(d1)};  % 编号
\coordinate (A) at (1,0);
\coordinate (B) at (3,0);
\coordinate (C) at (2,1.732);
\draw[thick] (A) -- (B) -- (C) -- cycle;
\node[below left] at (A) {$A$};
\node[below right] at (B) {$B$};
\node[above] at (C) {$C$};

% 红折线段
\coordinate (P1) at (0.5,-0.866);
\coordinate (P2) at (1.5,0.866);
\coordinate (P3) at (2.5,2.598);
\draw[red, thick] (P1) -- (P2);
\draw[red, thick] (P2) -- (P3);
\filldraw[red] (P1) circle (0.03) node[left] {$P_1$};
\filldraw[red] (P2) circle (0.03) node[above] {$P_2$};
\filldraw[red] (P3) circle (0.03) node[right] {$P_3$};

\end{tikzpicture}

\\
\begin{tikzpicture}[scale=0.7]
\node at (0.5, 2.1) {(a2)};  % 编号
\coordinate (A) at (1,0);
\coordinate (B) at (3,0);
\coordinate (C) at (2,1.732);
\draw[thick] (A) -- (B) -- (C) -- cycle;
\node[below left] at (A) {$A$};
\node[below right] at (B) {$B$};
\node[above] at (C) {$C$};

% 红折线段
\coordinate (P1) at (3,1.5);
\coordinate (P2) at (1.5,0.866);
\coordinate (P3) at (0.5,-0.866);
\draw[red, thick] (P1) -- (P2);
\draw[red, thick] (P2) -- (P3);
\filldraw[red] (P1) circle (0.03) node[right] {$P_1$};
\filldraw[red] (P2) circle (0.03) node[above] {$P_2$};
\filldraw[red] (P3) circle (0.03) node[right] {$P_3$};

% 交点 Q (P1P2 ∩ BC)
\coordinate (Q) at (2.304, 1.205);
\filldraw[blue] (Q) circle (0.03) node[below] {$Q$};
\end{tikzpicture}
&
% 图 (b)
\begin{tikzpicture}[scale=0.7]
\node at (0.5, 2.1) {(b2)};  % 编号
\coordinate (A) at (1,0);
\coordinate (B) at (3,0);
\coordinate (C) at (2,1.732);
\draw[thick] (A) -- (B) -- (C) -- cycle;
\node[below left] at (A) {$A$};
\node[below right] at (B) {$B$};
\node[above] at (C) {$C$};

% 红折线段
\coordinate (P1) at (2.5,-0.5);
\coordinate (P2) at (1.5,0.866);
\coordinate (P3) at (0.5,-0.866);
\draw[red, thick] (P1) -- (P2);
\draw[red, thick] (P2) -- (P3);
\filldraw[red] (P1) circle (0.03) node[right] {$P_1$};
\filldraw[red] (P2) circle (0.03) node[above] {$P_2$};
\filldraw[red] (P3) circle (0.03) node[right] {$P_3$};

% 交点 Q
\coordinate (Q) at (2.134, 0);
\filldraw[blue] (Q) circle (0.03) node[below] {$Q$};
\end{tikzpicture}

&
% 图 (c)
\begin{tikzpicture}[scale=0.7]
\node at (0.5, 2.1) {(c2)};  % 编号
\coordinate (A) at (1,0);
\coordinate (B) at (3,0);
\coordinate (C) at (2,1.732);
\draw[thick] (A) -- (B) -- (C) -- cycle;
\node[below left] at (A) {$A$};
\node[below right] at (B) {$B$};
\node[above] at (C) {$C$};

% 红折线段
\coordinate (P1) at (0,1.5);
\coordinate (P2) at (1.5,0.866);
\coordinate (P3) at (2.5,-0.5);
\draw[red, thick] (P1) -- (P2);
\draw[red, thick] (P2) -- (P3);
\filldraw[red] (P1) circle (0.03) node[left] {$P_1$};
\filldraw[red] (P2) circle (0.03) node[above] {$P_2$};
\filldraw[red] (P3) circle (0.03) node[right] {$P_3$};

% 交点 Q
\coordinate (Q) at (2.134, 0);
\filldraw[blue] (Q) circle (0.03) node[below] {$Q$};
\end{tikzpicture}

&

\begin{tikzpicture}[scale=0.7]
\node at (0.5, 2.1) {(d2)};  % 编号
\coordinate (A) at (1,0);
\coordinate (B) at (3,0);
\coordinate (C) at (2,1.732);
\draw[thick] (A) -- (B) -- (C) -- cycle;
\node[below left] at (A) {$A$};
\node[below right] at (B) {$B$};
\node[above] at (C) {$C$};

% 红折线段
\coordinate (P1) at (3.0,2.598);
\coordinate (P2) at (1.0,0.866);
\coordinate (P3) at (2.5,-0.5);
\draw[red, thick] (P1) -- (P2);
\draw[red, thick] (P2) -- (P3);
\filldraw[red] (P1) circle (0.03) node[above right] {$P_1$};
\filldraw[red] (P2) circle (0.03) node[above] {$P_2$};
\filldraw[red] (P3) circle (0.03) node[right] {$P_3$};

% 交点 Q1 (P2P3 ∩ AC)
\coordinate (Q1) at (1.328, 0.568);
\filldraw[blue] (Q1) circle (0.03) node[left] {$Q_1$};

% 交点 Q2 (P2P3 ∩ AB)
\coordinate (Q2) at (1.951, 0);
\filldraw[blue] (Q2) circle (0.03) node[below] {$Q_2$};
\end{tikzpicture}

\\
\begin{tikzpicture}[scale=0.7]
\node at (0.5, 2.1) {(a3)};  % 编号
\coordinate (A) at (1,0);
\coordinate (B) at (3,0);
\coordinate (C) at (2,1.732);
\draw[thick] (A) -- (B) -- (C) -- cycle;
\node[below left] at (A) {$A$};
\node[below right] at (B) {$B$};
\node[above] at (C) {$C$};

% 红折线段，P2P3 经过 A，但 P3 不与 A 重合
\coordinate (P1) at (3.0,2.598);
\coordinate (P2) at (1.0,0.866);
\coordinate (P3) at (1.0,-0.866);  % 向 A 方向延长
\draw[red, thick] (P1) -- (P2);
\draw[red, thick] (P2) -- (P3);
\filldraw[red] (P1) circle (0.03) node[above right] {$P_1$};
\filldraw[red] (P2) circle (0.03) node[above] {$P_2$};
\filldraw[red] (P3) circle (0.03) node[below right] {$P_3$};
\end{tikzpicture}
&
% 图 (b)
\begin{tikzpicture}[scale=0.7]
\node at (0.5, 2.1) {(b3)};  % 编号
\coordinate (A) at (1,0);
\coordinate (B) at (3,0);
\coordinate (C) at (2,1.732);
\draw[thick] (A) -- (B) -- (C) -- cycle;
\node[below left] at (A) {$A$};
\node[below right] at (B) {$B$};
\node[above] at (C) {$C$};

% 红折线段
\coordinate (P1) at (3,1.7);
\coordinate (P2) at (1.5,0.866);
\coordinate (P3) at (3.5,0.5);
\draw[red, thick] (P1) -- (P2);
\draw[red, thick] (P2) -- (P3);
\filldraw[red] (P1) circle (0.03) node[right] {$P_1$};
\filldraw[red] (P2) circle (0.03) node[left] {$P_2$};
\filldraw[red] (P3) circle (0.03) node[right] {$P_3$};

% 交点 Q1 与 BC
\coordinate (Q1) at (2.256, 1.286);
\filldraw[blue] (Q1) circle (0.03) node[above right] {$Q_1$};

% 交点 Q2 与 BC
\coordinate (Q2) at (2.618, 0.661);
\filldraw[blue] (Q2) circle (0.03) node[below right] {$Q_2$};
\end{tikzpicture}

&
% 图 (c)
\begin{tikzpicture}[scale=0.7]
\node at (0.5, 2.1) {(c3)};  % 编号
\coordinate (A) at (1,0);
\coordinate (B) at (3,0);
\coordinate (C) at (2,1.732);
\draw[thick] (A) -- (B) -- (C) -- cycle;
\node[below left] at (A) {$A$};
\node[below right] at (B) {$B$};
\node[above] at (C) {$C$};

% 红折线段
\coordinate (P1) at (1,-0.5);
\coordinate (P2) at (1.5,0.866);
\coordinate (P3) at (3.0,-0.5);
\draw[red, thick] (P1) -- (P2);
\draw[red, thick] (P2) -- (P3);
\filldraw[red] (P1) circle (0.03) node[right] {$P_1$};
\filldraw[red] (P2) circle (0.03) node[left] {$P_2$};
\filldraw[red] (P3) circle (0.03) node[right] {$P_3$};

% 交点 Q1 与 AB
\coordinate (Q1) at (1.183, 0);
\filldraw[blue] (Q1) circle (0.03) node[left] {$Q_1$};

% 交点 Q2 与 AB
\coordinate (Q2) at (2.451, 0);
\filldraw[blue] (Q2) circle (0.03) node[right] {$Q_2$};
\end{tikzpicture}
&

\begin{tikzpicture}[scale=0.7]
\node at (0.5, 2.1) {(d3)};  % 编号
\coordinate (A) at (1,0);
\coordinate (B) at (3,0);
\coordinate (C) at (2,1.732);
\draw[thick] (A) -- (B) -- (C) -- cycle;
\node[below left] at (A) {$A$};
\node[below right] at (B) {$B$};
\node[above] at (C) {$C$};

% 红折线段
\coordinate (P1) at (3,1.7);
\coordinate (P2) at (1.3,0.866);
\coordinate (P3) at (3.5,0.5);
\draw[red, thick] (P1) -- (P2);
\draw[red, thick] (P2) -- (P3);
\filldraw[red] (P1) circle (0.03) node[right] {$P_1$};
\filldraw[red] (P2) circle (0.03) node[left] {$P_2$};
\filldraw[red] (P3) circle (0.03) node[right] {$P_3$};

% Q1 = P1P2 ∩ BC
\coordinate (Q1) at (2.234, 1.327);
\filldraw[blue] (Q1) circle (0.03) node[above right] {$Q_3$};

% Q2 = P2P3 ∩ BC
\coordinate (Q2) at (2.629, 0.645);
\filldraw[blue] (Q2) circle (0.03) node[below right] {$Q_4$};

% 交点 Q3 与 AC
\coordinate (Q3) at (1.579, 1.007);
\filldraw[blue] (Q3) circle (0.03) node[above left] {$Q_1$};

% 交点 Q4 与 AC
\coordinate (Q4) at (1.482, 0.835);
\filldraw[blue] (Q4) circle (0.03) node[below left] {$Q_2$};
\end{tikzpicture}

\\
\begin{tikzpicture}[scale=0.7]
\node at (0.5, 2.1) {(a4)};  % 编号
\coordinate (A) at (1,0);
\coordinate (B) at (3,0);
\coordinate (C) at (2,1.732);
\draw[thick] (A) -- (B) -- (C) -- cycle;
\node[below left] at (A) {$A$};
\node[below right] at (B) {$B$};
\node[above] at (C) {$C$};

% 红折线段
\coordinate (P1) at (1,-0.5);
\coordinate (P2) at (1.3,0.866);
\coordinate (P3) at (3.0,-0.5);
\draw[red, thick] (P1) -- (P2);
\draw[red, thick] (P2) -- (P3);
\filldraw[red] (P1) circle (0.03) node[right] {$P_1$};
\filldraw[red] (P2) circle (0.03) node[left] {$P_2$};
\filldraw[red] (P3) circle (0.03) node[right] {$P_3$};

% 交点 Q1 与 AB
\coordinate (Q1) at (1.110, 0);
\filldraw[blue] (Q1) circle (0.03) node[left] {$Q_3$};

% 交点 Q2 与 AB
\coordinate (Q2) at (2.376, 0);
\filldraw[blue] (Q2) circle (0.03) node[right] {$Q_4$};

% 交点 Q3 与 AC
\coordinate (Q3) at (1.177, 0.308);
\filldraw[blue] (Q3) circle (0.03) node[right] {$Q_1$};

% 交点 Q4 与 AC
\coordinate (Q4) at (1.437, 0.757);
\filldraw[blue] (Q4) circle (0.03) node[right] {$Q_2$};
\end{tikzpicture}
&
% 图 (b)

&
% 图 (c)
&

\end{tabular}
\end{center}

\caption{
Representative configurations from the $(2,2)$ interaction class, where two consecutive boundary segments intersect a triangle with two intersection points each. These high-complexity cases demonstrate key subtypes including double edge crossings, vertex piercing, and corner-touching patterns. The illustrated cases are sampled from the complete lookup set, and each corresponds to a unique template-driven retriangulation rule. Full tabulation and rule assignments are provided in Appendix F of the Supplementary Material.
}

\label{fig:2_2}
\end{figure}
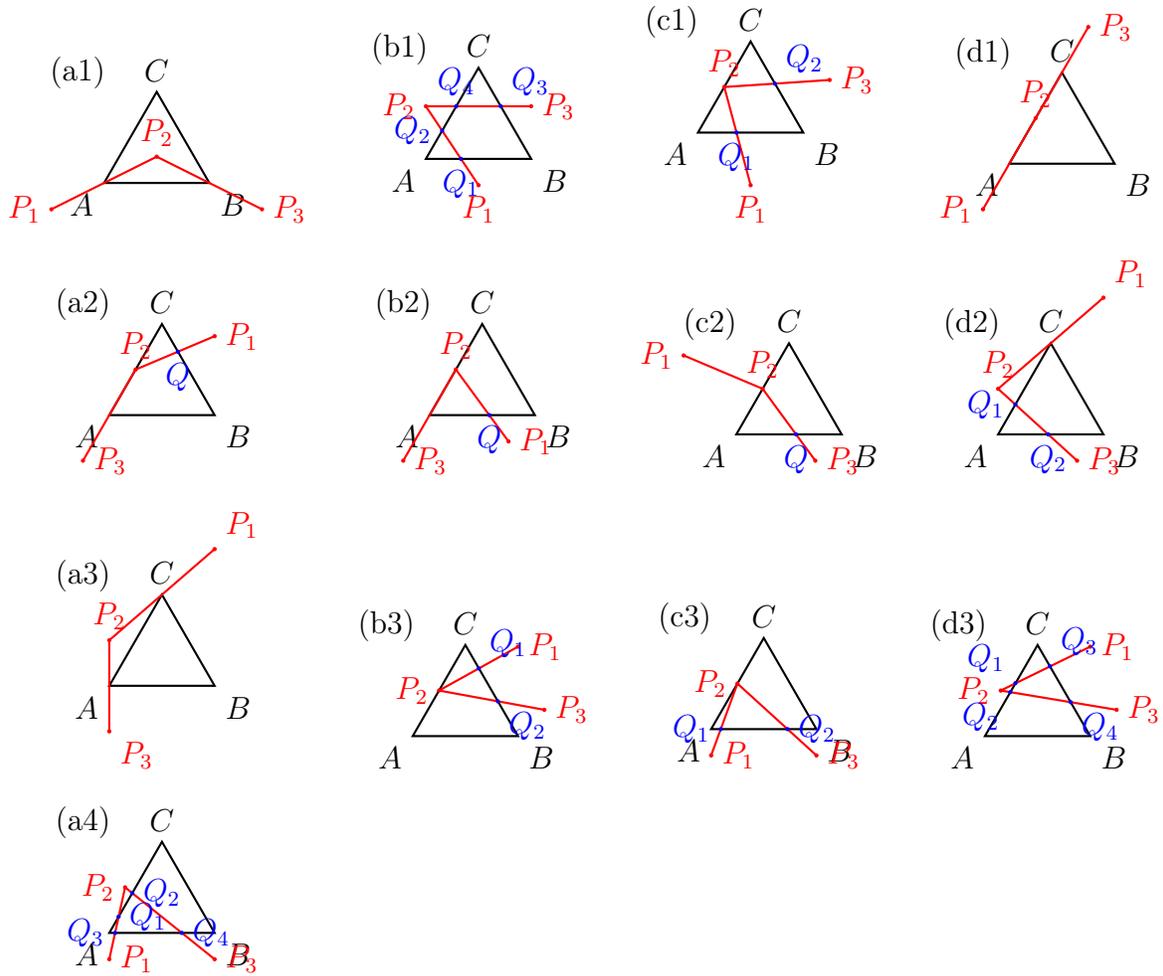

\begin{figure}[H]
\begin{center}
\renewcommand{\arraystretch}{1.5}
\begin{tabular}{cccc}
% 图 (a)
\begin{tikzpicture}[scale=0.7]
\coordinate (A) at (1,0);
\coordinate (B) at (3,0);
\coordinate (C) at (2,1.732);
\draw[thick] (A) -- (B) -- (C) -- cycle;
\node[below left] at (A) {$A$};
\node[below right] at (B) {$B$};
\node[above] at (C) {$C$};

% 红折线段（调整后的）
\coordinate (P1) at (0.5,1.866);  % 使 P1P2 经过 A
\coordinate (P2) at (1.5,0.866);
\coordinate (P3) at (4.5,-0.866);  % 使 P2P3 经过 B
\draw[red, thick] (P1) -- (P2);
\draw[red, thick] (P2) -- (P3);
\filldraw[red] (P1) circle (0.03) node[left] {$P_1$};
\filldraw[red] (P2) circle (0.03) node[above] {$P_2$};
\filldraw[red] (P3) circle (0.03) node[right] {$P_3$};
\end{tikzpicture}

\end{tabular}
\end{center}

\caption{
Representative configuration of type \textbf{(1,3)}, involving one and three intersections from two adjacent boundary segments. This case has a unique retriangulation and requires no special handling.
}
\label{fig:1_3}
\end{figure}
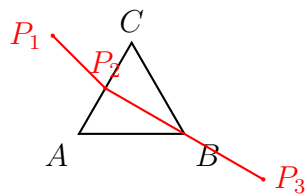

\begin{figure}[H]
\begin{center}
\renewcommand{\arraystretch}{1.5}
\begin{tabular}{cccc}

% 图 (a)
\begin{tikzpicture}[scale=0.7]
\node at (0.5, 2.1) {(a)};  % 编号
\coordinate (A) at (1,0);
\coordinate (B) at (3,0);
\coordinate (C) at (2,1.732);
\draw[thick] (A) -- (B) -- (C) -- cycle;
\node[below left] at (A) {$A$};
\node[below right] at (B) {$B$};
\node[above] at (C) {$C$};

% 红折线段（P1 在 CA 延长线 A 的外侧，P3 在 CB 延长线 B 的外侧）
\coordinate (P1) at (0, -1.732);   % 延长 CA 出 A
\coordinate (P2) at (2, 1.732);    % C 点
\coordinate (P3) at (4, -1.732);   % 延长 CB 出 B
\draw[red, thick] (P1) -- (P2);
\draw[red, thick] (P2) -- (P3);
\filldraw[red] (P1) circle (0.03) node[left] {$P_1$};
\filldraw[red] (P2) circle (0.03) node[below] {$P_2$};
\filldraw[red] (P3) circle (0.03) node[right] {$P_3$};
\end{tikzpicture}
&
\begin{tikzpicture}[scale=0.7]
\node at (0.5, 2.1) {(b)};  % 编号
\coordinate (A) at (1,0);
\coordinate (B) at (3,0);
\coordinate (C) at (2,1.732);
\draw[thick] (A) -- (B) -- (C) -- cycle;
\node[below left] at (A) {$A$};
\node[below right] at (B) {$B$};
\node[above] at (C) {$C$};

% 红折线段
\coordinate (P1) at (4.5,1.866);  
\coordinate (P2) at (1,0);
\coordinate (P3) at (4.5,0.866);  
\draw[red, thick] (P1) -- (P2);
\draw[red, thick] (P2) -- (P3);
\filldraw[red] (P1) circle (0.03) node[left] {$P_1$};
\filldraw[red] (P2) circle (0.03) node[above] {$P_2$};
\filldraw[red] (P3) circle (0.03) node[right] {$P_3$};

% 交点标注
\coordinate (Q1) at (2.53, 0.815);
\coordinate (Q2) at (2.75, 0.433);
\filldraw[blue] (Q1) circle (0.03) node[above right] {$Q_1$};
\filldraw[blue] (Q2) circle (0.03) node[below right] {$Q_2$};

\end{tikzpicture}
&
\begin{tikzpicture}[scale=0.7]
\node at (0.5, 2.1) {(c)};  % 编号
\coordinate (A) at (1,0);
\coordinate (B) at (3,0);
\coordinate (C) at (2,1.732);
\draw[thick] (A) -- (B) -- (C) -- cycle;
\node[below left] at (A) {$A$};
\node[below right] at (B) {$B$};
\node[above] at (C) {$C$};

% 红折线段
\coordinate (P1) at (4.5,1.866);  
\coordinate (P2) at (1,0);
\coordinate (P3) at (4.5,0);  
\draw[red, thick] (P1) -- (P2);
\draw[red, thick] (P2) -- (P3);
\filldraw[red] (P1) circle (0.03) node[left] {$P_1$};
\filldraw[red] (P2) circle (0.03) node[above] {$P_2$};
\filldraw[red] (P3) circle (0.03) node[right] {$P_3$};

% 交点标注
\coordinate (Q) at (2.53, 0.815);
\filldraw[blue] (Q) circle (0.03) node[above right] {$Q$};

\end{tikzpicture}

\\
\begin{tikzpicture}[scale=0.7]
\node at (0.5, 2.1) {(d)};  % 编号
\coordinate (A) at (1,0);
\coordinate (B) at (3,0);
\coordinate (C) at (2,1.732);
\draw[thick] (A) -- (B) -- (C) -- cycle;
\node[below left] at (A) {$A$};
\node[below right] at (B) {$B$};
\node[above] at (C) {$C$};

% 红折线段
\coordinate (P1) at (3,3.464);  % 修改后的P1，位于AC延长线外
\coordinate (P2) at (1,0);
\coordinate (P3) at (4.5,0.866);  
\draw[red, thick] (P1) -- (P2);
\draw[red, thick] (P2) -- (P3);
\filldraw[red] (P1) circle (0.03) node[left] {$P_1$};
\filldraw[red] (P2) circle (0.03) node[above] {$P_2$};
\filldraw[red] (P3) circle (0.03) node[right] {$P_3$};

% 交点标注
\coordinate (Q) at (2.75, 0.433);
\filldraw[blue] (Q) circle (0.03) node[below right] {$Q$};

\end{tikzpicture}

\end{tabular}
\end{center}

\caption{
Illustrative (3,3) cases, where two segments intersect a triangle at three points each.
Though excluded by SBMT’s $90^\circ$ angle constraint, these edge cases are shown for theoretical completeness.
They verify that remeshing remains valid under sharp junctions, as long as no triangle is intersected by more than two segments.
}
\label{fig:3_3}
\end{figure}
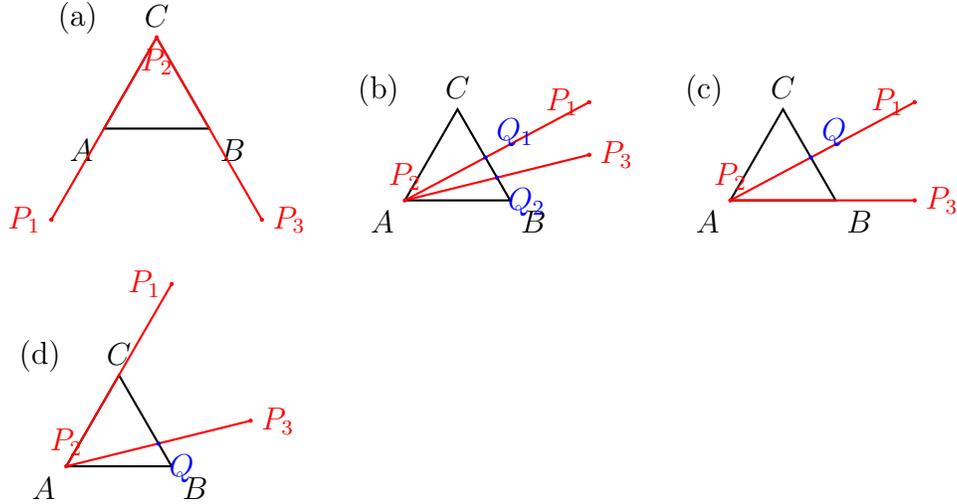

\section*{Appendix G\\Retriangulation Template Rules in SBMT}
\label{app:retriangulation_templates}

This section documents the specific retriangulation rules for each atomic intersection configuration described in Section~3.3.2 of the main paper. Each configuration, defined by its intersection count and geometric layout (e.g., type (1,2), (2,3), etc.), corresponds to a unique sub-triangle subdivision pattern.

\noindent\textbf{Terminology note.}
Throughout this appendix, ``configuration'' refers to an \emph{atomic canonical lookup key} $\kappa(K)$ (including edge attribution, ordered intersection records, and $D_3$ canonicalization), not merely the coarse counting class such
as $(1,1)$ or $(2,2)$.

\paragraph{Embedding Consistency}
Across all intersection configurations discussed in this section and the Appendix, the retriangulation rule for each triangle cell is constructed to explicitly embed the two boundary segments $P_1P_2$ and $P_2P_3$ that traverse it. This embedding preserves the exact geometry of the input polygonal boundary, ensuring that the reconstructed domain faithfully reflects both the shape and topology of the original signal contours.

\subsection*{Single-Segment Intersection Configurations}

Figure~3 of the main paper enumerates all possible configurations where a single boundary segment intersects a triangle cell. Depending on the number and location of intersection points, the retriangulation strategies are summarized as follows:

\paragraph{Cases (a) and (d): Two or three intersections, but no retriangulation required}
\begin{itemize}
  \item \textbf{Description:} These configurations involve no actual entry into the triangle or only a glancing contact.
  \item \textbf{Strategy:} The triangle remains unchanged. The segment lies entirely outside the triangle or intersects only a vertex or a single edge without penetrating the interior.
\end{itemize}

\paragraph{Case (b)}
\begin{itemize}
  \item \textbf{Intersection points:} $Q_1$ (on edge $BC$), $Q_2$ (on edge $CA$)
  \item \textbf{Strategy:} Introduce $Q_1$ and $Q_2$ as new vertices. Subdivide the triangle into three sub-triangles:
  \begin{enumerate}
    \item $\triangle Q_2 Q_1 C$
    \item $\triangle Q_2 A B$
    \item $\triangle Q_1 Q_2 B$
  \end{enumerate}
\end{itemize}

\paragraph{Case (c)}
\begin{itemize}
  \item \textbf{Intersection point:} $Q$ (on edge $BC$)
  \item \textbf{Strategy:} Introduce $Q$ as a new vertex. Subdivide the triangle into two sub-triangles:
  \begin{enumerate}
    \item $\triangle A Q C$
    \item $\triangle A B Q$
  \end{enumerate}
\end{itemize}

\subsection*{(1,1) Type Configurations}

Figure~4 illustrates three representative configurations of type \textbf{(1,1)}, where two boundary segments each intersect the triangle at a single point. In all cases, local retriangulation is required to ensure boundary conformity and maintain topological consistency.

\paragraph{Case (a)}
\begin{itemize}
  \item \textbf{Intersection points:} $Q_1$ (on edge $AC$), $Q_2$ (on edge $BC$)
  \item \textbf{Strategy:} Introduce $Q_1$ and $Q_2$ as new vertices. Subdivide the triangle into five sub-triangles:
  \begin{enumerate}
    \item $\triangle Q_1 P_2 C$
    \item $\triangle P_2 Q_2 C$
    \item $\triangle Q_1 A P_2$
    \item $\triangle A B P_2$
    \item $\triangle Q_2 P_2 B$
  \end{enumerate}
\end{itemize}

\paragraph{Case (b)}
\begin{itemize}
  \item \textbf{Strategy:} Subdivide the triangle into two sub-triangles:
  \begin{enumerate}
    \item $\triangle C P_2 B$
    \item $\triangle P_2 A B$
  \end{enumerate}
\end{itemize}

\paragraph{Case (c)}
\begin{itemize}
  \item \textbf{Intersection points:} $Q_1$ (on edge $AC$), $Q_2$ (on edge $AC$)
  \item \textbf{Strategy:} Introduce $Q_1$ and $Q_2$ as new vertices. Subdivide the triangle into five sub-triangles:
  \begin{enumerate}
    \item $\triangle Q_1 P_2 C$
    \item $\triangle P_2 Q_1 Q_2$
    \item $\triangle Q_2 A P_2$
    \item $\triangle B C P_2$
    \item $\triangle A B P_2$
  \end{enumerate}
\end{itemize}

\subsection*{(1,2) Type Configurations}

Figure~5 of the main paper shows several (1,2)-type configurations. We provide below the retriangulation strategies for selected cases.

\paragraph{Case (a1)}
\begin{itemize}
  \item \textbf{Intersection point:} $Q$ (on edge $BC$)
  \item \textbf{Strategy:} Introduce $Q$ as a new vertex. Subdivide the triangle into four sub-triangles:
  \begin{enumerate}
    \item $\triangle C A P_2$
    \item $\triangle A B P_2$
    \item $\triangle C P_2 Q$
    \item $\triangle P_2 B Q$
  \end{enumerate}
\end{itemize}

\paragraph{Case (b1)}
\begin{itemize}
  \item \textbf{Intersection point:} $Q$ (on edge $BC$)
  \item \textbf{Strategy:} Introduce $Q$ as a new vertex. Subdivide the triangle into four sub-triangles:
  \begin{enumerate}
    \item $\triangle A P_2 Q$
    \item $\triangle C Q P_2$
    \item $\triangle C P_2 B$
    \item $\triangle A B P_2$
  \end{enumerate}
\end{itemize}

\paragraph{Case (c1)}
\begin{itemize}
  \item \textbf{Intersection point:} $Q$ (on edge $BC$)
  \item \textbf{Strategy:} Introduce $Q$ as a new vertex. Subdivide the triangle into four sub-triangles:
  \begin{enumerate}
    \item $\triangle A P_2 C$
    \item $\triangle A Q P_2$
    \item $\triangle B P_2 Q$
    \item $\triangle B C P_2$
  \end{enumerate}
\end{itemize}

\paragraph{Case (d1)}
\begin{itemize}
  \item \textbf{Strategy:} Subdivide the triangle into two sub-triangles:
  \begin{enumerate}
    \item $\triangle A B P_2$
    \item $\triangle B C P_2$
  \end{enumerate}
\end{itemize}

\paragraph{Case (a2)}
\begin{itemize}
  \item \textbf{Strategy:} Subdivide the triangle into two sub-triangles:
  \begin{enumerate}
    \item $\triangle A B P_2$
    \item $\triangle B C P_2$
  \end{enumerate}
\end{itemize}

\paragraph{Case (b2)}
\begin{itemize}
  \item \textbf{Intersection point:} $Q$ (on edge $BC$)
  \item \textbf{Strategy:} Introduce $Q$ as a new vertex. Subdivide the triangle into three sub-triangles:
  \begin{enumerate}
    \item $\triangle P_2 Q C$
    \item $\triangle B Q P_2$
    \item $\triangle A B P_2$
  \end{enumerate}
\end{itemize}

\paragraph{Case (c2)}
\begin{itemize}
  \item \textbf{Intersection point:} $Q$ (on edge $AB$)
  \item \textbf{Strategy:} Introduce $Q$ as a new vertex. Subdivide the triangle into three sub-triangles:
  \begin{enumerate}
    \item $\triangle P_2 A Q$
    \item $\triangle Q B P_2$
    \item $\triangle B C P_2$
  \end{enumerate}
\end{itemize}

\subsection*{(2,2) Type Configurations}

Figure~6 of the main paper presents representative configurations of type \textbf{(2,2)}, where two boundary segments each intersect a triangle at two distinct points. These configurations involve four embedded intersection points and represent the most complex local scenarios addressed in our remeshing framework. The retriangulation strategy for each case is described below.

\paragraph{Case (a1)}
\begin{itemize}
  \item \textbf{Strategy:} Subdivide the triangle into three sub-triangles:
  \begin{enumerate}
    \item $\triangle C A P_2$
    \item $\triangle A B P_2$
    \item $\triangle C P_2 B$
  \end{enumerate}
\end{itemize}

\paragraph{Case (b1)}\label{app:case-b1}
\begin{itemize}
  \item \textbf{Intersection points:} $Q_1$ (on edge $AB$), $Q_2$ (on edge $CA$), $Q_3$ (on edge $BC$), $Q_4$ (on edge $CA$)
  \item \textbf{Strategy:} Introduce $Q_1$, $Q_2$, $Q_3$, and $Q_4$ as new vertices. Subdivide the triangle into five sub-triangles:
  \begin{enumerate}
    \item $\triangle Q_1 Q_2 A$
    \item $\triangle Q_1 Q_4 Q_2$
    \item $\triangle Q_1 Q_3 Q_4$
    \item $\triangle Q_1 B Q_3$
    \item $\triangle Q_4 Q_3 C$
  \end{enumerate}
\end{itemize}

\paragraph{Case (c1)}
\begin{itemize}
  \item \textbf{Intersection points:} $Q_1$ (on edge $AB$), $Q_2$ (on edge $BC$)
  \item \textbf{Strategy:} Introduce $Q_1$ and $Q_2$ as new vertices. Subdivide the triangle into four sub-triangles:
  \begin{enumerate}
    \item $\triangle A Q_1 P_2$
    \item $\triangle P_2 Q_2 C$
    \item $\triangle P_2 Q_1 Q_2$
    \item $\triangle Q_1 B Q_2$
  \end{enumerate}
\end{itemize}

\paragraph{Case (d1)}
\begin{itemize}
  \item \textbf{Strategy:} Subdivide the triangle into two sub-triangles:
  \begin{enumerate}
    \item $\triangle B C P_2$
    \item $\triangle A B P_2$
  \end{enumerate}
\end{itemize}

\paragraph{Case (a2)}
\begin{itemize}
  \item \textbf{Intersection point:} $Q$ (on edge $BC$)
  \item \textbf{Strategy:} Introduce $Q$ as a new vertex. Subdivide the triangle into three sub-triangles:
  \begin{enumerate}
    \item $\triangle P_2 Q C$
    \item $\triangle B Q P_2$
    \item $\triangle A B P_2$
  \end{enumerate}
\end{itemize}

\paragraph{Case (b2)}
\begin{itemize}
  \item \textbf{Intersection point:} $Q$ (on edge $AB$)
  \item \textbf{Strategy:} Introduce $Q$ as a new vertex. Subdivide the triangle into three sub-triangles:
  \begin{enumerate}
    \item $\triangle P_2 A Q$
    \item $\triangle Q B P_2$
    \item $\triangle B C P_2$
  \end{enumerate}
\end{itemize}

\paragraph{Case (c2)}
\begin{itemize}
  \item \textbf{Intersection point:} $Q$ (on edge $AB$)
  \item \textbf{Strategy:} Introduce $Q$ as a new vertex. Subdivide the triangle into three sub-triangles:
  \begin{enumerate}
    \item $\triangle P_2 A Q$
    \item $\triangle Q B P_2$
    \item $\triangle B C P_2$
  \end{enumerate}
\end{itemize}

\paragraph{Case (d2)}
\begin{itemize}
  \item \textbf{Intersection points:} $Q_1$ (on edge $CA$), $Q_2$ (on edge $AB$)
  \item \textbf{Strategy:} Introduce $Q_1$ and $Q_2$ as new vertices. Subdivide the triangle into three sub-triangles:
  \begin{enumerate}
    \item $\triangle A Q_2 Q_1$
    \item $\triangle Q_2 B Q_1$
    \item $\triangle Q_1 B C$
  \end{enumerate}
\end{itemize}

\paragraph{Case (a3)}
\begin{itemize}
  \item \textbf{Strategy:} This configuration permits the boundary segments to be embedded within the original triangle without modifying its mesh topology. No retriangulation is required, as the existing triangle already accommodates the segment trajectories.
\end{itemize}

\paragraph{Case (b3)}
\begin{itemize}
  \item \textbf{Intersection points:} $Q_1$ (on edge $BC$), $Q_2$ (on edge $BC$)
  \item \textbf{Strategy:} Introduce $Q_1$ and $Q_2$ as new vertices. Subdivide the triangle into four sub-triangles:
  \begin{enumerate}
    \item $\triangle P_2 Q_1 C$
    \item $\triangle P_2 Q_2 Q_1$
    \item $\triangle P_2 A B$
    \item $\triangle P_2 B Q_2$
  \end{enumerate}
\end{itemize}

\paragraph{Case (c3)}
\begin{itemize}
  \item \textbf{Intersection points:} $Q$ (on edge $AB$)
  \item \textbf{Strategy:} Introduce $Q$ as new vertices. Subdivide the triangle into four sub-triangles:
  \begin{enumerate}
    \item $\triangle Q_1 P_1 P_2$
    \item $\triangle P_1 Q_2 P_2$
    \item $\triangle P_2 Q_2 B$
    \item $\triangle P_2 B C$
  \end{enumerate}
\end{itemize}

\paragraph{Case (d3)}
\begin{itemize}
  \item \textbf{Intersection points:} $Q_1$ (on edge $CA$), $Q_2$ (on edge $CA$), $Q_3$ (on edge $BC$), $Q_4$ (on edge $BC$)
  \item \textbf{Strategy:} Introduce $Q_1$, $Q_2$, $Q_3$, and $Q_4$ as new vertices. Subdivide the triangle into five sub-triangles:
  \begin{enumerate}
    \item $\triangle A B Q_2$
    \item $\triangle Q_2 B Q_4$
    \item $\triangle Q_1 Q_2 Q_4$
    \item $\triangle Q_1 Q_4 Q_3$
    \item $\triangle Q_1 Q_3 C$
  \end{enumerate}
\end{itemize}

\paragraph{Case (a4)}
\begin{itemize}
  \item \textbf{Intersection points:} $Q_1$ (on edge $CA$), $Q_2$ (on edge $CA$), $Q_3$ (on edge $AB$), $Q_4$ (on edge $AB$)
  \item \textbf{Strategy:} Introduce $Q_1$, $Q_2$, $Q_3$, and $Q_4$ as new vertices. Subdivide the triangle into five sub-triangles:
  \begin{enumerate}
    \item $\triangle Q_3 Q_1 A$
    \item $\triangle Q_4 Q_1 Q_3$
    \item $\triangle Q_1 Q_4 Q_2$
    \item $\triangle Q_2 Q_4 B$
    \item $\triangle Q_2 B C$
  \end{enumerate}
\end{itemize}

\subsection*{(2,3) Type Configurations}

Figure~7 of the main paper presents several representative (2,3)-type configurations. These scenarios arise when two boundary segments intersect a triangle with two and three distinct points respectively. The following enumerates retriangulation strategies for selected cases.

\paragraph{Case (a1)}
\begin{itemize}
  \item \textbf{Strategy:} Subdivide the triangle into two sub-triangles:
  \begin{enumerate}
    \item $\triangle A B P_2$
    \item $\triangle P_2 B C$
  \end{enumerate}
\end{itemize}

\paragraph{Case (b1)}
\begin{itemize}
  \item \textbf{Strategy:} Subdivide into two sub-triangles:
  \begin{enumerate}
    \item $\triangle A B P_2$
    \item $\triangle P_2 B C$
  \end{enumerate}
\end{itemize}

\paragraph{Case (c1)}
\begin{itemize}
  \item \textbf{Intersection points:} $Q$ (on $BC$), $P_2$ coincides with vertex $A$.
  \item \textbf{Strategy:} Introduce $Q$ as new vertices. Subdivide the triangle into two sub-triangles:
  \begin{enumerate}
    \item $\triangle A B Q$
    \item $\triangle A Q C$
  \end{enumerate}
\end{itemize}

\paragraph{Case (d1)}
\begin{itemize}
  \item \textbf{Intersection points:} $P_2$ coincides with vertex $A$
  \item \textbf{Strategy:} No retriangulation is required. The five intersection points lie exactly along the triangle's boundary and do not introduce any ambiguity in region partitioning.
\end{itemize}

\paragraph{Case (a2)}
\begin{itemize}
  \item \textbf{Strategy:} Subdivide the triangle into three sub-triangles:
  \begin{enumerate}
    \item $\triangle A Q P_2$
    \item $\triangle P_2 B C$
    \item $\triangle P_2 Q B$
  \end{enumerate}
\end{itemize}

\paragraph{Case (b2)}
\begin{itemize}
  \item \textbf{Strategy:} Subdivide into three sub-triangles:
  \begin{enumerate}
    \item $\triangle A B P_2$
    \item $\triangle P_2 B Q$
    \item $\triangle P_2 Q C$
  \end{enumerate}
\end{itemize}

\subsection*{(1,3) Type Configurations}

Figure~8 of the main paper illustrates the sole (1,3)-type configuration observed in our remeshing framework. One boundary segment intersects the triangle at a single point, while the other intersects it at three distinct locations.

\paragraph{Only one case}
\begin{itemize}
  \item \textbf{Strategy:} Subdivide the triangle into two sub-triangles:
  \begin{enumerate}
    \item $\triangle A B P_2$
    \item $\triangle P_2 B C$
  \end{enumerate}
\end{itemize}

\paragraph{Orientation Handling}
All retriangulation rules respect the counterclockwise (CCW) orientation of the base triangle.
To maintain consistent mesh orientation and avoid topological artifacts, configurations that are mirror reflections of each other—such as those differing only by a left–right flip—are treated as distinct atomic cases in the lookup table. This explicit distinction ensures orientation consistency even under complex boundary embeddings.

\paragraph{Overview of Strategy (Canonical template choice).}
Geometrically, for some configurations (e.g., certain $(2,2)$ patterns), the same set of embedded intersection points and boundary-embedding constraints may admit more than one valid straight-line triangulation of the induced local patch.
SBMT removes this non-uniqueness \emph{by construction}: for each canonicalized lookup key $\kappa(K)$ (intersection class plus the ordered boundary intersection records), the lookup table stores \emph{one} fixed, canonical retriangulation template. Hence the mapping $\kappa(K)\mapsto P_K$ is single-valued and the runtime update performs no choice among multiple templates.

More aggressive local quality optimization (e.g., Delaunay-like flips restricted to the patch interior) could be incorporated as an \emph{extension} of the template family, but we intentionally omit it in the current prototype to preserve the purely local, rule-based semantics and the parallel-ready execution model. The preprocessing thresholds $(a,b,c)$ are chosen so that the canonical templates already provide sufficient element quality for the numerical tasks considered in this paper. The modular lookup-table design permits replacing the canonical template family by an enhanced family in future work, while retaining
the same canonicalization and consistency conventions.

\section*{Appendix H\\Parabolic Diffusion Behavior on SBMT Meshes}
\label{app:heat-diffusion}

We evaluate the numerical response of SBMT-generated meshes under transient heat diffusion governed by:
\[
\frac{\partial u}{\partial t} = \alpha \nabla^2 u, \quad u|_{\partial \Omega} = 0,
\]
with $\alpha = 500$ and a centered Gaussian initial pulse:
\[
u(x, y, 0) = 100 \exp\left( -\frac{(x - x_0)^2 + (y - y_0)^2}{2\sigma^2} \right).
\]
The domain is extracted from a medical cross-section (Figure~\ref{fig:sbmt-error-illustration}a), and meshing parameters match those in Section~4.4.
The Gmsh results in this appendix are generated with the same baseline
configuration as in Section~4.4: Gmsh~4.15.0,
the identical prescribed PSLG boundary, uniform characteristic length
$l_c=e=\sqrt{0.45}$, 2D Delaunay meshing (\texttt{Mesh.Algorithm=5}),
and fixed global size bounds with curvature- and boundary-extension-based
size adaptation disabled.

\begin{figure}[htbp]
    \centering
    \begin{tabular}{cc}
        \includegraphics[width=0.3\linewidth]{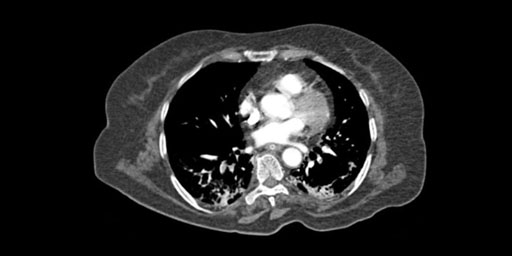} &
        \includegraphics[width=0.3\linewidth]{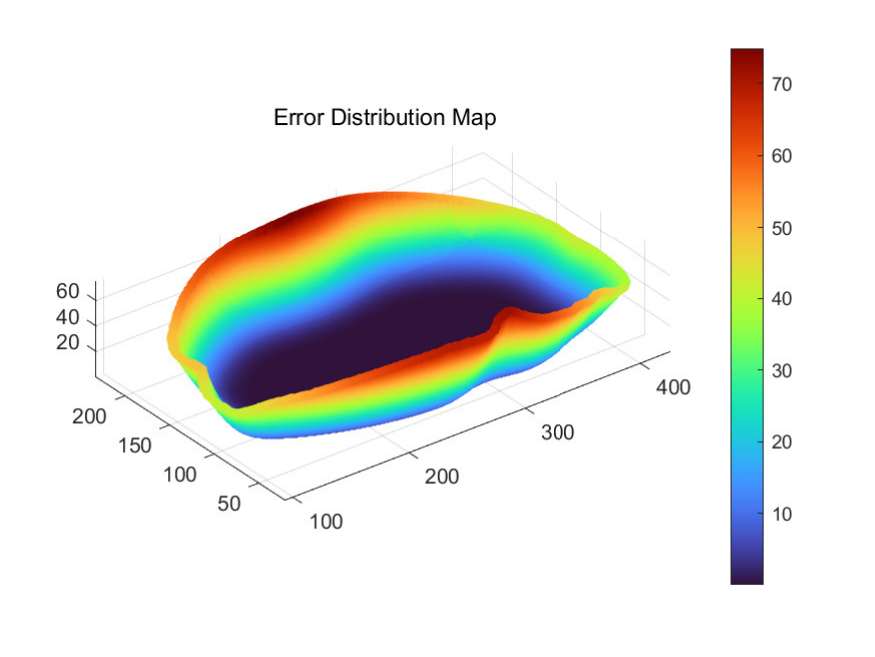} \\
        (a) Domain scan &
        (b) Pointwise error (SBMT vs. analytic)
    \end{tabular}
    \caption{Heat diffusion using SBMT mesh and its deviation from analytic reference.}
    \label{fig:sbmt-error-illustration}
\end{figure}

\paragraph{Observations.}
SBMT preserves boundary geometry during diffusion, maintaining coherent isothermal contours and sharper gradients near narrow lobes. Triangle yields smoother global diffusion but degrades at concave boundaries due to Steiner insertions.

\paragraph{Temperature Evolution.}
At $t = 0.5$ s, the peak temperature on SBMT mesh is $96.96^\circ$C, compared to $93.27^\circ$C on Triangle and Gmsh. SBMT exhibits slower heat spread due to higher geometric regularity and boundary refinement, while Triangle and Gmsh meshes, with smoother and more isotropic element distributions, promote slightly faster front propagation.

\paragraph{Error interpretation.}
Although SBMT yields a higher global $L_2$ error in this example
(e.g., $0.289$ at $t=0.3$ s, compared with $0.195$ for Triangle and
$0.200$ for Gmsh), the visual solution patterns suggest a different
trade-off: SBMT better preserves source-localized heat concentration and
boundary-aligned front geometry, with the most visible differences
occurring near thin or strongly curved boundary regions.
In this setting, a single global scalar error does not fully capture
differences in front location and boundary adherence, and therefore may
understate the structural advantage of SBMT.

\begin{figure}[htbp]
    \centering
    \begin{tabular}{ccc}
        \includegraphics[width=0.30\linewidth]{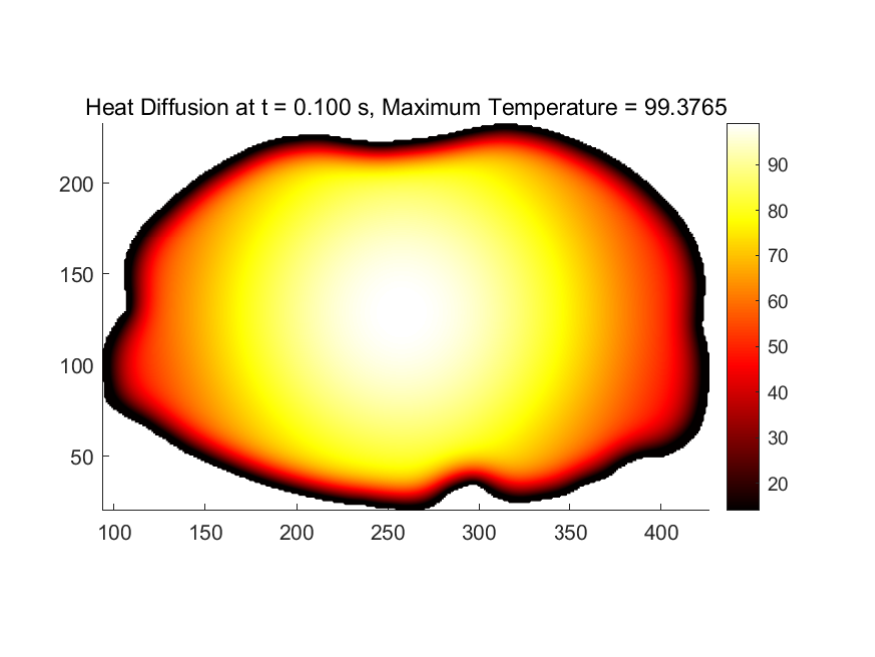} &
        \includegraphics[width=0.30\linewidth]{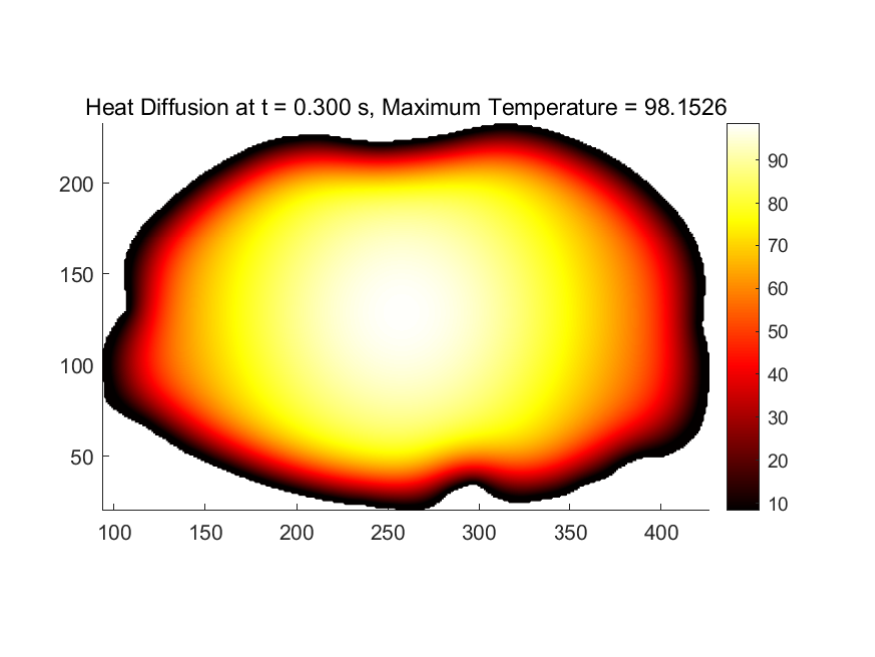} &
        \includegraphics[width=0.30\linewidth]{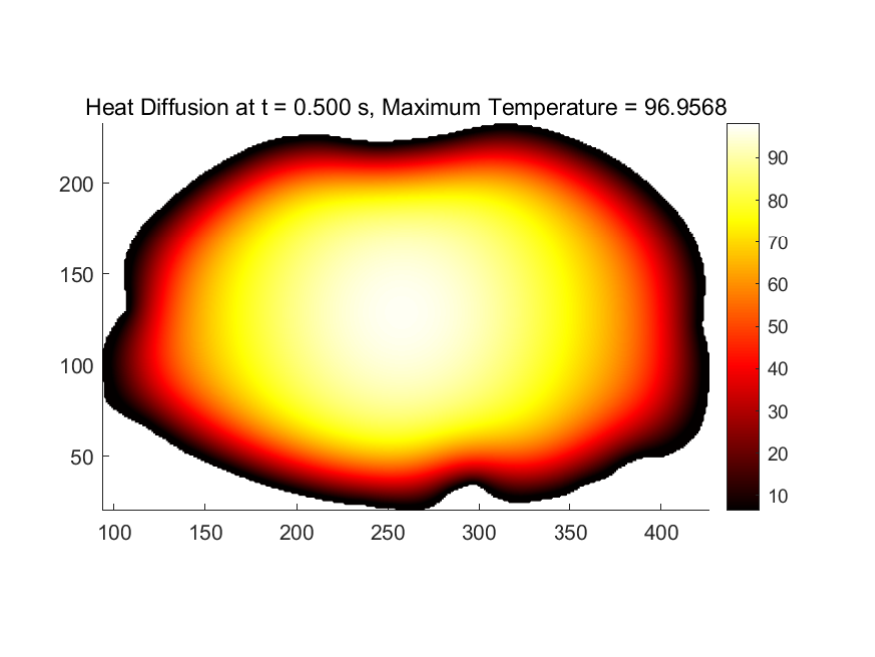} \\
        \includegraphics[width=0.30\linewidth]{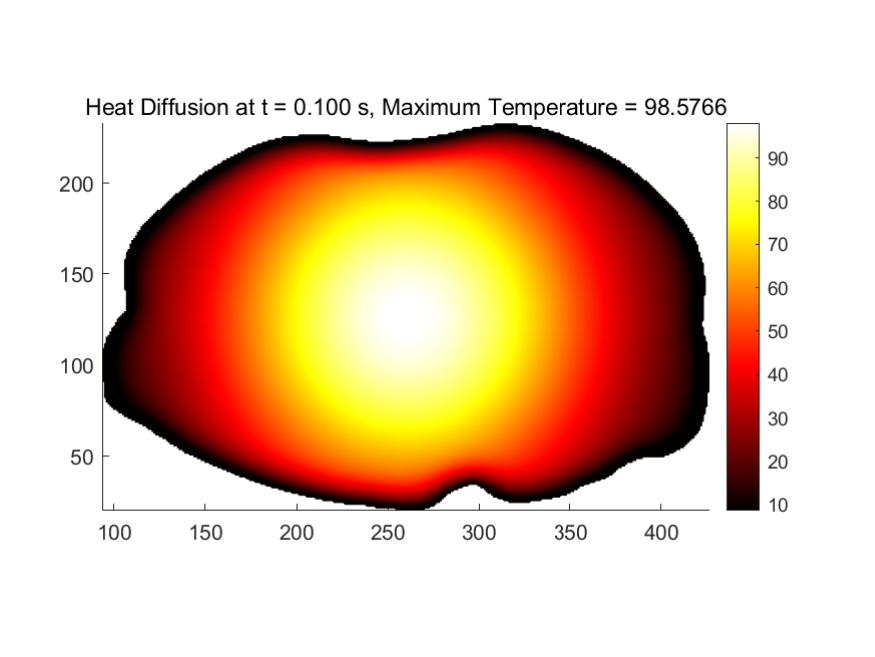} &
        \includegraphics[width=0.30\linewidth]{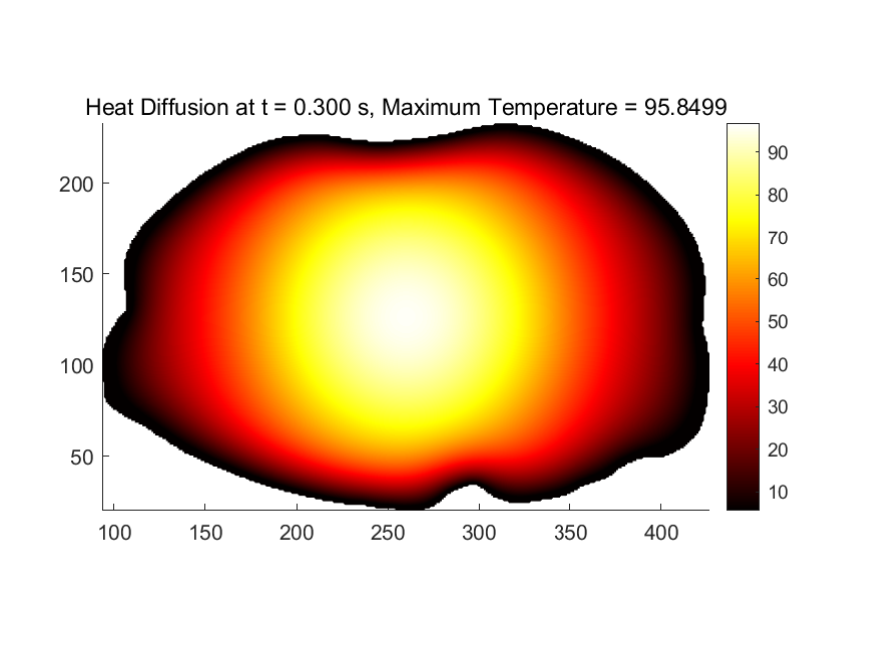} &
        \includegraphics[width=0.30\linewidth]{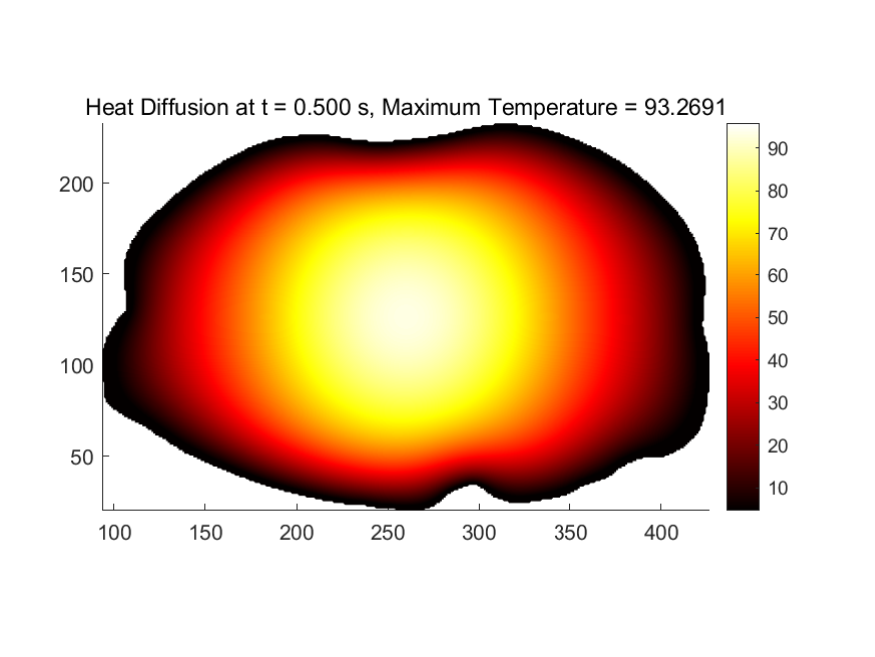} \\
        \includegraphics[width=0.30\linewidth]{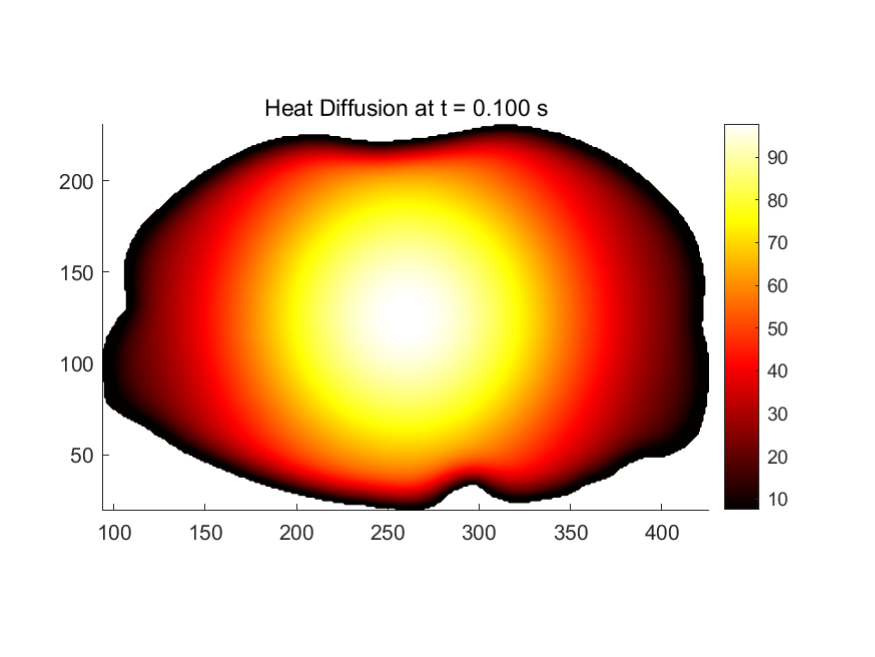} &
        \includegraphics[width=0.30\linewidth]{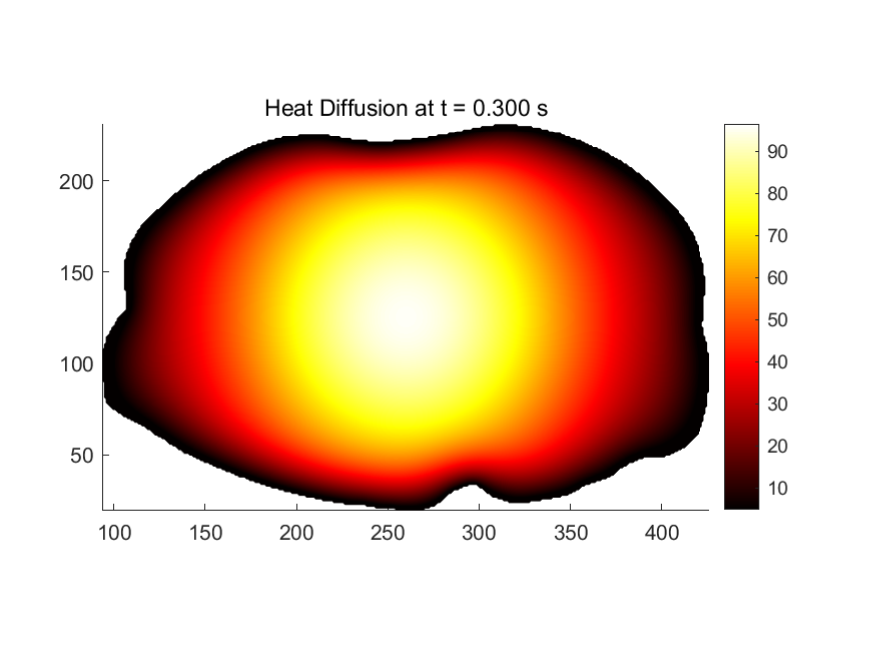} &
        \includegraphics[width=0.30\linewidth]{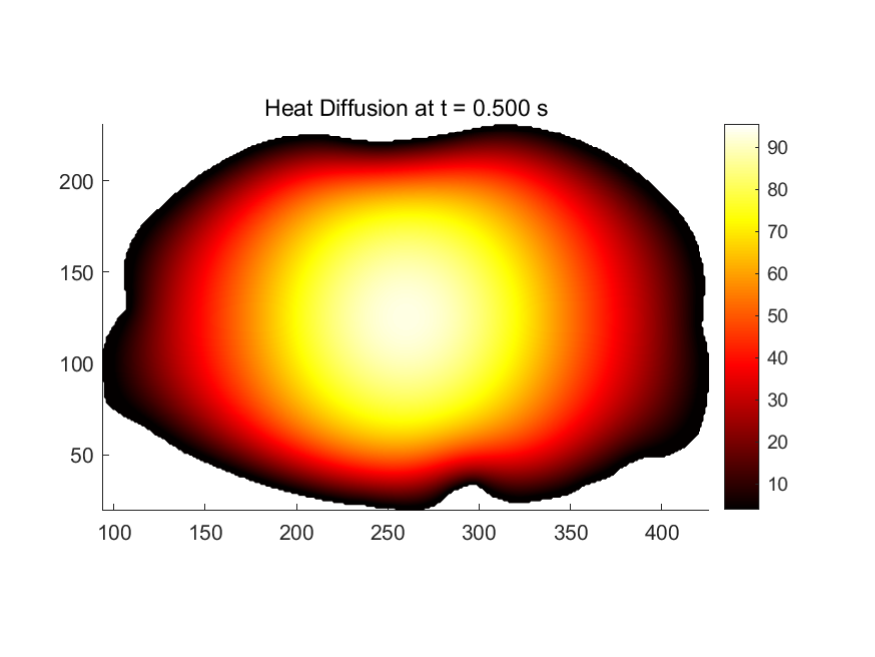} \\
    \end{tabular}
    \caption{Transient heat fields at $t = 0.1$, $0.3$, and $0.5$ s
    (left to right) on the same bitmap-derived domain. 
    Rows (top to bottom) show solutions on SBMT, Triangle (constrained
    Delaunay), and Gmsh meshes, respectively.
    SBMT yields boundary-aligned isotherms and sharper localization
    near thin lobes; Triangle and Gmsh produce smoother, slightly
    faster-spreading fronts.}
    \label{fig:heat-comparison}
\end{figure}

\paragraph{Conclusion.}
SBMT favors geometry-conforming numerical behavior, ideal for applications requiring fidelity to anatomical or structural features. Triangle and Gmsh may yield faster and slightly more accurate diffusion under smooth domains, but both lack SBMT’s bitmap-native topological consistency and local adaptivity along complex, thin anatomical structures.

\section*{Appendix I\\Sensitivity to Thresholds $(a,b,c)$}
\label{app:sensitivity_abc}

This section evaluates the sensitivity of SBMT to the three geometric preprocessing thresholds: snapping $a$, edge-elimination trigger $b$, and repulsion distance $c$ (cf.\ Section~3.3.4).
We report results on the same three benchmark domains used in the
comparative study (star, droplet, and handwritten ``Y''). All other settings are kept identical to the main experiments, including the fixed equilateral background resolution.

To address robustness over a broader admissible range, we now test five threshold triplets spanning both lower and higher settings around the nominal regime used in the paper, rather than only two nearby perturbations. The tested range explicitly includes the default threshold triplet $(0.26,0.125,0.183)$ used in the main text. All tested triplets satisfy the SBMT admissibility constraints, so the lookup-based remeshing logic remains unchanged and only the strength of the local preprocessing actions is varied.
For each setting, we measure:
(i) minimum interior angle $\theta_{\min}$,
(ii) minimum triangle area $A_{\min}$,
(iii) equilateral ratio (percentage of triangles classified as
equilateral by the edge-length test),
(iv) area variance $\mathrm{Var}(A)$, and
(v) aspect-ratio statistics
$\mathrm{AR}=\text{longest edge}/\text{shortest altitude}$. An equilateral element satisfies $\mathrm{AR}=2/\sqrt{3}\approx 1.1547$.

\begin{table}[htbp]
\centering
\scriptsize
\setlength{\tabcolsep}{6pt}
\renewcommand{\arraystretch}{1.08}
\caption{Sensitivity of SBMT mesh statistics over an expanded admissible
range of threshold triplets $(a,b,c)$.}
\label{tab:abc_sensitivity}
\begin{tabular}{llccccccc}
\toprule
\textbf{Domain} & $(a,b,c)$
& $\theta_{\min}$ ($^\circ$) & $A_{\min}$
& Eq.\ (\%) & $\mathrm{Var}(A)$
& $\mathrm{AR}_{95}$ & $\mathrm{AR}_{\max}$ & \#Tris \\
\midrule
Star
& $(0.22,0.10,0.15)$ & $9.140$ & $0.00867$ & $94.46$ & $0.000781$
& $1.29476$ & $9.20532$ & $70650$ \\
& $(0.24,0.115,0.17)$ & $9.847$ & $0.01162867$ & $94.52$ & $0.000712$
& $1.2793$ & $8.2771$ & $70475$ \\
& $(0.26,0.125,0.183)$ & $10.32$ & $0.01330030$ & $94.66$ & $0.000645$
& $1.24342$ & $7.79015$ & $70292$ \\
& $(0.27,0.125,0.185)$ & $10.383$ & $0.01347564$ & $94.72$ & $0.000623$
& $1.22759$ & $7.72205$ & $70222$ \\
& $(0.30,0.145,0.205)$ & $12.53$ & $0.01716856$ & $94.86$ & $0.000547$
& $1.19911$ & $7.11804$ & $69995$ \\
\midrule
Droplet
& $(0.22,0.10,0.15)$ & $10.079$ & $0.009722$ & $96.61$ & $0.000479$
& $1.1547$ & $9.20598$ & $84550$ \\
& $(0.24,0.115,0.17)$ & $11.18$ & $0.0118$ & $96.69$ & $0.000434$
& $1.1547$ & $8.27637$ & $84424$ \\
& $(0.26,0.125,0.183)$ & $11.629$ & $0.01312699$ & $96.77$ & $0.000386$
& $1.1547$ & $7.79155$ & $84271$ \\
& $(0.27,0.125,0.185)$ & $11.753$ & $0.01330139$ & $96.84$ & $0.000364$
& $1.1547$ & $7.72346$ & $84195$ \\
& $(0.30,0.145,0.205)$ & $12.968$ & $0.01597713$ & $96.96$ & $0.000313$
& $1.1547$ & $7.12546$ & $84016$ \\
\midrule
Y-shape
& $(0.22,0.10,0.15)$ & $9.494$ & $0.009431$ & $88.29$ & $0.001528$
& $2.12979$ & $9.2027$ & $35713$ \\
& $(0.24,0.115,0.17)$ & $10.299$ & $0.01204253$ & $88.44$ & $0.001381$
& $2.06321$ & $8.27742$ & $35517$ \\
& $(0.26,0.125,0.183)$ & $10.755$ & $0.01315683$ & $88.78$ & $0.001235$
& $2.00631$ & $7.7905$ & $35300$ \\
& $(0.27,0.125,0.185)$ & $10.821$ & $0.01332603$ & $88.92$ & $0.001195$
& $2.00047$ & $7.72239$ & $35242$ \\
& $(0.30,0.145,0.205)$ & $11.410$ & $0.0156$ & $89.38$ & $0.001014$
& $1.83636$ & $7.11733$ & $34967$ \\
\bottomrule
\end{tabular}
\end{table}

\paragraph{Observed trends.}
Across this expanded admissible range, the mesh statistics vary
smoothly and in a consistent direction.
As $(a,b,c)$ increase, the worst-case quality improves for all three domains: $\theta_{\min}$ increases substantially (from $9.14^\circ$ to $12.53^\circ$ on the star, from $10.08^\circ$ to $12.97^\circ$ on the droplet, and from $9.49^\circ$ to $11.41^\circ$ on the Y-shape), while $\mathrm{AR}_{\max}$ decreases from about $9.2$ to about $7.1$ in all three cases.
At the same time, the area variance decreases steadily, indicating a slightly more regular element-size distribution as the thresholds become more permissive.
The minimum triangle area also increases monotonically, reflecting the fact that stronger snapping/repulsion/elimination thresholds suppress the most degenerate boundary-band configurations more aggressively.

Equally importantly, the \emph{bulk} of the mesh remains stable across the tested range.
The equilateral ratio changes only mildly (about $0.4$ percentage point for the star, $0.35$ percentage point for the droplet, and about $1.1$ percentage points for the Y-shape), while the total triangle count varies by less than about $1\%$ for the star and droplet and by about $2\%$ for the Y-shape.
This indicates that, within the admissible regime, changing
$(a,b,c)$ primarily affects a thin boundary band and the most distorted local templates, rather than the structured interior scaffold that dominates the mesh.

Among the three domains, the Y-shaped example remains the most demanding, as reflected by its heavier aspect-ratio tail. This is expected: branching junctions and narrow channels necessarily activate a small subset of more elongated boundary-conforming elements. Nevertheless, even in this case the equilateral ratio stays close to $89\%$, confirming that the mesh bulk remains highly regular.

Overall, the expanded sensitivity study supports the same conclusion as the main paper, but now over a substantially broader perturbation range: SBMT is not hypersensitive to moderate threshold variation within its admissible parameter regime, and the default settings used in the paper are representative rather than specially tuned.

\section*{Appendix J\\Proof of Lemma~3.1}
\label{app:proof-segment-triangle-bound}

We prove statements (i)--(iii) in Lemma~3.1.

\paragraph{Preliminaries.}
Let $\gamma$ be a boundary chain represented as a polyline composed of
grid-aligned segments satisfying the angular and length constraints in
Section~3.2. Let $T$ be any base triangle in the
equilateral background mesh.

\paragraph{(ii) One edge-level event per edge (per segment).}
Fix an edge $e$ of $T$ and a boundary segment $S$.
Since both $e$ and $S$ are straight line segments, their intersection set
$e\cap S$ is one of the following: empty, a single point, or a colinear overlap
segment. Algorithm~1 assigns to the pair $(e,S)$ a
single, well-defined edge-level event in each case (crossing / endpoint hit /
colinear overlap with deterministic attribution priority), hence it produces at
most one edge-level intersection event per edge of $T$ for each segment. Note
that a geometric endpoint hit at a triangle vertex may be attributed to two
incident edges, but still yields \emph{no more than one} event per edge.

\paragraph{(i) At most two consecutive segments intersect a base triangle.}
We use the SBMT geometric protocol (Section~3.2), which
imposes (a) bounded turning along a chain (successive directions lie in a finite
set and turns are not sharper than the prescribed bound), and (b) a minimum
length separation / non-oscillation condition on the digital boundary so that a
chain cannot ``wiggle'' inside a single cell-scale neighborhood.

Because $T$ is convex, each individual segment can intersect $T$ only through a
connected portion (it cannot enter and exit $T$ multiple times). Suppose for
contradiction that three consecutive segments $S_{i-1},S_i,S_{i+1}$ along the
same chain all intersect $T$. Then the chain would have to realize two turns
while remaining within the local neighborhood of $T$; under the protocol's
turning and length constraints, such a double-turn would force either (a) an
inadmissible oscillation that revisits the same cell-scale neighborhood, or (b)
multiple distinct entry/exit episodes across $\partial T$, contradicting the
convexity-based connectivity of segment--triangle intersection and the
non-oscillation constraint. Therefore $T$ can be intersected by at most two
consecutive segments along any boundary chain.

\paragraph{(iii) Finiteness of admissible patterns and correspondence to templates.}
By (i), at most two segments contribute to the interaction with $T$. By (ii),
each contributing segment induces at most one edge-level event per edge, hence
each segment affects at most the three edges of $T$ with a bounded symbolic
signature. Therefore the set of possible edge-event patterns on $T$ is finite,
up to the symmetries of an equilateral triangle. Appendix~F enumerates all such
admissible patterns (up to symmetry) consistent with the protocol and
Algorithm~1, and Appendix~G assigns to each pattern a
static retriangulation template. This establishes (iii) and completes the proof.

\end{document}